\documentclass[sigconf,screen]{acmart}

\renewcommand\footnotetextcopyrightpermission[1]{} 
\usepackage[normalem]{ulem}

\hypersetup{draft}

\usepackage{algorithm}
\usepackage{algpseudocode}
\usepackage{listings}
\usepackage{mathpartir}
\usepackage{amsmath}
\usepackage{amsthm}
\newtheorem{theorem}{Theorem}
\newtheorem{lemma}{Lemma}
\newtheorem{definition}{Definition}

\usepackage{xspace}
\usepackage{listings}

\newcommand{\tool}{C11Tester\xspace}
\newcommand{\Tool}{C11Tester\xspace}

\newcommand{\papertitle}{\Tool: A Race Detector for C/C++ Atomics}
\newcommand{\mo}{\textit{mo}}
\newcommand{\rf}{\textit{rf}}
\newcommand{\hb}{\textit{hb}}
\newcommand{\sbo}{\textit{sb}}
\newcommand{\sco}{\textit{sc}}
\newcommand{\sw}{\textit{sw}\xspace}
\newcommand{\asw}{\textit{asw}\xspace}
\newcommand{\Mograph}{\hbox{\textit{mo-graph}}\xspace}
\newcommand{\code}[1]{{\tt #1}}

\newcommand{\br}[1]{\left( #1 \right) }
\newcommand{\TCV}{\mathbb{C}\xspace}
\newcommand{\CV}{\hbox{\textit{CV}}\xspace}
\newcommand{\Tid}{\hbox{\textit{Tid}}\xspace}
\newcommand{\Seq}{\hbox{\textit{Seq}}\xspace}
\newcommand{\Frel}[1]{\mathbb{F}^{\textit{rel} #1 }}
\newcommand{\Facq}[1]{\mathbb{F}^{\textit{acq} #1 }}
\newcommand{\RF}{\mathbb{RF}\xspace}
\newcommand{\System}{\hbox{$\Sigma$}\xspace}
\newcommand{\Epoch}{\textit{Epoch}\xspace}
\newcommand{\Val}{\hbox{\textit{Val}}\xspace}
\newcommand{\ThrState}{\hbox{\textit{ThrState}}\xspace}

\newcommand{\ALocs}{\hbox{\textit{ALocs}}\xspace}
\newcommand{\NALocs}{\hbox{\textit{NALocs}}\xspace}
\newcommand{\ALocInfo}{\hbox{\textit{ALocInfo}}\xspace}
\newcommand{\FenceInfo}{\hbox{\textit{FenceInfo}}\xspace}
\newcommand{\SysState}{\hbox{\textit{State}}\xspace}

\def\Yields{\Rightarrow}

\definecolor{listinggray}{gray}{0.9}
\definecolor{lbcolor}{rgb}{0.9,0.9,0.9}
\newcommand{\paper}[1]{}
\newcommand{\techreport}[1]{#1}
\newcommand{\Naively}{Naively\xspace}

\newcommand{\etal}{\hbox{\emph{et al.}}\xspace}
\newcommand{\ie}{\hbox{\emph{i.e.},}\xspace}
\newcommand{\eg}{\hbox{\emph{e.g.},}\xspace}
\usepackage[normalem]{ulem}
\usepackage{pdfpages}
\newcommand{\mycomment}[1]{}

\newcommand{\lasw}{\textit{lasw}}
\newcommand{\lsb}{\textit{lsb}}
\newcommand{\lsc}{\textit{lsc}}

\newcommand{\lift}{\hbox{\textit{lift}}}
\newcommand{\olift}{\overline{\textit{lift}}}
\newcommand{\myState}[1]{\textit{s}_{#1}}
\newcommand{\myTrans}[1]{\textit{t}_{#1}}
\newcommand{\trace}[1]{{\sigma}_{#1}}
\newcommand{\traces}{\hbox{\textit{traces}}}
\newcommand{\event}[1]{\textit{e}_{#1}}
\newcommand{\Exec}[1]{\textit{E}_{#1}}
\newcommand{\oExec}[1]{\overline{E}_{#1}}
\newcommand{\partialmo}[2]{\textit{S}_{#1}^{#2}}
\newcommand{\Consistent}{\hbox{\textit{Consistent}}}
\newcommand{\rConsistent}{\hbox{\textit{rConsistent}}}

\newcommand{\vcenterarrow}{\ensuremath{\Longrightarrow}}
\usepackage{array}
\newcolumntype{M}{>{\centering\arraybackslash}m{\dimexpr.44\linewidth-2\tabcolsep}}
\newcolumntype{C}{>{\centering\arraybackslash}m{\dimexpr.12\linewidth-2\tabcolsep}}

\definecolor{listinggray}{gray}{0.9}
\definecolor{lbcolor}{rgb}{0.9,0.9,0.9}
\begin{document}

\begin{abstract}
Writing correct concurrent code that uses atomics under the C/C++
memory model is extremely difficult.  We present \tool, a
race detector for the C/C++ memory model that can explore executions in a larger fragment
of the C/C++ memory model than previous race detector
tools.  Relative to previous work, \tool's larger fragment includes
behaviors that are exhibited by ARM processors.
\Tool uses a new constraint-based algorithm to implement modification order that is optimized
to allow \tool to make decisions in terms of application-visible behaviors.
We evaluate \tool on several benchmark applications, and compare \tool's performance to both
tsan11rec, the
state of the art tool that controls
scheduling for C/C++; and 
tsan11, the state of the art tool that does not control
scheduling.

\end{abstract}

\title{\papertitle{} \\ \textbf{Technical Report}}

\author{Weiyu Luo}
\affiliation{%
  \institution{University of California, Irvine}
  \city{Irvine, California}
  \country{USA}}
\email{weiyul7@uci.edu}

\author{Brian Demsky}
\affiliation{%
  \institution{University of California, Irvine}
  \city{Irvine, California}
  \country{USA}}
\email{bdemsky@uci.edu}

\date{}
\maketitle

\thispagestyle{empty}

\section{Introduction\label{sec:intro}}

The C/C++11 standards added a weak memory model with support for
low-level atomics operations~\cite{cpp11spec,c11spec} that allows
experts to craft efficient concurrent data structures that scale
better or provide stronger liveness guarantees than lock-based data
structures.  \techreport{The potential benefits of atomics can lure both experts
and novice developers to use them.} However, writing correct
concurrent code using these atomics operations is extremely difficult.

Simply executing concurrent code is not an effective approach to
testing.  Exposing concurrency bugs often requires 
executing a specific path that might only occur when the program
is heavily loaded during deployment, executed on a specific processor,
or compiled with a specific compiler.
Some prior work helps record and replay buggy executions~\cite{deloren}.
Debuggers like Symbiosis~\cite{symbiosis} and Cortex~\cite{cortex} focus on
sequential consistency and test programs by
modifying thread scheduling of given initial executions.
However, both the thread scheduling and relaxed behavior of C/C++
atomics are sources of nondeterminism in a C/C++ programs that use atomics.
Thus, it is necessary to develop tools
to help test for concurrency bugs.  We present the \Tool tool for testing
C/C++ programs that use atomics.

Figure~\ref{fig:overview} presents an overview of the \Tool system.
\Tool is implemented as a dynamically linked library together with an
LLVM compiler pass, which instruments atomic operations,
non-atomic accesses to shared memory locations, and fence
operations with function calls into the \tool dynamic library.  The C++
and pthread library functions are overridden by the \tool library---\tool
implements its own threading library using fibers to precisely control
the scheduling of each thread.  The \tool library implements a race detector and \tool reports any races or assertion violations that it discovers.

\begin{figure}[!htb]
\includegraphics[scale=0.35]{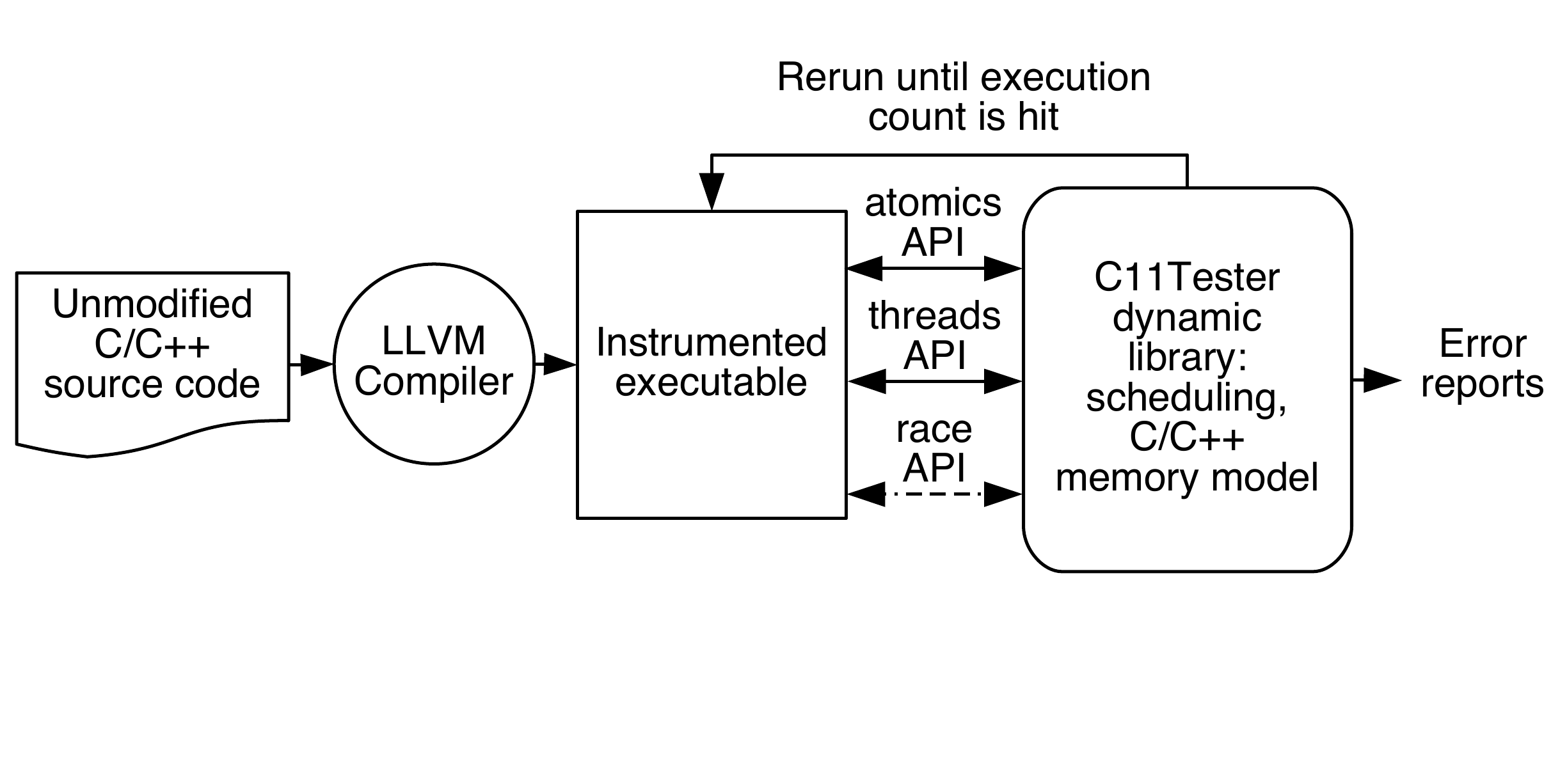}\vspace{-1.0cm}
\caption{\Tool system overview\label{fig:overview}}
\end{figure}

The C/C++ memory model defines the \emph{modification order} relation
to totally order all atomic stores to a memory location.  This
relation captures the notion of cache coherence.  The modification
order is not directly observable by the program execution --- it is
only observed indirectly through its effects on program visible
behaviors such as the values returned by loads.  Under the C/C++
memory model, modification order cannot be extended to be a total
order over all stores that is consistent with the happens-before relation.

This paper presents a new technique for scaling a constraint-based
treatment of the modification order relation to long
executions. \emph{This technique allows \tool to support a larger
  fragment of the C/C++ memory model than previous race detectors.}
In particular, this technique can handle the full range of
modification orders that are permitted by the C/C++ memory model.

Constraint-based modification order delays decisions
about the modification order until the decisions have observable effects on the program's behavior.
For example, when an algorithm decides which store a load will read from, \tool adds
the corresponding constraints to the modification order.  This
approach allows testing algorithms to focus on program visible
behaviors such as the value a load reads and does not require them to
eagerly decide the modification order.

Fibers provide a more efficient means to control thread schedules than
kernel threads.  However, C/C++ programs commonly make use of thread
local storage (TLS) and fibers do not directly support TLS.
This paper presents a new technique, thread context
borrowing, that allows fiber-based scheduling to support thread local
storage without incurring dependencies on TLS implementation details
that can vary across different library versions.

\subsection{Comparison to Prior Work on Testing C/C++11}
Prior work on data race detectors for C/C++11 such as
tsan11~\cite{tsan11} and tsan11rec~\cite{tsan11rec} require
$\hb \cup \rf \cup \mo \cup \sco$ be acyclic and thus miss
potentially bug-revealing executions that both are allowed by the C/C++
memory model and can be produced by mainstream hardware including ARM
processors.  We have found examples of bugs that \tool can detect 
but tsan11 and tsan11rec miss
due to the set of $\hb \cup \rf$ edges orders writes in the modification order.

\Tool's constraint-based approach to modification order supports a
larger fragment of the C/C++ memory model than tsan11 and
tsan11rec.  \Tool adds minor
constraints to the C/C++ memory model to forbid out-of-thin-air (OOTA) executions for relaxed atomics.
Furthermore, these constraints appear to incur minimal
overheads on existing ARM processors~\cite{oota} while x86 and PowerPC
processors already implement these constraints.

\subsection{Contributions}

This paper makes the following contributions:

\begin{itemize}
\item {\bf Scalable Concurrency Testing Tool:} It presents a 
  tool for the C/C++ memory model that can test full
  programs.
  
\item {\bf Supports a Larger Fragment of the C/C++ Memory Model:} It
  presents a tool that supports a larger fragment of the C/C++
  memory model than previous tools.
    
\item {\bf Constraint-Based Modification Order:} The modification
  order relation is not directly visible to the application, instead
  it constrains the behaviors of visible relations such as the reads-from
  relation.  \techreport{Eagerly selecting the modification order limits the
  choices of stores that a load can read from and thus limits the
  information available to algorithms.}  We develop a
  scalable constraint-based approach to modeling the modification
  order relation that allows algorithms to ignore the
  modification order relation and focus on program
  visible behaviors.

\item {\bf Support for Limiting Memory Usage:} The size of the C/C++
  execution graph and execution trace grows as the program executes and thus
  limits the length of executions that a testing tool can support.  Naively
  freeing portions of the graph can cause a tool to produce
  executions that are forbidden by the memory model.  We present
  techniques that can limit the memory usage of \tool while ensuring
  that \tool only produces executions that are allowed by the C/C++
  memory model. \mycomment{We present two approaches: (1) a conservative
  approach that does not change the set of executions that \tool can
  produce and (2) a more aggressive approach that allows users to
  decide how to trade off \tool's memory usage with shrinking the set of
  executions that \tool can generate.}
   
\item {\bf Fiber-based Support for Thread Local Storage:} Fibers are
  the most efficient way to control the scheduling of the application
  under test, but supporting thread local storage with fibers is
  problematic.  We develop a novel approach for borrowing the context
  of a kernel thread to support thread local storage.

\item {\bf Evaluation:} We evaluate \tool on several applications and
  compare against both tsan11 and tsan11rec.  We show that \tool can
  find bugs that tsan11 and tsan11rec miss.  We present a performance
  comparison with both tsan11 and tsan11rec.
\end{itemize}

\techreport{
\section{C/C++ Atomics\label{sec:background}}

In this section, we present general background on the C/C++ memory
model and then discuss the fragment of the C/C++ memory model that
\tool supports.  The C and C++ standards were extended in 2011 to
include a weak memory model that provides precise guarantees about the
behavior of both the compiler and the underlying processor.  The
standards divide memory locations into two types: normal types, which
are accessed using normal memory primitives; and atomic types, which
are accessed using atomic memory primitives.  The standards forbid
data races on normal memory types and allow arbitrary accesses to
atomic memory types.  Accesses to atomic memory types have an optional
\code{memory\_order} argument that explicitly specifies the ordering
constraints.  Any operation on an atomic object will have one of six
\textit{memory orders}, each of which falls into one or more of the
following categories.
Like all other tools for the C/C++ memory
model, compilers, and work on formalization to our knowledge, \Tool does not support
the consume memory order and thus we omit consume in our presentation.

\begin{description}
        \item[seq-cst:]
                \path{memory_order_seq_cst} --
                strongest memory ordering, there exists a total order of all operations with this memory ordering.  Loads that are seq\_cst either read from the last store in the seq\_cst order or from some store that is not part of seq\_cst total order.

        \item[release:]
                \path{memory_order_release},
                \path{memory_order_acq_rel}, and
                \path{memory_order_seq_cst} --
               when a load-acquire reads from a store-release, it establishes a happens-before relation between the store and the load.  Release sequences generalize this notion to allow intervening RMW operations to not break synchronization.

        \item[acquire:]
                \path{memory_order_acquire},
                \path{memory_order_acq_rel}, and
                \path{memory_order_seq_cst} --
                may form release/acquire synchronization.

        \item[relaxed:] \code{memory\_order\_relaxed} -- weakest
          memory ordering.  The only constraints for relaxed memory
          operations are a per-location total
          order, the modification order, that is equivalent to cache coherence.
\end{description}

The C/C++ memory model expresses program behavior in the form
of binary relations or orderings. We briefly summarize the relations:

\begin{itemize}
\item{\bf Sequenced-Before:}
The evaluation order within a program establishes an intra-thread
\textit{sequenced-before} (\textit{sb}) relation---a strict preorder of the atomic operations over the
execution of a single thread. 

\item{\bf Reads-From:}
The \textit{reads-from} ($\rf$) relation consists of store-load
pairs $(X, Y)$ such that $Y$ takes its value from 
$X$.  In the C/C++ memory model, this relation
is non-trivial, as a given load operation may read from one of many
potential stores in the execution.

\item{\bf Synchronizes-With:}
The \textit{synchronizes-with} (\textit{sw}) relation captures the
synchronization that occurs when certain atomic operations interact
across threads.

\item{\bf Happens-Before:}
In the absence of memory operations with the \code{consume} memory
ordering, the \textit{happens-before} ($\hb$) relation is the transitive
closure of the union of the \textit{sequenced-before} and the
\textit{synchronizes-with} relations.

\item{\bf Sequentially Consistent:}
All operations that declare the \path{memory_order_seq_cst} memory order
have a total ordering (\textit{sc}) in the program execution.

\item{\bf Modification Order:}
Each atomic object in a program has an associated \textit{modification
  order} ($\mo$)---a total order of all stores to that
object---which informally represents an ordering in which those stores
may be observed by the rest of the program.  
\end{itemize}

\subsection{Example} 
To explore some of the key concepts of the memory-ordering operations
provided by the C/C++ memory model, consider the example in
Figure~\ref{fig:ex}, assuming that two independent threads execute the
methods \code{threadA()} and \code{threadB()}.  This example uses the
C++ syntax for atomics; shared, concurrently-accessed variables are
given an \code{atomic} type, whose loads and stores are marked with an
explicit \code{memory\_order} governing their inter-thread ordering
and visibility properties.  In the example, the memory operations
are specified to have the \code{relaxed} memory ordering, which is the
weakest ordering in the C/C++ memory model and allows memory
operations to different locations to be reordered.

In this example, a few simple interleavings of \code{threadA()} and
\code{threadB()} show that we may see executions in which $\{\code{r1
  = r2 = 0}\}$, $\{\code{r1 = r2 = 1}\}$, or 
$\{\code{r1 = 0} \wedge \code{r2 = 1}\}$, but it is somewhat
counter-intuitive that we may also see $\{\code{r1 = 1} \wedge \code{r2 = 0}\}$,
in which the first load statement sees the second store
but the second load statement does not see the first store.
  While this latter behavior cannot occur under a
sequentially-consistent execution of this program, it is, in fact,
allowed by the \code{relaxed} memory ordering used in the example (and
achieved by compiler or processor reorderings).

Now, consider a modification of the same example, where the store and
load on variable \code{y} (Line~\ref{line:store-relaxed-example} and
Line~\ref{line:load-relaxed-example}) now use \code{memory\_order\_release} and
\code{memory\_order\_acquire}, respectively, so that when the
load-acquire reads from the store-release, they form a release/acquire
synchronization pair.  Then in any execution where \code{r1 = 1} and
thus the load-acquire statement (Line~\ref{line:load-relaxed-example}) reads from the
store-release statement (Line~\ref{line:store-relaxed-example}), the synchronization
between the store-release and the load-acquire forms an ordering
between \code{threadB()} and \code{threadA()}---particularly, that the
actions in \code{threadB()} after the \code{acquire} must observe the
effects of the actions in \code{threadA()} before the \code{release}.
In the terminology of the C/C++ memory model, we say that all actions
in \code{threadA()} sequenced before the \code{release} happen before
all actions in \code{threadB()} sequenced after the \code{acquire}.

So when \code{r1 = 1}, \code{threadB()} must see \code{r2 = 1}.  In
summary, this modified example allows only three of the four
previously-described behaviors: $\{\code{r1 = r2 = 0}\}$, $\{\code{r1
  = r2 = 1}\}$, and $\{\code{r1 = 0} \wedge \code{r2 = 1}\}$.

\begin{figure}[!ht]
\begin{lstlisting}[language=C++]
atomic<int> x(0), y(0);

void threadA() {
        x.store(1, memory_order_relaxed);
        y.store(1, memory_order_relaxed);/*@ \label{line:store-relaxed-example} @*/
}
void threadB() {
        int r1 = y.load(memory_order_relaxed);/*@ \label{line:load-relaxed-example} @*/
        int r2 = x.load(memory_order_relaxed);
        printf("r1 = %d\n", r1);
        printf("r2 = %d\n", r2);
}
\end{lstlisting}
\vspace{-.3cm}
\caption{\label{fig:ex}A Variant of Message Passing in C++}
\vspace{-.3cm}
\end{figure}

\subsection{\Tool's C/C++ Memory Model Fragment}
}

\paper{\section{C/C++ Memory Model Fragment}}
We next describe the fragment of the C/C++ memory model that \tool
supports.  Our memory model has the following changes based on
the formalization of Batty \etal~\cite{c11popl}:

{\bf 1) Use the C/C++20 release sequence definition: } Since the original C/C++11 memory model, the definition of
release sequences has been weakened~\cite{releasesequences}.  This
change is part of the C/C++20 standard~\cite{cpp-draft-n4849}.  \Tool uses the newly weakened definition.  The
new definition of release sequences does not allow
\code{memory\_order\_relaxed} stores by the thread that originally
performed the \code{memory\_order\_release} store that heads the release sequence
to appear in
the release sequence.

{\bf 2) Add $\hb \cup \sco \cup \rf$ is acyclic:} Supporting load buffering or out-of-thin-air executions is extremely
difficult and the existing approaches introduce high overheads in
dynamic tools~\cite{prescientmemory,oopsla2013,toplascdschecker}.  Thus, we
prohibit out-of-thin-air executions with a similar assumption made
by much work on the C/C++ memory model --- we add the constraint that the union of
happens-before, sequential consistency, and reads-from relations, \ie $\hb \cup \sco \cup \rf$,
is acyclic~\cite{vafeiadis2013relaxed}.\footnote{The C/C++11 memory model already requires that $\hb \cup \sco$ is acyclic.}
This feature of the C/C++ memory model is known
to be generally problematic and similar solutions have been proposed to
fix the C/C++ memory model~\cite{mspc14,N3786,N3710,oota}.

{\bf 3) Strengthen consume atomics to acquire:} No compilers support
the consume access mode.  Instead, all compilers strengthen consume
atomics to acquire.

\paper{We formalize the above changes in Section A.1
of our technical report~\cite{c11tester-arxiv}.}
\techreport{We formalize the above changes in Section~\ref{sec:restricted-model}
of the Appendix.}
Our fragment of the C/C++ memory model is larger than
that of tsan11 and tsan11rec~\cite{tsan11,tsan11rec}.
The tsan11 and tsan11rec tools add a very strong
restriction to the C/C++ memory model that requires that 
$\hb \cup \sco \cup \rf \cup \mo$ be acyclic.

\section{\Tool Overview}\label{sec:overview}

We present our algorithm in this section. In our presentation,
we adapt some terminology and symbols from stateless model
checking \cite{dpor}. We denote the initial state with $s_0$. We associate
every state transition $t$ taken by thread $p$ with the dynamic
operation that affected the transition. We use $\textit{enabled}(s)$ to denote the set of all
threads that are enabled in state $s$ (threads can be disabled
when waiting on a mutex, condition variable, or when completed). We
say that $\textit{next}(s, p)$ is the next transition in thread $p$ at 
state $s$.

\begin{figure}[!htb]
{\small
  \begin{algorithmic}[1]
    \Procedure{Explore}{}
      \State $s := s_0$
      \While{$\textit{enabled}(s)$ is not empty}
        \State Select $p$ from $\textit{enabled}(s)$
        \State $t := \textit{next}(s, p)$
        \State $\textit{behaviors}(t) := \{ \text{Initial behaviors} \}$
        \State Select a behavior $b$ from $\textit{behaviors}(t)$
        \State $s := \textit{Execute}(s, t, b)$
      \EndWhile
    \EndProcedure
  \end{algorithmic}
  \caption{\label{fig:testing-algorithm}Pseudocode for \tool's Algorithm}
  \vspace{-.3cm}
}
\end{figure}

Figure~\ref{fig:testing-algorithm} presents pseudocode for \tool's exploration
algorithm.
\Tool calls \textsc{Explore} multiple times---each time generates one program execution.
\paper{The thread schedule does not uniquely define the behavior of C/C++ atomics,
due to the weak behaviors of C/C++ atomics.}
\techreport{Recall from Section~\ref{sec:background} that the thread
schedule does not uniquely define the behavior of C/C++ atomics.}
Therefore, we split the exploration into two components: (1) selecting
the next thread to execute and (2) selecting the behavior of that
thread's next operation.  \Tool has a pluggable framework for testing 
algorithms---\tool generates a set of legal choices for the next thread and behavior, and
then the plugin selects the next thread and behavior.
The default plugin implements a random strategy.

\paragraph{Scheduling}

Thread scheduling decisions are made at each atomic operation, 
threading operation, or synchronization operation (such as locking a mutex).  
Every time a thread
finishes a visible operation, the next thread to execute is randomly
selected from the set of enabled threads.  However, when a thread
performs several consecutive stores with memory order release or
relaxed, the scheduler executes these stores consecutively without
interruption from other threads.  Executing these stores
consecutively does not limit the set of possible executions and
provides \tool with more stores to select from when deciding
which store a load should read from.
\techreport{This decision also reduces bias
in comparison to a purely randomized algorithm.

For example, in Figure ~\ref{fig:randomized}, under a purely randomized algorithm,
the probability that \code{r1 = 1} is much greater than that of \code{r1 = 2}, 
because in order for \code{r1 = 2}, the scheduler must schedule \code{threadA()} twice before
\code{threadB()} is scheduled.
However, under \tool's strategy, once \code{threadA} is scheduled to run,
both stores at line~\ref{line:bias-first-store} and line~\ref{line:bias-second-store}
will be performed consecutively.
So when the load is encountered, the \textit{may-read-from} set
(defined in the paragraphs below)
either only contains the initial store at line~\ref{line:bias-initial-store} or contains all
three stores.  Thus, \code{r1} is equally likely to read 1 or 2.

\begin{figure}[t]
\begin{lstlisting}[language=C++]
atomic<int> x(0); /*@ \label{line:bias-initial-store} @*/

void threadA() {
        x.store(1, memory_order_relaxed);/*@ \label{line:bias-first-store} @*/
        x.store(2, memory_order_relaxed);/*@ \label{line:bias-second-store} @*/
}
void threadB() {
        r1 = x.load(memory_order_relaxed);
}
\end{lstlisting}
\vspace{-.3cm}
\caption{\label{fig:randomized}Bias of a Purely Randomized Algorithm}
\end{figure}
}

\paragraph{Transition Behaviors}

The source of multiple behaviors for a given schedule arises from the
reads-from relation---in C/C++, loads can read from 
stores besides just the ``last'' store to an atomic object.

We use the concept of a \textit{may-read-from} set, which is
an overapproximation of the stores that a given atomic load may read from that just
considers constraints from the happens-before relation.
The \textit{may-read-from} set for a load $Y$ is constructed as:
\begin{align*}
\textit{may-read-from}(Y) &= \{X \in \textit{stores}(Y) \mid \neg (Y \stackrel{\hb}{\rightarrow} X) \wedge \\
&(\nexists Z \in \textit{stores}(Y) \text{ . } X \stackrel{\hb}{\rightarrow} Z \stackrel{\hb}{\rightarrow} Y) \}\text{,}
\end{align*}
where $\textit{stores}(Y)$ denotes the set of all stores to the same object
from which $Y$ reads. \Tool selects a store from the
\textit{may-read-from} set. \Tool then checks that establishing this
$\rf$ relation does not violate constraints imposed by the modification order, as described
in Section~\ref{sec:modification-order}.  If the given selection is
not allowed, \tool repeats the selection process.  \Tool delays the
modification order check until after a selection is made to optimize
for performance.

\section{Memory Model Support \label{sec:modification-order}}

In this section, we present how \tool efficiently supports key aspects
of the C/C++ memory model.  


CDSChecker~\cite{oopsla2013} initially introduced the technique of using a
constraint-based treatment of modification order to remove redundancy
from the search space it explores.  There are essentially two types of
constraints on the modification order: (1) that a store $s_A$ is 
modification ordered before a store $s_B$ and (2) that a store $s_A$
immediately precedes an RMW $r_B$ in the modification order.

CDSChecker models these constraints using a \emph{modification order
  graph}.  Two types of edges correspond to these two
types of constraints.  Edges only exist between two nodes if 
they both represent memory accesses to the same location.  There
is a cycle in the modification order graph if and only if the graph
corresponds to an unsatisfiable set of constraints.  Otherwise, a
topological sort of the graph (with the additional constraint that an RMW
node immediately follows the store that it reads from) yields a
modification order that is consistent with the observed program
behavior.  CDSChecker used depth first search to check for
cycles in the graph.  CDSChecker would add edges to the modification
order graph to determine whether a given reads-from edge was plausible
--- if the edge made the set of constraints unsatisfiable, CDSChecker
would rollback the changes that the edge made to the graph.

This approach works well for model checking where the graphs are small---the fundamental scalability limits of model checking ensure
that the executions always contain a very small number of stores.
\emph{This approach is infeasible when executions (and
  thus the modification order graphs) can contain millions of atomic
  stores, because the graph traversals become extremely expensive.}

\subsection{Modification Order Graph}

We next describe the modification order graph in more detail.  We
represent modification order ($\mo$) as a set of constraints, built as a constraint
graph, namely the modification order graph ($\Mograph$).  A node in the
$\Mograph$ represents a single store or RMW in the execution.
There are two types of edges in the graph.  An $\mo$ edge from
node $A$ to node $B$ represents the constraint $A
\stackrel{\mo}{\rightarrow} B$.  A \textit{rmw} edge from node $A$ to
node $B$ represents the constraint that $A$ must immediately precede $B$
or formally that: $A \stackrel{\mo}{\rightarrow} B$ and $\forall C.
C \neq A \wedge C\neq B \Rightarrow (A \stackrel{\mo}{\rightarrow} C \Rightarrow B \stackrel{\mo}{\rightarrow} C)
\wedge (C \stackrel{\mo}{\rightarrow} B \Rightarrow C
\stackrel{\mo}{\rightarrow} A)$.

\Tool must only ensure that there exists some $\mo$ that
satisfies the set of constraints, or equivalently an acyclic
$\Mograph$.  \Tool dynamically adds edges to $\Mograph$
when new $\rf$ and $\hb$ relations are formed.  
We briefly summarize the properties of $\mo$ as implications~\cite{oopsla2013} in Figure~\ref{fig:mo_implications}.
\Tool maintains a per-thread list of atomic memory accesses to each memory location.
Whenever a new atomic load or store is executed, \tool uses this list to evaluate the implications in Figure~\ref{fig:mo_implications} as well as additional implications for fences.  

\begin{figure}[!htbp]
        \centering
        \begin{tabular}{MCM}
                \multicolumn{3}{c}{\textsc{Read-Read Coherence}} \\
                \includegraphics[scale=0.4]{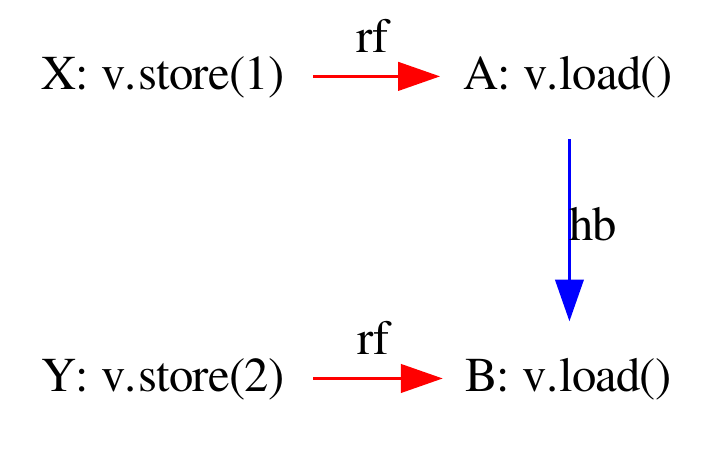} &
                \vcenterarrow &
                \includegraphics[scale=0.4]{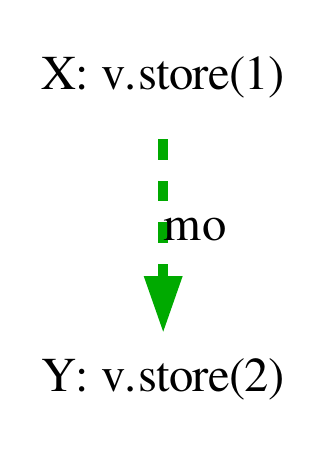} \vspace{3pt} \\
                \multicolumn{3}{c}{\textsc{Write-Read Coherence}} \\
                \includegraphics[scale=0.4]{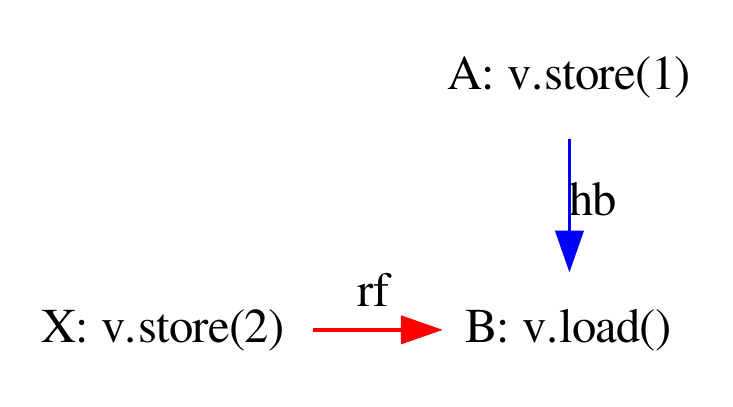} &
                \vcenterarrow &
                \includegraphics[scale=0.4]{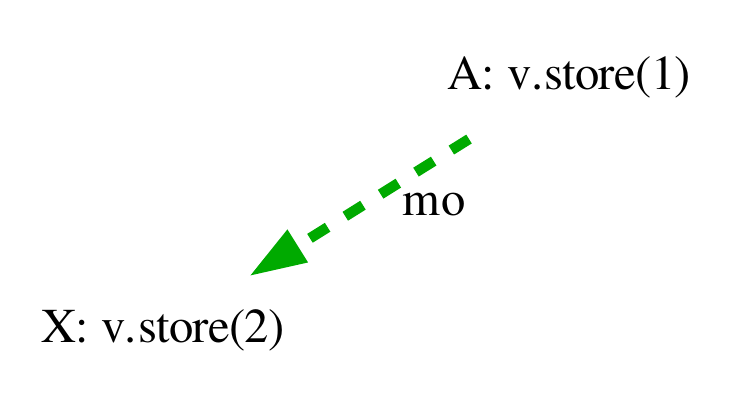} \vspace{3pt} \\
                \multicolumn{3}{c}{\textsc{Read-Write Coherence}} \\
                \includegraphics[scale=0.4]{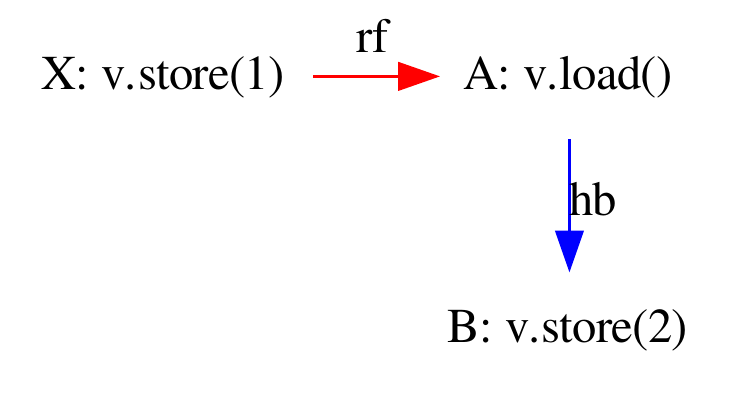} &
                \vcenterarrow &
                \includegraphics[scale=0.4]{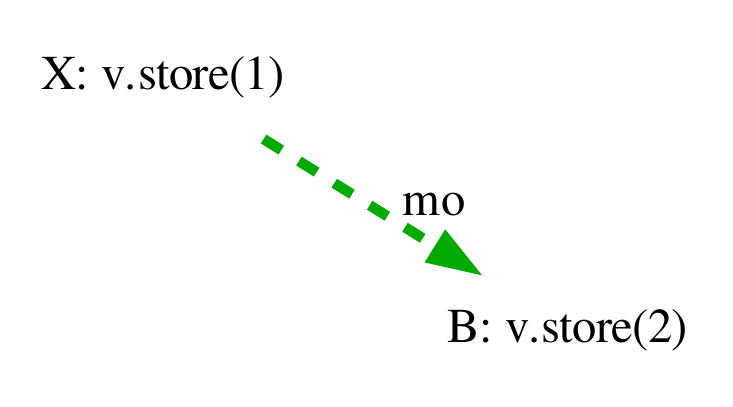} \vspace{3pt} \\
                \multicolumn{3}{c}{\textsc{Write-Write Coherence}} \\
                \includegraphics[scale=0.4]{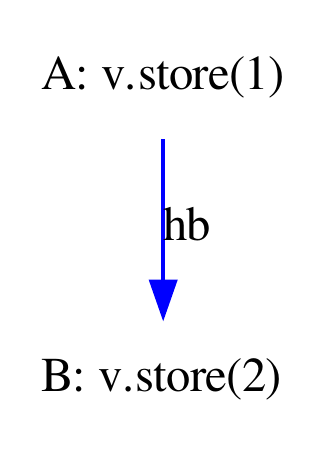} &
                \vcenterarrow &
                \includegraphics[scale=0.4]{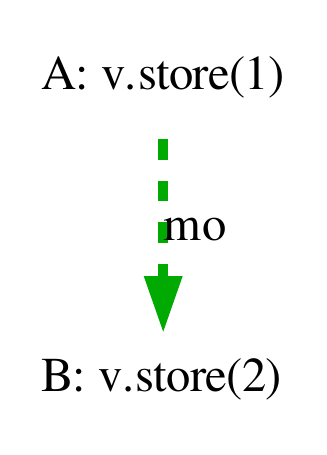} \vspace{3pt} \\
                \multicolumn{3}{c}{\textsc{Seq-cst / MO Consistency}} \\
                \includegraphics[scale=0.4]{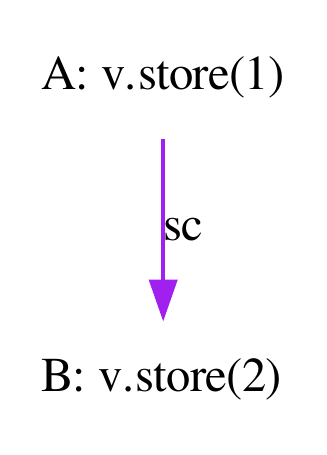} &
                \vcenterarrow &
                \includegraphics[scale=0.4]{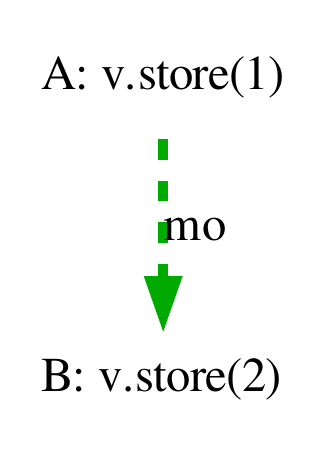} \\
                \multicolumn{3}{c}{\textsc{Seq-cst Write-Read Coherence}} \\
                \includegraphics[scale=0.4]{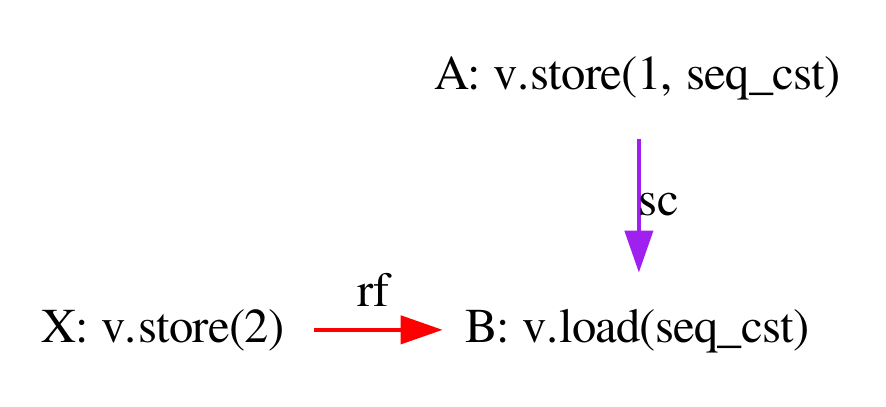} &
                \vcenterarrow &
                \includegraphics[scale=0.4]{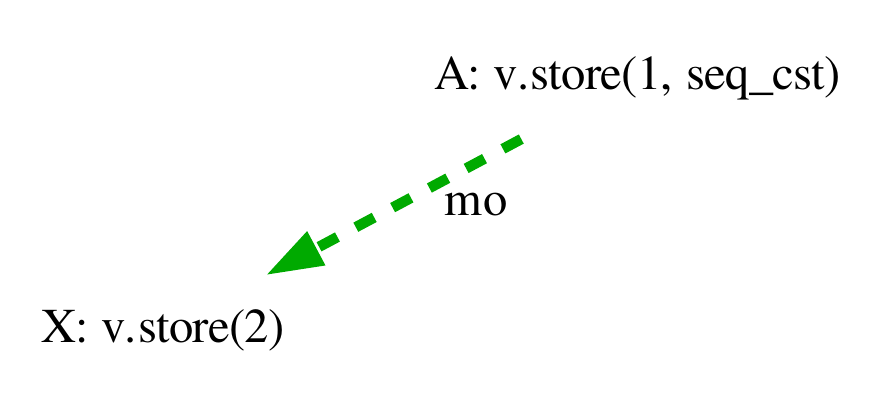} \\
                \multicolumn{3}{c}{\textsc{RMW / MO Consistency}} \\
                \includegraphics[scale=0.4]{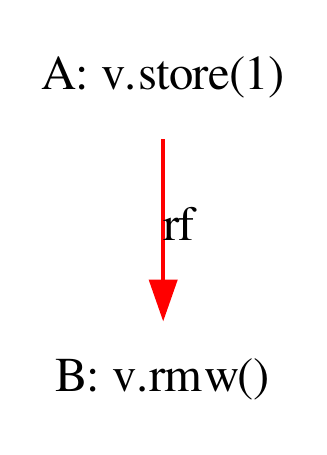} &
                \vcenterarrow &
                \includegraphics[scale=0.4]{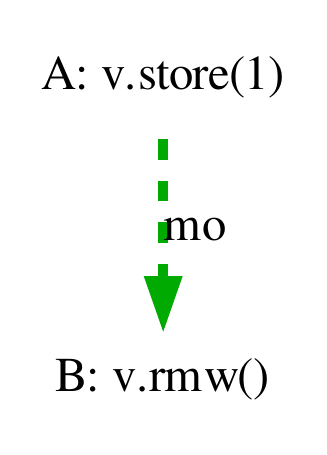} \\
                \multicolumn{3}{c}{\textsc{RMW Atomicity}} \\
                \includegraphics[scale=0.4]{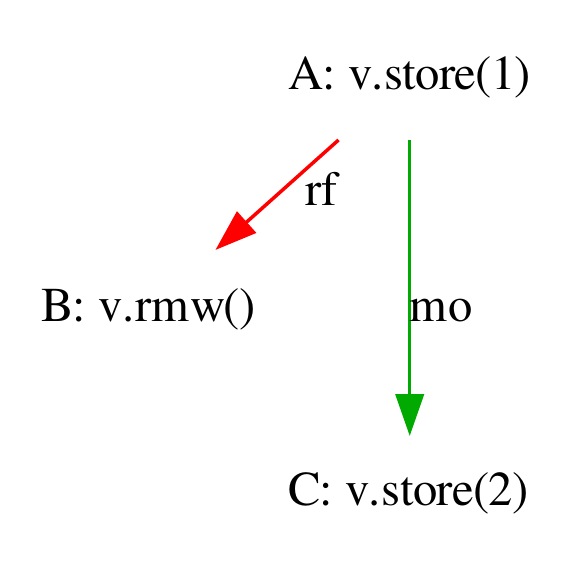} &
                \vcenterarrow &
                \includegraphics[scale=0.4]{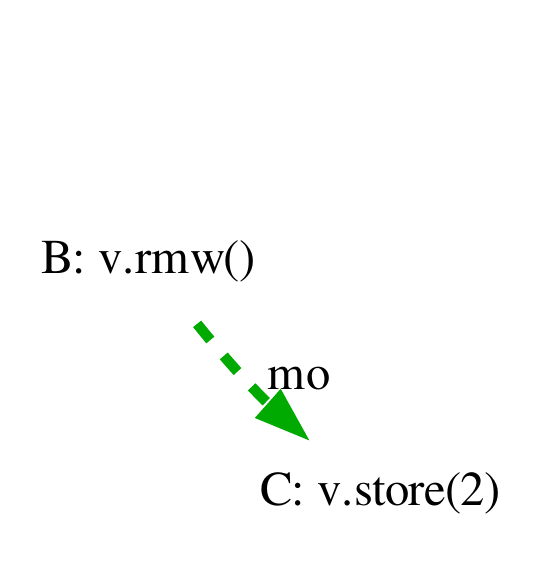} \\
        \end{tabular}
        \caption{\label{fig:mo_implications} Modification order implications. On the left side of each implication, $A$, $B$, $C$, $X$, and $Y$ must be distinct.}
\end{figure}

\subsection{Clock Vectors}\label{sec:mo-clock-vector}

Due to the high cost of graph traversals for large graphs, graph traversals are not a
feasible implementation approach for \tool.
We next describe how we adapt clock vectors~\cite{l-clocks} to efficiently compute
reachability in the $\Mograph$ and scale the constraint-based
modification order approach to large executions.  We associate a clock
vector with each node in the $\Mograph$.  \emph{It is important
  to note that our use of clock vectors in the $\Mograph$ is not to
  track the happens-before relation.  Instead we use clock vectors to
  efficiently compute reachability between nodes in the $\Mograph$.
  Thus, our $\Mograph$ clock vectors model a partial order that contains the
  current set of ordering constraints on the modification order.}

Each event $E$ \footnote{Events in each thread consist of atomic
  operations, thread creation and join, mutex lock and unlock, and
  other synchronization operations.} in \tool has a unique sequence
number $s_E$.  Sequence numbers
are a global counter of events across all threads, which is
incremented by one at each event.  We denote the thread that executed
$E$ as $t_{E}$.
Each node in the $\Mograph$ represents an atomic store.
The initial $\Mograph$ clock vector $\perp_{CV_A}$ associated with the node
representing an atomic store $A$, the union operator $\cup$,
and the comparison operator $\leq$ for $\Mograph$ 
clock vectors are defined as follows:
\begin{align*}
\perp_{CV_A}   &= \lambda t. \text{ if } t == t_A \text{ then } s_A \text{ else } 0,\\
CV_1 \cup CV_2 & \triangleq \lambda t.max(CV_1(t), CV_2(t)), \\
CV_1 \leq CV_2 & \triangleq \forall t. CV_1(t) \leq CV_2(t).
\end{align*}
\emph{Note that two $\Mograph$ clock vectors can only be compared
if their associated nodes represent atomic stores to the same memory location.}

The $\Mograph$ clock vectors are updated when new
$\mo$ relations are formed. 
For example, if $A \stackrel{\mo}{\rightarrow} B$ is a newly formed
$\mo$ relation, then the node $B$'s $\Mograph$ clock vector is merged with 
that of node $A$, \ie $CV_B := CV_A \cup CV_B$.
If $CV_B$ is updated by this merge, the change in $CV_B$ must be
propagated to all nodes reachable from $B$ using the union operator.

Figure~\ref{fig:mopseudo} presents pseudocode for updating the
modification order graph.  The \textsc{Merge} procedure merges the
$\Mograph$ clock vector of the \code{src} node into the \code{dst} node and returns true if the
\code{dst} $\Mograph$ clock vector changed.  The \textsc{AddEdge} procedure adds a new
modification order edge to the graph.  It first compares $\Mograph$ clock vectors to
check if the edge is redundant and if so drops the edge update. Recall
that RMW operations are ordered immediately after the stores that they
read from.  To implement this, \textsc{AddEdge} checks to see if the
\code{from} node has a \code{rmw} edge, and if so, follows the \code{rmw} edge.
\textsc{AddEdge} finally adds the relevant edge, and then propagates any
changes in the $\Mograph$ clock vectors.  The \textsc{AddRMWEdge} procedure has two
parameters, where the \code{rmw} node reads from the \code{from} node. 
It first adds an \code{rmw} edge and then migrates any outgoing edges from the source of the
edge to the \code{rmw} node.  Finally, it calls the \textsc{AddEdge}
procedure to add a normal modification order edge and to propagate
$\Mograph$ clock vector changes.

Figure~\ref{fig:mopseudo2} presents pseudocode for the helper method \textsc{AddEdges}
that adds a set of edges to the $\Mograph$.
The parameter \textit{set} is a set of atomic stores or RMWs,
and $S$ is an atomic store or RMW.
The \code{GetNode} method converts an atomic action to the corresponding node
in the $\Mograph$.  If such node does not exist yet, then the method will
create a new node in the $\Mograph$.

\begin{figure}[!htb]
{\small
  \begin{algorithmic}[1]
    \Procedure{Merge}{Node dst, Node src}
    \If{src.cv $\leq$ dst.cv}
    \State \textbf{return false}
    \EndIf
    \State dst.cv := dst.cv $\cup$ src.cv
    \State \textbf{return true}
    \EndProcedure
  \end{algorithmic}
  \begin{algorithmic}[1]
    \Procedure{AddEdge}{Node from, Node to}
    \State mustAddEdge := (from.rmw == to $\vee$ from.tid == to.tid)
  \If{from.cv $\leq$ to.cv $\wedge \neg$ mustAddEdge}
  \State \textbf{return}
  \EndIf
  \While{from.rmw $\neq$ null}
  \State next := from.rmw
  \If{next == to}
    \State \textbf{break}
    \EndIf
    \State from := next
    \EndWhile
    \State from.edges := from.edges $\cup$ to
    \If{\Call{Merge}{to, from}}	\label{addedge:merge}
    \State Q := \{ to \}
    \While{Q is not empty} \label{addedge:while-start}
    \State{node := remove item from Q}
    \For{each dst in node.edges}
    \If{\Call{Merge}{dst, node}}
    \State Q := Q $\cup$ dst
    \EndIf
    \EndFor
    \EndWhile \label{addedge:while-end}
    \EndIf \label{addedge:merge-end}
    \EndProcedure
\end{algorithmic}
  \begin{algorithmic}[1]
    \Procedure{AddRMWEdge}{Node from, Node rmw}
    \State from.rmw := rmw
    \For{each dst in from.edges} \label{addrmwedge:migrate-start}
    \If{dst $\neq$ rmw}
    \State rmw.edges := rmw.edges $\cup$ dst
    \EndIf
    \EndFor \label{addrmwedge:migrate-end}
    \State{from.edges := $\emptyset$}
    \State{\Call{AddEdge}{from, rmw}}
    \EndProcedure
  \end{algorithmic}
  \caption{Pseudocode for Updating $\Mograph$\label{fig:mopseudo}}

  \begin{algorithmic}[1]
    \Procedure{AddEdges}{\textit{set}, $S$}
    \State $n_S := \code{GetNode}(S)$
    \For{each $e$ in \textit{set}}
      \State $n_e := \code{GetNode}(e)$
      \State \Call{AddEdge}{$n_e$, $n_S$}
    \EndFor
    \EndProcedure
  \end{algorithmic}
  \caption{Helper method for adding a set of edges to the $\Mograph$\label{fig:mopseudo2}}
}
\end{figure}

Theorem~\ref{thm:main} guarantees the soundness of our use of $\Mograph$ clock vectors.
We present the theorem and its proof in Section~\ref{sec:mographtheorem}.
This theorem states that we can solely rely on $\Mograph$ clock vectors
to compute reachability between nodes in $\Mograph$.

\subsection{Eliminating Rollback in \textit{Mo-graph}}

Prior work on constraint-based modification order utilized rollback
when it was determined that a given reads-from relation was not
feasible \cite{oopsla2013,toplascdschecker}.
\Tool may also hit such infeasible executions because the \textit{may-read-from} set
defined in Section~\ref{sec:overview} is an overapproximation
of the set of stores that a load can read from.
To determine precisely whether a load can read from a store, a naive approach is
to add edges to the $\Mograph$ and then utilize rollback if adding these edges
introduces cycles in the $\Mograph$.
However, the addition of clock
vectors and clock vector propagation makes rollback much more expensive.  It is
thus critical that \tool avoids the need for rollback.  We now discuss how
\tool avoids rollback.

The $\Mograph$ is updated whenever a new atomic store, atomic
load, or atomic RMW is encountered.  Processing a new atomic store, atomic
load, or atomic RMW can potentially add multiple edges to the
$\Mograph$.  We next analyze each case to understand how to
avoid rollback:
\begin{itemize}
\item {\bf Atomic Store:}
Since an atomic load can only read from past stores, a newly created
store node in $\Mograph$ has no outgoing edges. By the properties of
$\mo$, only incoming edges from other nodes to this new node will be
created. Hence, a new store node cannot introduce any cycles.

\item {\bf Atomic Load:}
Consider a new atomic load $Y$ that reads from a store $X_0$.
Forming a new $\rf$ relation may only cause edges to be created
from other nodes to the node representing the store $X_0$.  We denote
this set of "other nodes" as $\textit{ReadPriorSet}(X_0)$ and compute
it using the \textsc{ReadPriorSet} procedure in Figure~\ref{alg:priorset}.
Lines~\ref{line:rpriorset-s1},~\ref{line:rpriorset-s2}, and~\ref{line:rpriorset-s3}
in the \textsc{ReadPriorSet} procedure consider statements 5, 4, and 6 in
Section 29.3 of the C++11 standard.
Line~\ref{line:rpriorset-s4} in the procedure considers write-read and
read-read coherences.  Therefore, the set returned by the
\textsc{ReadPriorSet} procedure captures the
set of stores from where new $\mo$ relations are to be formed
if the $\rf$ relation is established.

Before forming the $\rf$ relation, \tool checks
whether any node in $\textit{ReadPriorSet}(X_0)$ is reachable from $X_0$. If
so, then having load $Y$ read from store $X_0$ will introduce a cycle
in the $\Mograph$, so we discard $X_0$ and try another store.  While it
is possible for a cycle to contain two or more edges in the set of newly
created edges, this also implies that there is a cycle with one edge
(since all edges have the same destination).

\item {\bf Atomic RMWs:}
An atomic RMW is similar to both a load and store, but with the
constraint that it must be immediately modification ordered after the
store it reads from.  We implement this by moving modification order
edges from the store it reads from to the RMW.  Thus, the same checks
used by the load suffice to check for cycles for atomic RMWs.
\end{itemize}
Thus, \tool first computes a set of edges that reading from a given
store would add to the $\Mograph$.  Then for each edge, it checks the
$\Mograph$ clock vectors to see if the destination of the edge can
reach the source of the edge.  If none of the edges would create a
cycle, it adds all of the edges to the $\Mograph$ using the
\textsc{AddEdge} and \textsc{AddRMWEdge} procedures.

\numberwithin{equation}{section}
\setcounter{theorem}{0}

\section{Correctness of \textit{Mo-graph}\label{sec:mographtheorem}}

To prove the correctness of \textit{mo-graphs}, we first prove three Lemmas and then prove Theorem~\ref{thm:main}.
Lemma~\ref{thm:lemma1} and Lemma~\ref{thm:lemma2} characterize some important
properties of \textit{mo-graph} clock vectors.  Lemma~\ref{thm:lemma-backward} proves
one direction in Theorem~\ref{thm:main}.  \textit{Mo-graph} clock vectors are
simply referred to as clock vectors in the following context. 

\begin{lemma} \label{thm:lemma1}
Let $C_0 \stackrel{\mo}{\rightarrow} C_1 \stackrel{\mo}{\rightarrow} ... \stackrel{\mo}{\rightarrow} C_n$
be a path in a modification order graph $G$, such that
$CV_{C_0} \leq ... \leq CV_{C_n}$.
Then if any new edge $E$ is added to $G$ using procedures in Figure~\ref{fig:mopseudo},
it holds that 
\begin{align}
  CV_{C_0}' \leq ... \leq CV_{C_n}' \label{eq:lemma1-main}
\end{align}
for the updated clock vectors.
We define $CV_{C_i}' := CV_{C_i}$ if the values of $CV_{C_i}$ are not actually updated.
\end{lemma}

\begin{proof}
To simplify notation, we define $CV_i := CV_{C_i}$ for all $i \in \{0...,n\}$.
Let's first consider the case where no \code{rmw} edge is added, \ie the \textsc{AddRMWEdge}
procedure is not called.

By the definition of the union operator, each slot in clock vectors is monotonically increasing
when the \textsc{Merge} procedure is called. By the structure of procedure \textsc{AddEdge}'s
algorithm, a node $X$ is added to $Q$ if and only if this node's clock vector
is updated by the \textsc{Merge} procedure.

Let's assume that adding the new edge $E$ updates any of $CV_0, ..., CV_n$.
Otherwise, it is trivial.  Let $i$ be the smallest integer in $\{0, ..., n\}$
such that $CV_i$ is updated.  Then $CV_k' = CV_k$ for
all $k \in I := \{0, ..., i-1\}$, and we have
\begin{align}
  CV_0' \leq ... \leq CV_i'. \label{eq:lemma1-1}
\end{align}
If $i = 0$, then we take $I = \varnothing$. There are two cases.

\textbf{Case 1}: Suppose $CV_i' \leq CV_j$ for some $j \in \{i+1, ..., n\}$,
let $j_0$ be the smallest such integer.
Then $ CV_k' = CV_k$ for all $k \in \{j_0, ..., n\}$, as nodes $\{C_{j_0},...,C_n\}$ 
will not be added to $Q$ in the \textsc{AddEdge} procedure, and it holds trivially that
\begin{align}
  CV_{j_0}' \leq ... \leq CV_n'. \label{eq:lemma1-2}
\end{align}
By line~\ref{addedge:merge} to line~\ref{addedge:merge-end}
in the \textsc{AddEdge} procedure, we have
\begin{align}
  CV_k' = CV_k \cup CV_{k-1}', \label{eq:lemma1-3}
\end{align}
for all $k \in S := \{i+1, ..., j_0-1\}$.  If $j_0$ happens to be $i+1$,
then take $S = \varnothing$.
And we have for all $k \in S$, $CV_{k-1}' \leq CV_k'$. 
Then combining with inequality (\ref{eq:lemma1-1}), we have
\[ CV_0' \leq ... \leq CV_i \leq ... \leq CV_{j_0-1}'. \]
Together with inequality (\ref{eq:lemma1-2}), we only need to
show that $CV_{j_0-1}' \leq CV_{j_0}'$ to complete the proof.

If $j_0 = i+1$, then we are done, because by assumption
$CV_i' \leq CV_{j_0} = CV_{j_0}'$.
If $j_0 > i+1$, then $CV_i' \leq CV_{j_0}$ and $CV_{i+1} \leq CV_{j_0}$ imply
that $CV_{i+1}' = CV_{i+1} \cup CV_i' \leq CV_{j_0} = CV_{j_0}'$.
Based on equation (\ref{eq:lemma1-3}), we can deduce in a similar way that
$CV_{i+2}' \leq ... \leq CV_{j_0-1}' \leq CV_{j_0}'$.

\textbf{Case 2}:
Suppose $CV_i \nleq CV_j$ for all $j \in \{i+1, ..., n\}$.
Then by line~\ref{addedge:merge} to line~\ref{addedge:merge-end}
in the \textsc{AddEdge} procedure, all nodes $\{C_i, ..., C_n\}$ are added to
$Q$ in the \textsc{AddEdge} procedure, and
$ CV_k' = CV_k \cup CV_{k-1}'$
for all $k \in S := \{i+1, ..., n\}$.  This recursive formula guarantees that
for all $k \in S$, $CV_{k-1}' \leq CV_k'$.
Therefore, combining with inequality (\ref{eq:lemma1-1}), we have
$ CV_0' \leq ... \leq CV_n'$.

Now suppose the newly added edge $E$ is a \code{rmw} edge. 
If $E: X \xrightarrow{\textit{rmw}} C_i$ where $i \in \{0,...,n\}$
and $X$ is some node not in path $P$,
then the path $P$ remains unchanged and \Call{AddEdge}{$X$,$C_i$} is called.
Then the above proof shows that inequality (\ref{eq:lemma1-main}) holds. 
If $E:C_i \xrightarrow{\textit{rmw}} X$, then $C_i \stackrel{\mo}{\rightarrow} C_{i+1}$
is migrated to $X \stackrel{\mo}{\rightarrow} C_{i+1}$ by
line~\ref{addrmwedge:migrate-start} to line~\ref{addrmwedge:migrate-end}
in the \textsc{AddRMWEdge} procedure, and $C_i \stackrel{\mo}{\rightarrow} X$ is added.

If $X$ is not in path $P$, then path $P$ becomes
\[C_0 \stackrel{\mo}{\rightarrow} ... \stackrel{\mo}{\rightarrow} C_i \stackrel{\mo}{\rightarrow} X \stackrel{\mo}{\rightarrow} C_{i+1} \stackrel{\mo}{\rightarrow} ... \stackrel{\mo}{\rightarrow} C_n. \]
Since \Call{AddEdge}{$C_i$,$X$} is called, the same proof in the case without \code{rmw} edges
applies.  If $X$ is in path $P$, then $X$ can only be $C_{i+1}$
and the path $P$ remains unchanged.
Otherwise, a cycle is created and this execution is invalid. 
In any case, the same proof applies. 
\end{proof}

Let $\vec{x} = (x_1,x_2,...,x_n)$.
We define the projection function $U_i$ that extracts the $i^{\textit{th}}$
position of $\vec{x}$ as
$U_i(\vec{x}) = x_i,$ where we assume $i \leq n$.

\begin{lemma} \label{thm:lemma2}
Let $A$ be a store with sequence number $s_A$ performed
by thread $i$ in an acyclic modification order graph $G$.
Then $U_i(CV_A) = U_i(\perp_{CV_A}) = s_A$ throughout each execution that terminates. 
\end{lemma}

\begin{proof}
We will prove by contradiction.
Let $S = \{A_1, A_2, ...\}$ be the sequence of stores performed by thread $i$
with sequence numbers $\{s_1, s_2, ...\}$, respectively.
Suppose that there is a point of time in a terminating execution such that the first store $A_n$ 
in the sequence with $U_i(CV_{A_n}) > s_n$ appears.
Sequence numbers are strictly increasing and
by the \textsc{Merge} procedure, $U_i(CV_{A_n}) \in \{s_{n+1},s_{n+2},...,\}$.
Let $U_i(CV_{A_n}) = s_N$ for some $N > n$.

For $U_i(CV_{A_n})$ to increase to $s_N$ from $s_n$,
$CV_{A_n}$ must be merged with the clock vector of some node $X$ (\ie some store $X$)
in $G$ such that $U_i(CV_X) = s_N$.  Such $X$ is modification ordered before $A_n$.

If $X$ is performed by thread $i$, then $X$ has to be the store
$A_N$, because $U_i(CV_{A_j})$ is unique for all stores $A_j$
in the sequence $S$ other than $A_n$.
Then $\perp_{CV_X} \geq \perp_{CV_{A_n}}$.  By the definition of initial values
of clock vectors and sequence numbers, $X$ happens after
and is modification ordered after $A_n$.
However, $X$ is also modification ordered before $A_n$, and we have a cycle in $G$.
This is a contradiction.

If $X$ is not performed by thread $i$, then $U_i(\perp_{CV_X}) = 0$.
For $U_i(CV_X)$ to be $s_N$,
$X$ must be modification ordered after by some store
$Y$ in $G$ such that $U_i(CV_Y) = s_N$.
If $Y$ is done by thread $i$, then the same argument in the last paragraph
leads to a contradiction; otherwise, by repeating the same argument as in
this paragraph finitely many times (there are only a finite number
of stores in such a terminating execution), we would eventually deduce that
$X$ is modification ordered after some store by thread $i$. Hence, we would have
a cycle in $G$, a contradiction. 

\end{proof}

\begin{lemma} \label{thm:lemma-backward}
Let $A$ and $B$ be two nodes that write to the same location in
an acyclic modification order graph $G$.  If $B$ is reachable from $A$ in $G$,
then $CV_A \leq CV_B$.
\end{lemma}

\begin{proof}
Suppose that $B$ is reachable from $A$ in $G$.
Let $A \stackrel{\mo}{\rightarrow} C_1 \stackrel{\mo}{\rightarrow} ... \stackrel{\mo}{\rightarrow} C_{n-1} \stackrel{\mo}{\rightarrow} B$ be the shortest path $P$ from $A$ to $B$ in graph $G$.
To simplify notation, $X \stackrel{\mo}{\rightarrow} Y$ is abbreviated as
$X \rightarrow Y$ in the following.
As the \textsc{AddRMWEdge} procedure calls the \textsc{AddEdge} procedure
to create an \textit{mo} edge, we can assume that all the \textit{mo} edges
in $P$ are created by directly calling \textsc{AddEdge}. 

\textbf{Base Case 1}: Suppose the path $P$ has length 1, \ie $A$ immediately precedes $B$.
Then when the edge $A \rightarrow B$ was formed by calling \Call{AddEdge}{$A$,$B$},
$CV_B$ was merged with $CV_A$ in line~\ref{addedge:merge}
of the \textsc{AddEdge} procedure.  In other words,
$CV_B = CV_B \cup CV_A \geq CV_A.$

\textbf{Base Case 2}: Suppose the path $P$ has length 2, \ie $A \rightarrow C_1 \rightarrow B$. There are two cases:

(a) If $A \rightarrow C_1$ was formed first, then $CV_A \leq CV_{C_1}$.
When $C_1 \rightarrow B$ was formed, $CV_B$ was merged with $CV_{C_1}$ and $CV_{C_1} \leq CV_B$.
According to Lemma~\ref{thm:lemma1}, adding the edge $C_1 \rightarrow B$ or any edge
not in path $P$ (if any such edges were formed before $C_1 \rightarrow B$ was formed)
to $G$ would not break the inequality $CV_A \leq CV_{C_1}$.
It follows that $CV_A \leq CV_{C_1} \leq CV_B$.

(b) If $C_1 \rightarrow B$ was formed first, then 
$CV_{C_1} \leq CV_B$.  Based on Lemma~\ref{thm:lemma1}, this inequality remains true
when $A \rightarrow C_1$ was formed.  Therefore $CV_A \leq CV_{C_1} \leq CV_B$.

\textbf{Inductive Step}: 
Suppose that $B$ being reachable from $A$ implies that $CV_A \leq CV_B$ for
all paths with length $k$ or less, for some $k > 2$.
We want to prove that the same holds for paths with length $k + 1$.
Let $P$ be a path from $A$ to $B$  with length $k+1$,
\[ P: A = C_0 \rightarrow C_1 \rightarrow ... \rightarrow C_{k} \rightarrow C_{k+1} = B. \]
We denote $A$ as $C_0$ and $B$ as $C_{k+1}$ in the following.

Let $E: C_i \rightarrow C_{i+1}$ be the last edge formed in path $P$, where $i \in \{0,...,k\}$.
Then before edge $E$ was formed, the inductive hypothesis
implies that $CV_{C_0} \leq ... \leq CV_{C_i}$
and $CV_{C_{i+1}} \leq ... \leq CV_{C_{k+1}}$,
because both $C_0 \rightarrow ... \rightarrow C_i$ and
$C_{i+1} \rightarrow ... \rightarrow C_{k+1}$ have length $k$ or less.
Lemma~\ref{thm:lemma1} guarantees that 
\begin{align*}
  CV_{C_0} &\leq ... \leq CV_{C_i}, \\
  CV_{C_{i+1}} &\leq ... \leq CV_{C_{k+1}}
\end{align*}
remain true if any edge not in path $P$ was added to $G$
as well as the moment when $E$ was formed.
Therefore when the edge $E$ was formed,
we have $CV_{C_i} \leq CV_{C_{i+1}}$, and
\[ CV_A = CV_{C_0} \leq ... \leq CV_{C_{k+1}} = CV_B. \]
\end{proof}

\begin{theorem}  \label{thm:main}
Let $A$ and $B$ be two nodes that write to the same location in
an acyclic modification order graph $G$ for a terminating execution.
Then $CV_A \leq CV_B$ iff $B$ is reachable from $A$ in $G$.
\end{theorem}

\begin{proof}
Lemma~\ref{thm:lemma-backward} proves the backward direction, so
we only need to prove the forward direction. 
Suppose that $CV_A \leq CV_B$. Let's first consider the situation
where the graph $G$ contain no \code{rmw} edges.

\textbf{Case 1}:
$A$ and $B$ are two stores performed by the same thread with thread id $i$.
Then it is either $A$ happens before $B$ or $B$ happens before $A$. 
If $A$ happens before $B$, then $A$ precedes $B$ in the modification order
because $A$ and $B$ are performed by the same thread.
Hence $B$ is reachable from $A$ in $G$.  We want to show that the other
case is impossible.

If $B$ happens before $A$ and hence precedes $A$ in the modification order,
then $A$ is reachable from $B$. 
By Lemma~\ref{thm:lemma-backward}, $A$ being reachable from $B$ implies that
$CV_B \leq CV_A$.  Since $CV_A \leq CV_B$ by assumption, we deduce that
$CV_A = CV_B$.  This is impossible according to Lemma~\ref{thm:lemma2},
because each store has a unique sequence number and
$U_i(CV_A) = s_A \neq s_B = U_i(CV_B)$, implying that $CV_A \neq CV_B$.

\textbf{Case 2}: $A$ and $B$ are two stores done by different threads.
Suppose that $A$ is performed by thread $i$.
Let $CV_A = (...,s_A,...)$ and $CV_B = (...,t_b,...)$ where both $s_A$
and $t_b$ are in the $i^{\textit{th}}$ position.
By assumption, we have $0 < s_A \leq t_b$.

Since $B$ is not performed by thread $i$, we have $U_i(\perp_{CV_B}) = 0$.
We can apply the same argument similar to the second, third and fourth paragraphs
in the proof of Lemma~\ref{thm:lemma2} and deduce that
$B$ is modification ordered after $A$ or some store sequenced after $A$.
Since modification order is consistent with \textit{sequenced-before} relation,
if follows that $B$ is reachable from $A$ in graph $G$.

Now, consider the case where \code{rmw} edges are present.
Adding a \code{rmw} edge from a node $S$ to a node $R$
first transfers to $R$ all outgoing \textit{mo} edges coming from $S$ and 
then adds a normal \textit{mo} edge from $S$ to $R$.  So, any updates in $CV_S$
are propagated to all nodes that are reachable from $S$.
Therefore, the above argument still applies.
\end{proof}

\section{Operational Model}

We present our operational model with respect to the tsan11~\cite{tsan11} core language
described by the grammar in Figure~\ref{fig:grammar}.
A program is a sequence of statements.  \code{LocNA} and \code{LocA}
denote disjoint sets of non-atomic and atomic memory locations.
A statement can be one of these forms: an \code{if} statement, assigning the
result of an expression to a non-atomic location, forking a new
thread, joining a thread via its thread handle, and atomic statements.
The symbol $\epsilon$ denotes an empty statement.  Atomic statements
denoted by \code{StmtA} include atomic loads, store, \code{RMW}s, and
fences. An \code{RMW} takes a functor, \code{F}, to implement \code{RMW}
operations, such as \code{atomic\_fetch\_add}.  We omit loops 
for simplicity and leave the details of an expression unspecified.
We omit lock and unlock operations because they can be 
implemented with atomic statements.

\begin{figure}[!htb]
  \begin{lstlisting}[numbers=none, escapeinside={(*}{*)}]
Prog  ::= Stmt ; (*$\epsilon$*)
Stmt  ::= Stmt ; Stmt
        | if (LocNA) {Stmt} else {Stmt}
        | LocNA := Expr
        | LocNA = Fork(Prog)
        | Join(LocNA)
        | StmtA
        | (*$\epsilon$*)
StmtA ::= LocNA = Load(LocA, MO)
        | Store(LocNA, LocA, MO)
        | RMW(LocA, MO, F)
        | Fence(MO)
MO    ::= relaxed | release | acquire | rel_acq
        | seq_cst
Expr  ::= <literal> | LocNA | Expr op Expr
  \end{lstlisting}
  \caption{Syntax for our core language \label{fig:grammar}}
\end{figure}

\begin{figure}[!htb]
{\footnotesize
\textbf{States:}
\begin{align*}
  \Tid & \triangleq \mathbb{Z} & \Seq & \triangleq \mathbb{Z} & \TCV &: \Tid \rightarrow \CV \\
  \Frel{} &: \Tid \rightarrow \CV & \RF  &: \Seq \rightarrow \CV & \Facq{} &: \Tid \rightarrow \CV
\end{align*}
\indent [RELEASE STORE]
\begin{mathpar}
  \inferrule* 
  { \RF' = \RF[ s := \TCV_t] }
  {\br{\TCV, \RF, \Frel{}, \Facq{}} \Yields^{\textit{store}_\textit{rel}(s, t)} \br{\TCV, \RF', \Frel{}, \Facq{}}}
\end{mathpar}

[RELAXED STORE]
\begin{mathpar}
  \inferrule*
  { \RF' = \RF[ s := \Frel{}_t] }
  {\br{\TCV, \RF, \Frel{}, \Facq{}} \Yields^{\textit{store}_\textit{rlx}(s, t)} \br{\TCV, \RF', \Frel{}, \Facq{}}}
\end{mathpar}


[RELEASE RMW]
\begin{mathpar}
  \inferrule*
  { \RF' = \RF[ s := \TCV_t \cup \RF_{s'}] }
  {\br{\TCV, \RF, \Frel{}, \Facq{}} \Yields^{\textit{rmw}_\textit{rel}(s, t), \rf(s', t')}  \br{\TCV, \RF', \Frel{}, \Facq{}}}
\end{mathpar}

[RELAXED RMW]
\begin{mathpar}
  \inferrule*
  { \RF' = \RF[ s := \Frel{}_t \cup \RF_{s'}] }
  {\br{\TCV, \RF, \Frel{}, \Facq{}} \Yields^{\textit{rmw}_\textit{rlx}(s, t), \rf(s', t')}  \br{\TCV, \RF', \Frel{}, \Facq{}}}
\end{mathpar}

\indent [ACQUIRE LOAD]
\begin{mathpar}
  \inferrule*
  { \TCV' = \TCV[ t := \TCV_t \cup \RF_{s'} ] }
  {\br{\TCV, \RF, \Frel{}, \Facq{}} \Yields^{\textit{load}_\textit{acq}(s, t), \rf(s', t')} \br{\TCV', \RF, \Frel{}, \Facq{}}}
\end{mathpar}

[RELAXED LOAD]
\begin{mathpar}
  \inferrule*
  { \Facq{'} = \TCV[ t := \Facq{}_t \cup \RF_{s'} ] }
  {\br{\TCV, \RF, \Frel{}, \Facq{}} \Yields^{\textit{load}_\textit{rlx}(s, t), \rf(s', t')} \br{\TCV, \RF, \Frel{}, \Facq{'}}}
\end{mathpar}

\indent [RELEASE FENCE]
\begin{mathpar}
  \inferrule*
  { \Frel{'} = \Frel{}[ t := \TCV_t ] }
  {\br{\TCV, \RF, \Frel{}, \Facq{}} \Yields^{\textit{fence}_\textit{rel}(t)} \br{\TCV', \RF, \Frel{'}, \Facq{}}}
\end{mathpar}

[ACQUIRE FENCE]
\begin{mathpar}
  \inferrule*
  { \TCV' = \TCV[ t := \TCV_t \cup \Facq{}_t ] }
  {\br{\TCV, \RF, \Frel{}, \Facq{}} \Yields^{\textit{fence}_\textit{acq}(t)} \br{\TCV', \RF, \Frel{}, \Facq{}}}
\end{mathpar}
\caption{Semantics for tracking happens-before clock vectors for atomic loads,
stores, RMWs, and fences.  An RMW also triggers a load rule initially. \label{fig:sem-hbcv}}
}
\end{figure}

\subsection{Happens-Before Clock Vectors}
We next discuss the various happens-before clock vectors that \tool
uses to implement happens-before relations.
Figure~\ref{fig:sem-hbcv} presents our algorithm for updating clock vectors
used to track happens-before relations for atomic loads, stores, RMWs,
and fences.
The union operator $\cup$ between clock vectors is defined the same way as in
Section~\ref{sec:mo-clock-vector}.

For each thread $t$, the algorithm maintains the thread's own clock vector
$\TCV_t$, and release- and acquire-fence clock vectors $\Frel{}_t$ and $\Facq{}_t$. 
The algorithm also records a reads-from clock vector $\RF_s$ for each atomic store and RMW. 
Recall that the sequence number is a global counter of events across all threads,
and thus uniquely identifies an event. 
We use $\TCV, \Frel{}, \Facq{}$ and $\RF$ to denote these clock vectors
across all threads, and atomic stores and RMWs. 
The rules for atomic loads and RMWs also require the stores or RMWs that are
read from to be specified, which are denoted as $\rf$.

\paragraph{Release Sequences}

The 2011 standard used a complicated definition of
release sequences that allowed the possibility of relaxed writes
blocking release sequences~\cite{tsan11}.  The 2020
standard simplifies and weakens the definition of release
sequences.  In a recently approved draft~\cite{cpp-draft-n4849}, a store-release heads
a release sequence and an RMW is part of the release sequence if and only if it
reads from a store or RMW that is part of the release sequence.  A
load-acquire synchronizes with a store-release $S$ if the load reads
from a store or RMW in the release sequence headed by $S$.

We first discuss \tool's treatment of release sequences in the absence
of fences.  \Tool uses two clock vectors for store/RMW operations:
both the current thread clock vector $\TCV_t$ and a second
\emph{reads-from} clock vector $\RF_S$ that tracks the happens-before
relation for all release sequences that the RMW/store $S$ is part of.  For
a normal store release, these two clock vectors are the same.  When a
relaxed or release RMW $A$ reads from another store $B$, \tool
computes the RMW's reads-from clock vector $\RF_A$ as the union of: (1)
the store $B$'s reads-from clock vector $\RF_B$ and (2) the
RMW $A$'s current thread clock vector $\TCV_{t_A}$ if $A$ is a release.
When a load-acquire $A$ reads from a store-release or RMW, \tool computes
the load-acquire's new thread clock vector as the union of: (1) the
load-acquire's current thread clock vector $\TCV_{t_A}$ and (2) the store
release/RMW's reads-from clock vector.

\paragraph{Fences}
The C/C++ memory model also contains fences.  Fences can have one of
four different memory orders: acquire, release, acq\_rel, and
seq\_cst.  Release fences effectively make later relaxed stores into
store-releases, but the happens-before relation is established at the
fence-release.  \Tool maintains a release fence clock vector $\Frel{}_t$
for each thread and uses this clock vector when
computing the clock vector for release sequences.  Acquire fences
effectively make previous relaxed loads into load-acquires, but the
happens-before relation starts at the fence.  When a relaxed load
reads from a release sequence, \tool updates the per-thread
acquire-fence clock vector $\Facq{}_t$.  When \tool processes an acquire
fence, it uses $\Facq{}_t$ to
update the thread's clock vector $\TCV_t$.  Seq\_cst fences
constrain the interactions between sequentially consistent atomics and
non-sequentially consistent atomics.  The behavior of seq\_cst fences
can be represented as rules for generating modification order
constraints~\cite{c11popl}.  \Tool maintains a list of all seq\_cst
fences for each thread so that \tool can quickly locate the relevant
fence instructions.  It then generates the relevant modification order
edges to implement the fence semantics.

\subsection{Formal Operational Model}

Figure~\ref{fig:opstate} formalizes the operational state of a program. 
The state of system $\SysState$
consists of the list of $\ThrState$, the mapping $\ALocs$ from memory locations
to atomic information, the mapping $\NALocs$ from memory locations to values stored
at non-atomic locations, the mapping $\FenceInfo$, and
the $\Mograph$ described in Section~\ref{sec:modification-order}.
$\ALocInfo$ records the list of atomic loads, stores, and RMWs
performed at a given atomic location.
$\FenceInfo$ records the list of fences performed by each thread.
\textit{Prog} is a program described by the grammar
in Figure~\ref{fig:grammar}. 
The initial state of the system has empty mappings $\ALocs$ and $\NALocs$,
and $\FenceInfo$, only one thread representing the main function, and an empty
$\Mograph$. 

\begin{figure}[!htb]
{\footnotesize
  \begin{align*}
  \Tid & \triangleq \mathbb{Z} \quad \Epoch \triangleq \mathbb{Z} \quad
  \Val \triangleq \mathbb{Z}   \quad \Seq \triangleq \mathbb{Z} \\
  \CV & \triangleq \Tid \rightarrow \Epoch \\
  \ThrState & \triangleq (t: \Tid) \times (\TCV: \CV) \times (\mathbb{F}^{ \{\textit{rel}, \textit{acq}\} }: \CV) \times (\RF: \Seq \rightarrow \CV)\\ & \times (P: \textit{Prog}) \\
  \textit{StoreElem} & \triangleq (t: \Tid) \times (s: \Seq) \times (a: \textit{LocA}) \times (\mo: \textit{MemoryOrder}) \\ & \quad \times (v: \Val) \\
  \textit{LoadElem} & \triangleq (t: \Tid) \times (s: \Seq) \times (a: \textit{LocA}) \times (\mo: \textit{MemoryOrder}) \\ & \quad \times (\rf: \textit{StoreElem}) \\
  \textit{RMWElem} & \triangleq (t: \Tid) \times (s: \Seq) \times (a: \textit{LocA}) \times (\mo: \textit{MemoryOrder}) \\ & \quad \times (\rf: \textit{StoreElem or RMWElem}) \times (v: \Val) \\
  \textit{FenceElem} & \triangleq (t: \Tid) \times (s: \Seq) \times (\mo: \textit{MemoryOrder}) \\
  \ALocInfo & \triangleq (\textit{StoreElem or LoadElem or RMWElem}) \text{ list} \\
  \FenceInfo & \triangleq \Tid \rightarrow \textit{FenceElem} \text{ list} \\
  \ALocs  & \triangleq \textit{LocA} \rightarrow \ALocInfo \\
  \NALocs & \triangleq \textit{LocNA} \rightarrow \Val \\
  \SysState  & \triangleq \ThrState \text{ list} \times \ALocs \times \NALocs \\ & \quad \times \FenceInfo \times (M: \Mograph)
  \end{align*}
  \caption{Operational State \label{fig:opstate}}
  \vspace{-.3cm}
}
\end{figure}

\subsection{Operational Semantics}

Figures~\ref{fig:atomic-statement} to~\ref{alg:priorset} present
state transitions and related algorithms for our operational model.
A system under evaluation is a triple of the form ($\System$, \textit{ss}, $T$),
where $\System$ represents the state of the system $\SysState$,
\textit{ss} is the program being executed,
and $T$ represents $\ThrState$ of the thread currently running the program.
The current thread only updates its own state $T$ when the program \textit{ss} executes,
which causes the copy of $T$ in $\System$ to become outdated.
However, the updated $T$ will replace the old copy in $\System$
when the thread switching function $\delta$ is called at the end of each
atomic statement. 
The $\Mograph$ is a data structure in $\SysState$ and represented as $\System.M$.
The $\Mograph$ has methods \textsc{Merge}, \textsc{AddEdge},
\textsc{AddRMWEdge}, and \textsc{AddEdges} described in Figure~\ref{fig:mopseudo}
and Figure~\ref{fig:mopseudo2}.

Figure~\ref{fig:atomic-statement} shows semantics for atomic statements.
Every time an atomic statement is encountered, a corresponding \textit{LoadElem},
\textit{StoreElem}, \textit{RMWElem}, or \textit{FenceElem} is created
with the sequence number auto-assigned.  The process of assigning sequence numbers
are omitted in Figure~\ref{fig:atomic-statement}.
Function calls [LOAD], [STORE], [RMW], and [FENCE] invokes the corresponding
inference rules for updating clock vectors described in Figure~\ref{fig:sem-hbcv}
based on the type of atomic statements and the memory orders. 
Atomic statements with \code{seq\_cst} or \code{acq\_rel} memory
orderings invoke both acquire and release clock vector rules if they apply. 
[LOAD], [STORE], [RMW], and [FENCE] take the current state of the system,
the current atomic element, and the state of the
current thread as arguments, pass necessary input into the
inference rules for updating clock vectors, and finally return the updated
state of the current thread.

For atomic loads and RMWs, the store that is read from is randomly selected
from the \textit{may-read-from} set computed using the algorithm \textsc{BuildMayReadFrom}
presented in Figure~\ref{alg:may-read-from},
and the store must satisfy the constraint that the second return value of
\textsc{ReadPriorSet} is true, \ie having the load reading from the selected store
does not create a cycle in the $\Mograph$.
The atomic RMW rule first triggers an atomic load rule, and the store/RMW $S$
that is read from is recorded in the $\rf$ field of the \textit{RMWElem}.
Then, the $\Mograph$ is updated using the procedure \textsc{AddRMWEdge},
and the atomic RMW rules is finally finished by invoking an atomic store rule.
Both atomic load and atomic store rules call the helper method \textsc{AddEdges}
in Figure~\ref{fig:mopseudo2} to add edges to the $\Mograph$.

Figure~\ref{alg:priorset} presents the procedures \textsc{ReadPriorSet}
and \textsc{WritePriorSet} which compute the set of atomic actions
($\Mograph$ nodes) from where new $\mo$ edges will be formed.

We use the following helper functions in Figure~\ref{alg:may-read-from} and Figure~\ref{alg:priorset}:
\begin{itemize}
  \item $\textit{last\_sc\_fence}(t)$ returns the last \code{seq\_cst} fence in thread $t$;
  \item $\textit{last\_sc\_store}(a, S)$ returns the last \code{seq\_cst} store performed at location $a$ and is different from $S$;
  \item $\textit{sc\_fences}(t)$ returns the list of \code{seq\_cst} fences performed by thread $t$;
  \item $\textit{sc\_stores}(t, a)$ returns the list of \code{seq\_cst} stores and RMWs performed by thread $t$ at location $a$;
  \item $\textit{stores}(t, a)$ returns the list of stores and RMWs performed by thread $t$ at location $a$;
  \item $\textit{loads\_stores}(t, a)$ returns the list of loads, stores, and RMWs performed by thread $t$ at location $a$;
  \item $\textit{last}(\text{list})$ returns the element with the largest sequence number in the list, excluding null elements;
  \item $\textit{get\_write}(A)$ returns $A$ if $A$ is an atomic store or RMW and returns $A.\rf$ if $A$ is an atomic load.
\end{itemize}
All the above functions return null if the result does not exist. 

\begin{figure}[!htb]
{\footnotesize
  \textbf{[ATOMIC LOAD]}
  \begin{mathpar}
    \inferrule
    { \br{\System, T} \rightarrow_{\textit{load}} \br{\System, T'} \\
      L.t = T'.t \\ L.a = a \\ L.\mo = \mo \\ S \in \Call{BuildMayReadFrom}{L} \\\\
      L.\rf = S \\
      (\textit{pset}, \textit{ret}) = \Call{ReadPriorSet}{L, S} \\
      \textit{ret} == True \\
      T'' = [\text{LOAD}](\System, L, T') \\\\
      \System' = \System[M := \System.M.\Call{AddEdges}{\textit{pset}, S}] \\
      \System'' = \System'[\NALocs := \System'.\NALocs[l := S.v]] \\
      \System''' = \System''[\ALocs := \System''.\ALocs(a).\code{pushback}(L)]
    }
    {\br{\System, l = \code{Load}(a, \mo); \textit{ss}, T} \Yields \br{\System'', \delta; \textit{ss}, T''}}
  \end{mathpar}

  \textbf{[ATOMIC STORE]}
  \begin{mathpar}
    \inferrule*
    { \br{\System, T} \rightarrow_{\textit{store}} \br{\System', T} \\
      S.t = T.t \\ S.a = a \\ S.\mo = \mo \\ S.v = \System'.\NALocs(l) \\
      \textit{pset} = \Call{WritePriorSet}{S} \\\\
      T' = [\text{STORE}](\System', S, T) \\\\
      \System'' = \System'[M := \System'.M.\Call{AddEdges}{\textit{pset}, S}] \\
      \System''' = \System''[\ALocs := \System''.\ALocs(a).\code{pushback}(S)]
    }
    {\br{\System, \code{Store}(l, a, \mo); \textit{ss}, T} \Yields \br{\System''', \delta; \textit{ss}, T'}}
  \end{mathpar}

  \textbf{[ATOMIC RMW]}
  \begin{mathpar}
    \inferrule*
    { \br{\System, T} \rightarrow_{\textit{rmw}} \br{\System', T'} \\
      R.t = T'.t \\ R.a = a \\ R.mo = \mo \\\\
      \br{\System', l = \code{Load}(a, \mo), T'} \rightarrow \br{\System'', \textit{ss}, T''}  \\\\
      R.\rf = S \\ 
      T''' = [\text{RMW}](\System'', R, T'') \\\\
      \System''' = \System''[M := \System''.M.\Call{AddRMWEdge}{\code{GetNode}(R.\rf), \code{GetNode}(R)}] \\
      \System'''' = \System'''[\ALocs := \System'''.\ALocs(a).\code{pushback}(R)]
    }
    { \br{\System, \code{RMW}(a, \mo, F); \textit{ss}, T} \Yields \\\\ 
      \br{\System'''', l = F(l); R.v = \System''''.\NALocs(l); \code{Store}(l, a, \mo); \delta; \textit{ss}, T'''}}
  \end{mathpar}

  \textbf{[ATOMIC FENCE]}
  \begin{mathpar}
    \inferrule*
    {
      F.t = T.t \\ F.\mo = \mo \\
      T' = [\text{FENCE}](\System, F, T) \\
      \System' = \System[\FenceInfo := \System.\FenceInfo(t).\code{pushback}(F)  ]
    }
    { \br{\System, \code{Fence}(\mo); \textit{ss}, T} \Yields
      \br{\System', \delta; \textit{ss}, T'}}
  \end{mathpar}

  \caption{Semantics for atomic statements \label{fig:atomic-statement}}
}
\end{figure}

\begin{figure}[!htb]
{\footnotesize
  \begin{algorithmic}[1]
    \Procedure{BuildMayReadFrom}{$L$}
    \State \textit{ret} := $\emptyset$
    \If{$L.\mo == \code{seq\_cst}$}
      \State $S := \textit{last\_sc\_store}(L.a, L)$ \label{line:may-read-from-lastsc}
	\EndIf
	\mycomment{Maybe no need to keep the for all loop?}
    \ForAll{threads $t$}
      \State $\textit{stores} := \textit{stores}(t, L.a)$
      \State $\textit{base} := \{X \in \textit{stores} \mid \neg (X \stackrel{\hb}{\rightarrow} L) \lor (X \stackrel{\hb}{\rightarrow} L \land (\nexists Y \in \textit{stores} \text{ . } X \stackrel{\sbo}{\rightarrow} Y \stackrel{\hb}{\rightarrow} L) ) \}$ \label{line:may-read-from-base}
        \If{$L.\mo == \code{seq\_cst} \land S \neq$ null}
          \State $\textit{base} := \textit{base} \setminus \{ X \in \textit{stores} \mid X \stackrel{\sco}{\rightarrow} S \lor X \stackrel{\hb}{\rightarrow} S \}$ \label{line:may-read-from-scrm}
        \EndIf
      \State $\textit{ret} := \textit{ret} \cup \textit{base}$
    \EndFor
    \If{$L$ is rmw}
      \State $\textit{ret} := \{ X \in \textit{ret} \mid \text{no rmw has read from } X\}$
    \EndIf
    \State \Return \textit{ret}
    \EndProcedure
  \end{algorithmic}
  \caption{Pseudocode for computing \textit{may-read-from} sets \label{alg:may-read-from}}
}
\end{figure}

\begin{figure}[!htb]
{\footnotesize
  \begin{algorithmic}[1]
    \Procedure{WritePriorSet}{$S$}
    \State $\textit{priorset} := \emptyset$; $F_S := \textit{last\_sc\_fence}(S.t)$;
	  \textit{is\_sc\_store} := ( $S.\mo$ == \code{seq\_cst} )
	\If{\textit{is\_sc\_store}}
	  \State add $\textit{last\_sc\_store}(S.a, S)$ to \textit{priorset} \label{line:wpriorset-sc}
	\EndIf
    \ForAll{threads $t$}
	  \State $F_t := \textit{last\_sc\_fence}(t)$
	  \State $F_b := \textit{last}(\{ F \in \textit{sc\_fences}(t) | F_S \neq \text{null} \land F \stackrel{\sco}{\rightarrow} F_S \})$
	  \State $S_1 := \textit{last}(\{ X \in \textit{stores}(t, S.a) \mid \textit{is\_sc\_store}
	    \land F_t \neq \text{null} \land X \stackrel{\sbo}{\rightarrow} F_t \})$
	  \State $S_2 := \textit{last}(\{ X \in \textit{sc\_stores}(t, S.a) \mid F_S \neq \text{null} \land X \stackrel{\sco}{\rightarrow} F_S \})$
	  \State $S_3 := \textit{last}(\{ X \in \textit{stores}(t, S.a) \mid F_b \neq \text{null}
	    \land X \stackrel{\sbo}{\rightarrow} F_b \})$ \label{line:wpriorset-fence}
	  \State $S_4 := \textit{last}(\{ X \in \textit{load\_stores}(t, S.a) \mid X \stackrel{\hb}{\rightarrow} S \})$ \label{line:wpriorset-hb}
	  \State add $\textit{get\_write} \, ( \textit{last}(\{ S_1, S_2, S_3, S_4 \} ))$ to \textit{priorset} \label{line:wpriorset-last}
    \EndFor
	\State \Return \textit{priorset}
    \EndProcedure
  \end{algorithmic}

  \begin{algorithmic}[1]
    \Procedure{ReadPriorSet}{$L$, $S$}
    \State $\textit{priorset} := \emptyset$; $F_L := \textit{last\_sc\_fence}(L.t)$;
	  \textit{is\_sc\_load} := ( $L.\mo$ == \code{seq\_cst} )
    \ForAll{threads $t$}
	  \State $F_t := \textit{last\_sc\_fence}(t)$
	  \State $F_b := \textit{last}(\{ F \in \textit{sc\_fences}(t) \mid F_L \neq \text{null} \land F \stackrel{\sco}{\rightarrow} F_L \})$
      \State $S_1 := \textit{last}(\{ X \in \textit{stores}(t, L.a) \mid \textit{is\_sc\_load}
	    \land F_t \neq \text{null} \land X \stackrel{\sbo}{\rightarrow} F_t \})$ \label{line:rpriorset-s1}
	  \State $S_2 := \textit{last}(\{ X \in \textit{sc\_stores}(t, L.a) \mid F_L \neq \text{null} \land X \stackrel{\sco}{\rightarrow} F_L \})$ \label{line:rpriorset-s2}
	  \State $S_3 := \textit{last}(\{ X \in \textit{stores}(t, L.a) \mid F_b \neq \text{null} \land X \stackrel{\sbo}{\rightarrow} F_b \})$ \label{line:rpriorset-s3}
	  \State $S_4 := \textit{last}(\{ X \in \textit{load\_stores}(t, L.a) \mid X \stackrel{\hb}{\rightarrow} L \})$ \label{line:rpriorset-s4}
	  \State $A := \textit{get\_write}(\textit{last}(\{ S_1, S_2, S_3, S_4 \}))$ \label{line:rpriorset-last}
	  \If{$A \neq S$}
	    \State add $A$ to \textit{priorset}
	  \EndIf
    \EndFor
	\For{each $e$ in \textit{priorset}}
	  \If{$e$ is reachable from $S$ in $\Mograph$} \label{line:rpriorset-reject}
	    \State \Return ($\emptyset$, false)
	  \EndIf
	\EndFor
	\State \Return (\textit{priorset}, true)
    \EndProcedure
  \end{algorithmic}
  \caption{Pseudocode for computing \textit{priorsets} for atomic stores and loads \label{alg:priorset}}
}
\end{figure}

\subsection{Equivalent to Axiomatic Model}

We make our axiomatic model precise and prove the equivalence of our operational
and axiomatic models in \paper{Section A of our technical report~\cite{c11tester-arxiv}.}\techreport{Section~\ref{sec:equivalence} of the Appendix.}

\section{Implementation}

We next present several aspects of the \tool implementation.
\techreport{
Section~\ref{sec:pruning} presents \tool's support for limiting memory usage.
Section~\ref{sec:mixedmode} presents \tool's support for 
mixed mode accesses to a memory location.}
\techreport{Section~\ref{sec:scheduling} discusses the overheads of different
approaches to controlling thread schedules.
Section~\ref{sec:borrowing} describes how \tool implements thread
local storage with fiber-based scheduling.
Section~\ref{sec:staticinit} presents \tool's support for static
initializers. Section~\ref{sec:repeat} presents \tool's support for repeated execution.
}

\subsection{Pruning the Execution Graph}\label{sec:pruning}

While keeping the complete C/C++ execution graph and execution trace is
feasible for short executions and can help with debugging, for longer
executions their size eventually becomes too large to store in memory.
\Naively pruning the execution trace to retain the most recent actions
is not safe---an older store $S_A$ to an atomic location $X$ in the trace can be
modification ordered after a later store $S_B$ to $X$ in the trace.  If
a thread has already read from $S_A$, it cannot read from $S_B$ because it is modification ordered before $S_A$.
\Naively pruning $S_A$ from execution graph without also removing $S_B$ might erroneously produce an
invalid execution in which a thread reads from $S_A$ and then $S_B$.

\Tool supports two approaches to limiting memory usage: (1) a
conservative mode that limits the size of the execution graph with the
constraint that \tool must retain the ability to generate all possible
executions and (2) an aggressive mode that can potentially reduce
the set of executions that \tool can produce.

\paragraph{Conservative Mode} The key idea behind the conservative
mode is to compute a set of older stores that can no longer be read by
any thread and thus can be safely removed from the execution graph.
The basic idea is to compute the latest action $A_t$ for each thread
$t$ such that for the last action $L_{t'}$ in every other thread $t'$,
we have $A_t \stackrel{\hb}{\rightarrow} L_{t'}$.  If action $S$ is a store
that either happens before $A_t$ or is $A_t$, then any new loads from
the same memory location must either read from $S$
or some store that is modification ordered after $S$.
Thus any store $S_{\text{old}}$ that is modification ordered before the
store $S$ can
no longer be read from by any thread and can be safely pruned.

\Tool efficiently computes a clock vector $CV_{\text{min}}$ to identify such actions $A_t$ for each thread by
using the intersection operator, $\cap$, to combine the clock vectors
of all running threads.  We define the intersection operator $\cap$ as follows:
\[CV_1 \cap CV_2 \triangleq \lambda t.min(CV_1(t), CV_2(t)).\]

\Tool then searches for stores that happen before these operations.
It then uses the $\Mograph$ to identify old
stores to prune.  Finally, it prunes these stores and any loads that read from
them.

\paragraph{Aggressive Mode} If a thread fails to synchronize with
other threads, this can prevent \tool from freeing much of the execution graph or execution trace as such
a thread can potentially read from older stores in the execution trace and thus prevent freeing those stores.  In
the aggressive mode, the user provides a window of the trace that
\tool attempts to keep in the graph.  Simply deleting all memory
operations before that window is not sound as newer (with respect to
the trace) memory operations may be modification ordered before older
memory operations.  Thus removing older memory operations could cause
\tool to erroneously allow loads to read from stores they should not.

For a store $S$ outside of this window, \tool attempts to remove all
stores modification ordered before $S$.  Such stores can in some cases be inside of the window that \tool
attempts to preserve, but they must also be removed.  \Tool then removes
any loads that read from the removed stores.

\paragraph{Fences}

Release fences that happen before actions whose sequence numbers correspond to components of $CV_{\text{min}}$ are
not necessary to keep since every running thread has already synchronized with
a later point in the respective thread's execution.  Thus such release
fences can be safely removed.

After an acquire fence is executed, its effect is
summarized in the clock vector of subsequent actions in the same
thread.  Thus acquire fences can be safely removed.

Sequentially consistent fences that happen before $CV_{\text{min}}$
are no longer necessary since the happens-before relation will enforce
the same orderings.  Thus, such sequentially consistent fences can be
safely removed.

\techreport{
\subsection{Supporting Mixed Access Modes\label{sec:mixedmode}}

The C/C++ memory model is silent on the semantics of how atomic
accesses interact with non-atomics accesses to the same memory
location. Researchers have recognized this as a serious limitation of
the standard~\cite{batty2015problem}. It is necessary to handle
mixtures of atomics and non-atomics in \tool for three reasons: (1)
\code{atomic\_init} is implemented in the header files as a
non-atomic store and may race with concurrent atomic accesses to the same memory location, (2) memory can be reused
in C/C++, \eg via \code{malloc} and \code{free}, and the new use may use a memory location for a different
purpose, and (3) C/C++ programs may use non-atomic accesses to copy memory
that contains atomics, \eg via \code{realloc} or \code{memcpy}.  \Tool thus supports
non-atomic operations that access the same memory location as
an atomic as they must be tolerated provided that the accesses are ordered by
the happens-before relation.  If the accesses conflict and are not ordered by happens-before, \tool reports a data race.

Handling non-atomic stores poses a challenge.  For performance
reasons, it is important to implement non-atomic stores as simple
writes to memory.  But if a non-atomic store is later read by an
atomic load, then \tool must include that non-atomic store in the
modification order graph and other internal data structures.
The challenge is that by the time \tool observes the atomic load, it
has lost information about the non-atomic store.

\Tool uses a FastTrack~\cite{fasttrack}-like approach to race
detection.  It maintains a 64-bit shadow word for
each byte of memory.  The shadow word either contains 25-bit read and
write clocks and 6-bit read and write thread identifiers or a
reference to an expanded access record.  We use one bit in the shadow
word to record whether the last store to the address was from a
non-atomic or an atomic store.  If \tool performs an atomic access to a
memory location that was last written to by a non-atomic store, \tool
creates a special non-atomic access record and adds the access to the
modification order graph.

Many applications also contain legacy libraries that use pre-C/C++11
atomic operations such as LLVM intrinsics and volatile accesses.
\Tool supports converting such volatile accesses into atomic accesses
(with a user specific memory order) to allow code that incorporates legacy libraries to execute.
}

\paper{
\subsection{Race Detection}
\Tool uses a FastTrack~\cite{fasttrack}-like approach to race
detection with the per-thread happens-before clock vectors from the
operational semantics.  It uses the standard C/C++ definition of races.
Section 7.2 of our technical report~\cite{c11tester-arxiv} discusses
this in more detail.
}

\subsection{Scheduling}\label{sec:scheduling}

There are two general techniques for controlling the schedule for executing
threads.  The first technique is to map application threads to kernel
threads and then use synchronization constructs to control which
thread takes a step.  The second technique is to simulate application
threads with user threads or fibers that are all
mapped to one kernel thread.  While there is a proposal for user-space
control of thread scheduling that provides very low latency context
switches, unfortunately it still has not been implemented in the mainline Linux
kernel~\cite{googleuser} after six years.

We implemented a microbenchmark on x86 to measure the context switch costs
for several implementations of these two techniques.  Our
microbenchmark starts two threads or fibers and measures the time to
switch between these threads.
\paper{The experiment results are available in our technical report~\cite{c11tester-arxiv}.}
\techreport{Figure~\ref{fig:micro} reports the
results of these experiments.  We ran each experiment in two
configurations: (1) in the all-core configuration the microbenchmark
could use all 4 hardware cores and (2) in the single-core
configuration the microbenchmark was pinned to a single hardware thread.}

\techreport{
  \begin{figure}
{\footnotesize
\begin{center}
  \begin{tabular}{|l|l|l|}
  \hline
  Scheduling Approach & Time for & Time for\\
     & all cores & 1 core\\
  \hline
  Pthread condition variable & 1.95$\mu$s& 1.61$\mu$s\\
  Futex & 1.85$\mu$s&1.32$\mu$s\\
  \hline
  Spinning & 0.07$\mu$s& 15,976.7$\mu$s\\
  Spinning w/ yield & 0.21$\mu$s&0.54$\mu$s\\
  \hline
  Swapcontext & 0.34$\mu$s & 0.34$\mu$s\\
  Swapcontext w/ tls & 0.63$\mu$s & 0.63$\mu$s\\
  Setjmp/Longjmp & 0.01$\mu$s& 0.01$\mu$s\\
  Setjmp/Longjmp w/ tls& 0.30$\mu$s& 0.30$\mu$s\\
  \hline
  \end{tabular}
  \end{center}}
\caption{\label{fig:micro}Context Switch Costs}
\end{figure}
}

For the kernel threads, we implemented four approaches to
context switches.  The first approach uses standard pthread condition
variables and was generally the slowest approach.  The second approach
uses Linux futexes and is a little faster.  The next
two approaches use spinning to wait.  Simply spinning is very fast if
every thread has its own core.  As soon as two threads have to share a
core, this approach becomes 10,000$\times$ slower than the other approaches
because it has to wait for a scheduling epoch to occur to switch
contexts.  We also implemented a version that adds a yield call.  This
hurts performance if both threads run on their own core, but
significantly helps performance if threads share a core.  But in
general, spinning is problematic as idle threads keep cores busy.

For the fiber-based approaches, we used both \code{swapcontext} and
\code{setjmp} to implement fibers.  \code{Swapcontext} is significantly
slower than \code{setjmp} because it makes a system call to update the
signal mask.  An issue with these approaches is that neither call
updates the register that points to thread local storage.  Updating
this register requires a system call, and this slows down both fiber
approaches.  \techreport{We report context switches with this system call in the
``w/ tls'' entries.}

For practical implementation strategies, the fiber-based approach is
faster than kernel threads.  Thus, \tool uses
fibers implemented via \code{swapcontext} to simulate application threads.

\subsection{Thread Context Borrowing}\label{sec:borrowing}

A major challenge with implementing fibers is supporting thread local
storage.  The specification for thread local storage on
x86-64~\cite{tlslinux} is complicated and leaves many important
details implementation-defined and these details vary across different
versions of the standard library.  Generating a correct thread local
storage region for each thread is a significant effort as it requires
continually updating \tool code to support the current set of library
implementation strategies.  This is complicated by the fact that
creating the thread local storage may involve calling initializers and
freeing the thread local storage may involve calling destructors.

Instead, \tool implements a technique for borrowing the thread context
including the thread local storage from a kernel thread.  The idea is
that for each fiber \tool creates a real kernel thread and the fiber
borrows the kernel thread's entire context including its thread local
storage.  

\Tool implements thread context borrowing by first creating and
locking a mutex to protect the thread context and then creating a new
kernel thread to serve as a lending thread that lends its context to
\tool.  The lending thread then creates a fiber context and switches
to the fiber context.  The fiber context then transfers the lending
thread's context along with its thread local storage to the \tool.
Finally, the fiber context grabs the context mutex to wait for the
\tool to return its context.  Once the application thread is
finished, \tool returns the thread context to the lending thread by
releasing the context mutex.  The lending thread then switches back to
its original context, frees its fiber context, and then exits.
Migrating thread local storage on x86 requires a system call to change
the \emph{fs} register.  \Tool implements thread context borrowing for
x86, but the basic idea should work for any architecture.

\techreport{
\subsection{Static Initializers}\label{sec:staticinit}

Static initializers in C++ can and do create threads, perform atomic
operations, and call arbitrary functions in the C++ and pthread
libraries.  \Tool guards access to itself with initialization checks.
In the first call to a \tool routine, \tool initializes itself and
converts the current application thread into a fiber context.  It then
takes control of the execution and controls the remainder of the
program execution.  This allows \tool to support programs that perform
arbitrary operations in their static initializers.
}

\subsection{Repeated Execution}\label{sec:repeat}

\Tool supports repeatedly executing the same benchmark to find hard-to-trigger bugs.
It can be desirable for testing algorithms to maintain
state between executions to attempt to explore different program
behaviors across different executions.  \Tool maintains its internal
state across executions of the application under test and resets the
application's state between executions.

\Tool uses fork-based snapshots to restore the application to its
initial state.  \Tool uses the \code{mmap} library call to map a shared memory region to store its
internal state.  The data in this shared memory region persists across
different executions.  This state allows \tool to report data races
only once as opposed to reporting the same race on each execution.  It
also allows for the creation of smart plugins that explore
different behaviors across different executions.

\section{Evaluation\label{sec:eval}}
We compare \tool with both tsan11rec, a race detector that supports controlled execution \cite{tsan11rec} and tsan11 \cite{tsan11}, a race detector that relies on the operating system scheduler to control the scheduling of threads. We ran our experiments on an Ubuntu Linux 18.04 LTS machine with a 6 core Intel Core i7-8700K CPU and 64GB RAM.  
We first evaluated the above tools on buggy implementations of seqlock
and reader-writer lock to check whether all three tools can detect the
injected bugs.  Then we evaluated the three tools on both a set of five applications that make extensive use of C/C++ atomics and the data structure benchmarks used to evaluate CDSChecker previously \cite{oopsla2013}.

\techreport{
We were not able to build tsan11rec and tsan11 directly on our machine due to dependencies on legacy versions of software. Nevertheless, we compiled tsan11rec and tsan11 inside two docker containers whose base images were both Ubuntu 14.04 LTS. The tsan11rec-instrumented benchmarks were compiled with Clang v4.0 revision 286346, the tsan11-instrumented benchmarks were compiled with Clang v3.9 revision 375507, and the \tool-instrumented benchmarks were compiled with Clang v8.0 revision 346999. 
}

The way these three tools support multi-threading differs significantly. \Tool sequentializes thread executions and only allows one thread to execute at a single time, tsan11 allows multiple threads to execute in parallel, while tsan11rec falls in between---it sequentializes visible operations (such as atomics, thread operations, and synchronization operations) and runs invisible operations in parallel.  The closest tool to compare \tool with is tsan11rec because both \tool and tsan11rec support controlled scheduling, while results for tsan11 are also presented for completeness.  Although both tsan11 and tsan11rec execute all or some operations in parallel, we present a best effort comparison in the following.

\subsection{Benchmarks with Injected Bugs}

We have injected bugs into two commonly used data structures and verified
that both tsan11 and tsan11rec miss these bugs due to the restrictions
of their memory models and that the buggy executions contained cycles in 
$\hb \cup \rf \cup \mo \cup \sco$.

\paragraph{Seqlock}
We took the seqlock implementation from Figure 5 of Hans Boehm's 
MSPC 12 paper \cite{seqlocks}, made the writer correctly use
release atomics for the data field stores,
and injected a bug by weakening atomics that initially increment
the counter to relaxed memory ordering.

\paragraph{Reader-Writer Lock}
We also implemented a broken reader-writer lock where the write-lock operation incorrectly uses relaxed atomics.
The test case uses the read-lock to protect reads from atomic variables and the write-lock to protect writes to atomic variables.

\Tool was able to detect the injected bugs in the broken seqlock and reader-writer lock
with bug detection rates of 28.8\% and 55.3\%, respectively, in 1,000 runs.
However, tsan11 and tsan11rec failed to detect the bugs in 10,000 runs.

\subsection{Real-World Applications}

Ideally, we would evaluate the tools against real world applications
that make extensive use of C/C++ atomics.  However, to our knowledge,
no such standard benchmark suite exists so far.  So we gathered our benchmarks
through searching for benchmarks evaluated in previous work
as well as concurrent programs on GitHub.

The five large applications that we have gathered include: 
GDAX~\cite{gdax}, an in-memory copy of the order book of the GDAX
cryptocurrency exchange; Iris~\cite{iris}, a low-latency C++ logging
library; Mabain~\cite{mabain}, a key-value store library;
Silo~\cite{silocode,silopaper}, a multicore in-memory storage engine;
and the Firefox JavaScript engine
release
50.0.1.\footnote{https://ftp.mozilla.org/pub/firefox/releases/50.0.1/source/}
To make our results as reproducible as possible, we tested the
JavaScript engine using the offline version of JSBench
v2013.1.~\cite{jsbench} \footnote{https://plg.uwaterloo.ca/\textasciitilde
dynjs/jsbench/}

As the three tools supported multi-threading in different ways,
to make a fair comparison, we ran each experiment on application benchmarks in
both the all-core configuration, where all hardware cores could be utilized,
and the single-core configuration, where the tools were restricted
to running on a single CPU using the Linux command \code{taskset}.
As it is always trivial to parallelize testing by running several copies
of a tool in parallel, the rationale behind the single-core experiment is 
to compare the total CPU time used to execute a benchmark or the equivalent throughput 
under different tools.  However, to understand the performance benefits of parallelism for the other tools, 
we also ran experiments in the all-core configuration.
The performance of \tool does not vary much in two configurations, because
\tool only schedules one thread to run at a time. 

Table~\ref{table:benchmark-results} summarizes the average
and relative standard deviation (in parentheses)
of execution time or throughput for each of the five benchmarks
in the single-core and all-core configurations.
Table~\ref{table:benchmark-results} reports wall-clock time for Iris and Mabain.
The throughput of Silo is the aggregate throughput (\code{agg\_throughput})
reported by Silo, and the unit is ops/sec, \ie the number of database operations performed per second.
The throughput of GDAX is the number of iterations which the entire
data set is iterated over in 120s. 
The time and relative standard deviation
reported for JSBench are the statistics reported
by the python script in JSBench over 10 runs. For the other four benchmarks,
the average and relative standard deviation
of the time and throughput are calculated over 10 runs. 

\Tool is slower than tsan11 in all benchmarks except
Silo in the single-core configuration.
\Tool is faster than tsan11rec in all benchmarks except JSBench in
the all-core configuration.

Figure~\ref{fig:benchmark-results-plot} summarizes
speedups compared to tsan11 on the
single-core configuration for each tool under both configurations,
which are derived from data in Table~\ref{table:benchmark-results}.
Tsan11 on the single-core configuration is set as the baseline and
is omitted from Figure~\ref{fig:benchmark-results-plot}.

Based on the results in Figure~\ref{fig:benchmark-results-plot},
we further calculated the geometric mean of the speedup over the five
benchmarks for each tool under both configurations. 
According to the geometric means, \tool is 14.9$\times$ and 11.1$\times$ faster than
tsan11rec in the single-core configuration and all-core configuration, respectively.
\Tool is 1.6$\times$ and 3.1$\times$ slower than tsan11
in the single-core configuration and all-core configuration, respectively.

\techreport{Table~\ref{table:results-operation-count} presents the number of atomic operations
and normal accesses to shared memory locations executed by \tool for each benchmark. 
As the compiler pass of \tool was adapted from the LLVM ThreadSanitizer pass,
the number of atomic operations and normal accesses to shared memory executed by
tsan11 and tsan11rec should be relatively similar,
except for the two throughput-based benchmarks --- Silo and GDAX,
as the amount of work depends on how fast a tool is.}

\begin{figure}[!htbp]
  \centering
    \includegraphics[scale=0.5]{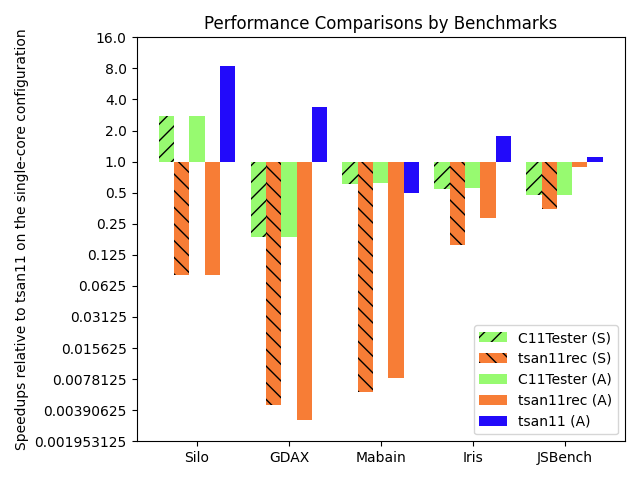}\vspace{-.3cm}
    \caption{Speedups compared to tsan11 on the single-core configuration for all three tools under both configurations, derived from Table~\ref{table:benchmark-results}.   The performance results of tsan11 on the single-core configuration is set as the baseline and is omitted in the Figure.  The larger values the faster the tools are.  The "(S)" label stands for the single-core configuration, and "(A)" stands for the all-core configuration.
    \label{fig:benchmark-results-plot}}

    \includegraphics[scale=0.45]{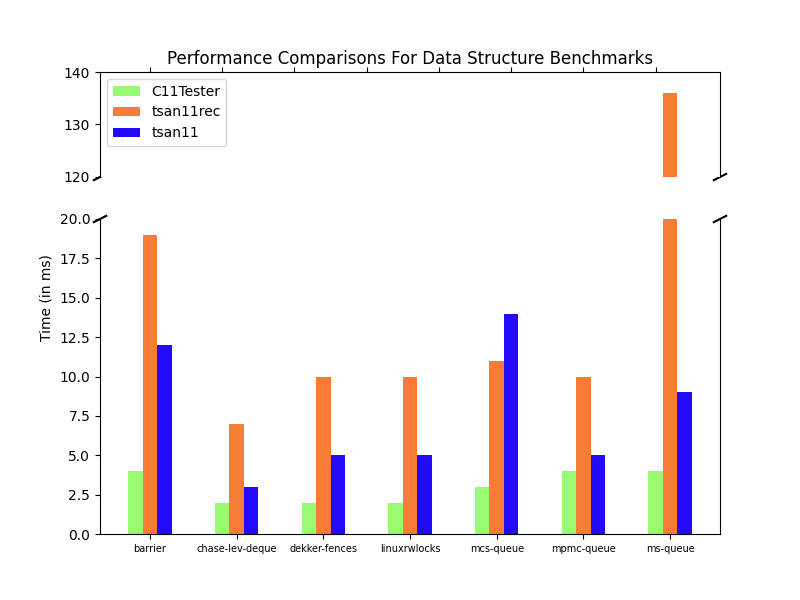}\vspace{-.5cm}
    \caption{Performance comparisons for data structure benchmarks, based on data in table~\ref{table:data-structure-results}.
    \label{fig:data-structure-results-plot}}
\end{figure}

\begin{table*}[!htb]
\centering
 \caption{Performance results for application benchmarks in the \textbf{single-core} and \textbf{all-core} configurations.  The results are averaged over 10 runs.  Relative standard deviation is reported in parentheses. Larger throughputs are better for throughput-based measurements, smaller times are better for time-based measurements.\label{table:benchmark-results}}
\vspace{-.3cm}
{\footnotesize
 \begin{tabular}{|l|r|r|r|r|r|r|l|}
  \hline
  & \multicolumn{3}{c|}{\textbf{Single-core Configuration}} & \multicolumn{3}{c|}{\textbf{All-core Configuration}} &  \\
  \textbf{Test} & \multicolumn{1}{c}{\textbf{\tool}} & \multicolumn{1}{c}{\textbf{tsan11rec}} & \multicolumn{1}{c|}{\textbf{tsan11}} & \multicolumn{1}{c}{\textbf{\tool}} & \multicolumn{1}{c}{\textbf{tsan11rec}} & \multicolumn{1}{c|}{\textbf{tsan11}} & \textbf{Measurement} \\
  \hline
  Silo   & 15267 (0.45\%) & 436 (2.52\%)   & 5496 (4.54\%)  & 15297 (1.17\%) & 438.3 (0.59\%) & 46688 (1.68\%) & Throughput (ops/sec)\\
  GDAX   & 2953  (1.80\%) & 69.3 (0.97\%)  & 15700 (0.12\%) & 2946  (1.64\%) & 49.4 (1.04\%)  & 53362 (11.4\%) & Throughput (\# of iterations)\\
  Mabain & 5.77  (0.25\%) & 593.4 (0.98\%) & 3.513 (1.15\%) & 5.69  (0.04\%) & 441.6 (0.69\%) & 7.00 (0.22\%)  & Time (in s) \\
  Iris   & 8.95  (1.46\%) & 31.31 (0.89\%) & 4.873 (1.64\%) & 8.86  (0.22\%) & 17.20 (1.05\%) & 2.725 (4.07\%) & Time (in s) \\
  JSBench& 1835  (0.26\%) & 2522 (1.41\%)  & 867.8 (0.21\%) & 1836  (0.35\%) & 970.7 (0.68\%) & 781.9 (0.61\%) & Time (in ms) \\
  \hline
 \end{tabular}
}
\end{table*}

\begin{table*}[!htb]
\centering
 \caption{Performance results for data structure benchmarks.  
  The time column gives the time taken to execute the test case once, averaged over 500 runs.  The rate column gives the percentage of executions in which the data race is detected among 500 runs.\label{table:data-structure-results}}
 \vspace{-.3cm}
{\footnotesize
\begin{tabular}{|l|ll|ll|ll|}
  \hline
  \textbf{Test} & \textbf{C11Tester} && \textbf{tsan11rec} && \textbf{tsan11} & \\
  & Time & rate & Time & rate & Time & rate \\
  \hline
  barrier         & 4ms & 76.6\%  & 19ms & 36.4\%  & 12ms & 0.0  \%  \\
  chase-lev-deque & 2ms & 94.6\%  &  7ms & 0.0 \%  & 3ms  & 0.0  \%  \\
  dekker-fences   & 2ms & 21.6\%  & 10ms & 41.4\%  & 5ms  & 53.2 \%  \\
  linuxrwlocks    & 2ms & 86.2\%  & 10ms & 53.4\%  & 5ms  & 1.6  \%  \\
  mcs-lock        & 3ms & 89.4\%  & 11ms & 71.4\%  & 14ms & 0.8  \%  \\
  mpmc-queue      & 4ms & 59.4\%  & 10ms & 58.2\%  & 5ms  & 0.4  \%  \\
  ms-queue        & 4ms & 100.0\% & 136ms& 100.0\% & 9ms  & 100.0\%  \\
  \hline
  Average         &      & 75.4\%  &      & 51.5\%  &      & 22.3\%  \\
  \hline
\end{tabular}
}
\end{table*}

\techreport{
\begin{table*}[h]
\centering
\caption{The number of atomic operations (including synchronization operations such as mutex and condition variable operations) and normal accesses to shared memory locations executed in each benchmark by \tool. \label{table:results-operation-count}}
\vspace{-.3cm}
{\footnotesize
\begin{tabular}{|l|r|r|r|r|r|}
  \hline
    & Silo & GDAX & Mabain & Iris & JSBench \\
  \hline
  \# normal memory accesses & 63.7M & 408.3M & 77.1M  & 35.1M & 5747M \\
  \# atomic operations      & 11.3M & 44.9M  & 2.98M  & 4.8M  & 8.02M \\
  \hline
\end{tabular}
}
\end{table*}
}

\paragraph{Silo} Silo~\cite{silocode,silopaper} is an in-memory database that is designed for performance and scalability for modern multicore machines.
The test driver we used is \code{dbtest.cc}.
We ran the driver for 30 seconds each run with option \code{"-t 5"},
\ie 5 threads in parallel. 

In the first part of the
experiment, Silo was compiled with invariant checking turned on. \Tool
found executions in which invariants were violated.
We found that it was because Silo used volatiles with gcc intrinsic atomics
to implement a spinlock and assumed stronger behaviors from volatiles than
\tool's default handling of volatiles as relaxed atomics.
The bug disappeared when we handled volatile loads and stores as load-acquire
and store-release atomics.  
Volatile variables were commonly used to implement atomic memory accesses
before C/C++11. 
However, this usage of volatile is technically incorrect, because the C++ standard
provides no guarantee when volatiles are mixed with atomics,
and weaker behaviors for volatiles can be exhibited by ARM processors. 

We ran both tsan11rec and tsan11 on Silo for 100
runs with 30s each run.  Tsan11rec was not able to reproduce the
weak behaviors that \tool discovered, while tsan11 could reproduce the weak
behaviors 35\% of the time.  Tsan11rec and tsan11 both found racy
accesses on volatile variables that were used to implement a spin
lock.  \Tool did not report an error message for the
volatile races because \tool intentionally elides race warnings for races involving volatiles and atomic accesses or races involving volatiles and volatiles
because volatiles are in practice still commonly used to implement
atomics.

When measuring performance for Silo, we turned off invariant checking.
We measured performances in terms of aggregate throughput reported by Silo. 
\Tool is faster than tsan11 in the single-core configuration,
because reporting data races caused significant overhead for tsan11 in the case of Silo.


\paragraph{Mabain} Mabain is a lightweight key-value store library~\cite{mabain}.  Mabain
contains a few test drivers that insert key-value pairs concurrently
into the Mabain system---we used
\code{mb\_multi\_thread\_insert\_test.cpp}.  All tools
discovered an application bug that caused assertions in the test driver to fail,
although tsan11 required us to set a different number of threads than our standard
test harness to detect it.
For performance measurements, we
turned off assertions in the test driver. All tools found data races in Mabain.

The application bug is as follows. 
The test driver has one asynchronous writer and a few workers.
The workers and the writer communicate via a shared queue protected by a lock. 
The writer consumes jobs (insertion into the database) in the queue
and insert values into the Mabain database, while the workers submit jobs into the queue.
When workers finish submitting all jobs into the queue, the writer is stopped.
However, there is no check to make sure that all jobs in the queue
have been cleared before the writer is stopped. 
Thus, after the writer is stopped, some values may not be found in the Mabain
database, causing assertion failures.

The time reported in Table~\ref{table:benchmark-results} was measured for
inserting 100,000 key-value pairs into the Mabain system. 

\paragraph{GDAX} GDAX~\cite{gdax} implements an in-memory copy of the order book for the GDAX cryptocurrency exchange using a lock-free skip list with garbage collection from the libcds library~\cite{libcds-url}.
The original GDAX fetches data from a server, but we have recorded
input data from a previous run and modified GDAX to read local
data. All tools reported data races in GDAX.

In our experiment, GDAX was run for 120s each time, during which
5 threads kept iterating over the data set. 
We counted the number of iterations the data set was iterated over by each tool
in each run and computed statistics based on 10 runs. 

\paragraph{Iris}  Iris~\cite{iris} is a low latency asynchronous C++ logging library that buffers data using lock-free circular queues.
The test driver we used to measure performance was \code{test\_lfringbuffer.cpp},
in which there is one producer and one consumer. To make the test driver finish in a
timely manner, we reduced the number of ITERATIONS to 1 million in the test driver.
All tools reported data races in Iris.


\paragraph{Firefox JavaScript Engine}

We compiled the Firefox JavaScript engine release 50.0.1 following the instructions for
building the JavaScript shell with Thread Sanitizer given by the developers of Firefox.
\footnote{https://developer.mozilla.org/en-US/docs/Mozilla/Projects/Thread\_Sanitizer}
We tested the JavaScript engine with the JSBench suite,
which contains 25 JavaScript benchmarks, sampled from
real-world applications.
The Python script of JSBench first calculated the arithmetic mean of all 25
benchmarks over 10 runs, and then took the geometric means of the 25
arithmetic mean, as reported in Table~\ref{table:benchmark-results}. 


\subsection{Data Structure Benchmarks}

To assess the ability of \Tool to discover data races, 
we also used the data structure benchmarks that were originally used to
evaluate CDSChecker and subsequently modified to evaluate tsan11 and tsan11rec.
We used the version of the benchmarks available at
\url{https://github.com/mc-imperial/tsan11}.  Note that sleep
statements were added to 6 of these benchmarks to induce some
variability in the schedules explored by the tsan11 \cite{tsan11}.  We
replicated the same timing strategy used in~\cite{tsan11} and
reported times that were the sum of the user time and system time
measured by the \code{time} command.  Due to differences in the
implementation of the sleep statement, sleep time is partially
included in \tool's user time and thus we removed the sleep statements
for \tool to make the comparison fair.
We executed the benchmarks in the all-core configuration to ensure that we did not put tsan11 at a disadvantage since it does not control the thread schedule.

Table~\ref{table:data-structure-results} summarizes the experiment
results for the data structure benchmarks.  The times reported in Table~\ref{table:data-structure-results} were averaged over 500 runs, and
the rate columns report data races detection rates based on
500 runs.  Out of 7 benchmarks, \tool detects data races with rates
higher than tsan11rec in 4 benchmarks and tsan11 in 5 benchmarks.
Tsan11 and tsan11rec did not detect races in chase-lev-deque, but \tool did.
All three tools always detected races in ms-queue.

\section{Related Work}

Related work falls into three categories: model checkers,
fuzzers, and race detectors.

\paragraph{Model Checkers}
In the context of weak hardware memory models, researchers have
developed stateful model
checkers~\cite{sparcmodel,vechevfmcad,dillmodel}.  Stateful model
checkers however are limited by the state explosion problem and have
the general problem of comparing abstractly equivalent but concretely
different program states.

Stateless model checkers have been developed for the C/C++ memory
model.  CDSChecker can model check real-world C/C++ concurrent data
structures~\cite{oopsla2013,toplascdschecker}.
More recent work has led to the development of
other model checking tools that can efficiently check fragments of the
C/C++ memory model~\cite{tracer,rcmc,genmc}.  Recent work on model
checking for sequential consistency has developed partial order
reduction techniques that only explore all reads-from relations and
do not need to explore all sequentially consistent orderings~\cite{optimalstatelessrf}.
Other tools such as Herd~\cite{herd}, Nitpick~\cite{nitpick}, and
CppMem~\cite{c11popl} are intended to help understand the behaviors of
memory models and do not scale to real-world data structures.

\techreport{
CHESS~\cite{chess} is designed to find and reproduce concurrency bugs
in C, C++, and C\#.  It systematically explores thread interleavings.
However, it can miss concurrency bugs for C/C++ as it does not reorder
memory operations.  Line-Up~\cite{lineup} extends CHESS to check for
linearization.  Like CHESS, it can miss bugs that are exposed by
reordering of memory operations.  The Inspect tool combines stateless
model checking and stateful model checking to model check C and C++
code~\cite{inspect1,inspect2,inspect3,inspect4}.  The Inspect tool
checks code using the sequential consistency model rather than the
weaker C/C++ memory model and therefore may miss concurrency bugs
arising from reordered memory operations.}

Dynamic Partial Order Reduction~\cite{dpor} and Optimal Dynamic
Partial Order Reduction~\cite{odpor} seek to make stateless
model checking more efficient by skipping equivalent executions. Maximal causal
reduction~\cite{statelessmcr} further refines the technique with the
insight that it is only necessary to explore executions in which threads
read different values.  Recent work has extended these algorithms to
handle the TSO and PSO memory
models~\cite{nidhugg,mcrpsotso,pldipsotso}.  SATCheck further develops
partial order reduction with the insight that it is only necessary to
explore executions that exhibit new behaviors~\cite{satcheck}.
CheckFence checks concurrent code by translating it into
SAT~\cite{checkfence}.  Despite these advances, model checking faces
fundamental limitations that prevent it from scaling to full
applications.

\paragraph{Fuzzers}

The Relacy race detector~\cite{relacy} explores thread interleavings
and memory operation reorderings for C++ code.  The Relacy race
detector has several limitations that cause it to miss executions
allowed by the C/C++ memory model.  Relacy imposes an execution order
on the program under test in which it executes the program.  Relacy
then derives the modification order from the execution order; it
cannot simulate (legal) executions in which the modification order is
inconsistent with the execution order.

\techreport{
Industry tools like the IBM ConTest tool support testing concurrent
software.  They work by injecting noise into the execution
schedule~\cite{contest}.  While such tools may increase the likelihood
of finding races, they do not precisely control the schedule.  They also do not handle weak memory models like the C/C++ memory model.
}

Adversarial memory increases the likelihood of observing weak memory
system behaviors for the purpose of testing~\cite{adversarialmemory}.
In the context of Java, prescient memory can simulate some of the weak
behaviors allowed by the Java memory model~\cite{prescientmemory}.
Prescient memory however requires that the entire application be
amenable to deterministic record and replay and uses a single
profiling run to generate future values limiting the executions it can
discover.

Concutest-JUnit extends JUnit with checks for concurrent unit tests
and support for perturbing schedules using randomized
waits~\cite{concutestjunit}.  Concurrit is a DSL designed to help
reproduce concurrency bugs~\cite{concurrit}.  Developers write code in
a DSL to help guide Concurrit to a bug reproducing schedule.
CalFuzzer more uniformly samples non-equivalent
thread interleavings by using techniques inspired by partial order
reduction~\cite{calfuzzer}.  These approaches are largely orthogonal
to \tool.

\paragraph{Race Detectors}

Several tools have been designed to detect data races in code that
uses standard lock-based concurrency
control~\cite{goldilocks,racerx,fasttrack,
  maximalracedetection,conflictexceptions}.  These tools
typically verify that all accesses to shared data are protected by a
locking discipline.  They miss higher-level semantic races that occur
when the locks allow unexpected orderings that produce incorrect
results.

\techreport{
Tsan11 extends the tsan tool to support a fragment of the C/C++11
memory model~\cite{tsan11}.  Tsan11 supports a restricted version of
the C/C++ memory model and cannot produce many of the behaviors that
real-world programs may exhibit.  In particular, it requires that the
modification order relation can be extended to a total order (that is
the order in which tsan executes the statements) and thus can only
produce executions in which the modification order for all memory
locations is a total order.  Tsan11 also does not control the thread
schedule---threads execute in whatever order happens to occur.
Tsan11rec extends tsan11 with support for controlled
execution~\cite{tsan11rec}.  It also has same limitation regarding the
fragment of the memory model it supports.}

\techreport{
In~\cite{Zhou2007}, the lockset algorithm~\cite{Savage1997} is
implemented in hardware. However, such data race detection algorithms
have a high false-positives rate due to their inferential nature as
well as their handling of a limited portfolio of synchronization
primitives, namely locks.  Moreover, these approaches were not
designed to detect errors in concurrent data structures based on
atomic operations.
}

\section{Conclusion\label{sec:conclusion}}

We have presented \tool, which implements a novel approach for
efficiently testing C/C++11 programs.  \Tool supports a larger
fragment of the C/C++ memory model than prior work while still
delivering competitive performance to prior systems.  \Tool uses a
constraint-based approach to the modification order that allows
testing tools to make decisions about the modification order implicitly when
they select the store that a load reads from.  \Tool includes a data
race detector that can identify races.  \Tool supports controlled
scheduling for C/C++11 at lower overhead than prior
systems.  Our evaluation shows that \tool can find bugs in all of our
benchmark applications including bugs that were missed by other tools.

\section*{Acknowledgments}

We thank the anonymous reviewers for their thorough and insightful comments.
We are especially grateful to our shepherd 	Caroline Trippel for her feedback.
We also thank Derek Yeh for his work on performance improvement for the \tool tool.
This work is supported by the National Science Foundation grants
CNS-1703598, OAC-1740210, and CCF-2006948.

\bibliographystyle{plain}
\bibliography{confstrs-long,paper}

\begin{thebibliography}{10}

\bibitem{cpp-draft-n4849}
N4849: Working draft, standard for programminglanguage c++.
\newblock
  \url{http://www.open-std.org/jtc1/sc22/wg21/docs/papers/2020/n4849.pdf}, Jan
  2020.

\bibitem{odpor}
Parosh Abdulla, Stavros Aronis, Bengt Jonsson, and Konstantinos Sagonas.
\newblock Optimal dynamic partial order reduction.
\newblock In {\em Proceedings of the 2014 Symposium on Principles of
  Programming Languages}, pages 373--384, 2014.

\bibitem{nidhugg}
Parosh~Aziz Abdulla, Stavros Aronis, Mohamed~Faouzi Atig, Bengt Jonsson, Carl
  Leonardsson, and Konstantinos Sagonas.
\newblock Stateless model checking for {TSO} and {PSO}.
\newblock In {\em Proceedings of the 21st International Conference on Tools and
  Algorithms for the Construction and Analysis of Systems}, pages 353--367,
  2015.

\bibitem{optimalstatelessrf}
Parosh~Aziz Abdulla, Mohamed~Faouzi Atig, Bengt Jonsson, Magnus L{\aa}ng,
  Tuan~Phong Ngo, and Konstantinos Sagonas.
\newblock Optimal stateless model checking for reads-from equivalence under
  sequential consistency.
\newblock {\em Proceedings of ACM on Programming Languages},
  3(OOPSLA):150:1--150:29, October 2019.

\bibitem{tracer}
Parosh~Aziz Abdulla, Mohamed~Faouzi Atig, Bengt Jonsson, and Tuan~Phong Ngo.
\newblock Optimal stateless model checking under the release-acquire semantics.
\newblock {\em Proceedings of the ACM on Programming Languages},
  2(OOPSLA):135:1--135:29, October 2018.

\bibitem{herd}
Jade Alglave, Luc Maranget, and Michael Tautschnig.
\newblock Herding cats: Modelling, simulation, testing, and data mining for
  weak memory.
\newblock {\em ACM Transactions on Programming Languages and Systems},
  36(2):7:1--7:74, July 2014.

\bibitem{gdax}
F.~Eugene Aumson.
\newblock gdax-orderbook-hpp.
\newblock \url{https://github.com/feuGeneA/gdax-orderbook-hpp}, June 2018.

\bibitem{batty2015problem}
Mark Batty, Kayvan Memarian, Kyndylan Nienhuis, Jean Pichon-Pharabod, and Peter
  Sewell.
\newblock The problem of programming language concurrency semantics.
\newblock In {\em European Symposium on Programming Languages and Systems},
  pages 283--307. Springer, 2015.

\bibitem{batty-memory-model-url}
Mark Batty, Scott Owens, Susmit Sarkar, Peter Sewell, and Tjark Weber.
\newblock \url{https://www.cl.cam.ac.uk/~pes20/cpp/cmm.pdf}, 2011.

\bibitem{c11popl}
Mark Batty, Scott Owens, Susmit Sarkar, Peter Sewell, and Tjark Weber.
\newblock Mathematizing {C++} concurrency.
\newblock In {\em Proceedings of the 38th Annual ACM SIGPLAN-SIGACT Symposium
  on Principles of Programming Languages}, 2011.

\bibitem{cpp11spec}
Pete Becker.
\newblock {ISO/IEC 14882:2011}, {Information} technology -- programming
  languages -- {C++}, 2011.

\bibitem{nitpick}
Jasmin~Christian Blanchette, Tjark Weber, Mark Batty, Scott Owens, and Susmit
  Sarkar.
\newblock Nitpicking {C++} concurrency.
\newblock In {\em Proceedings of the 13th International ACM SIGPLAN Symposium
  on Principles and Practices of Declarative Programming}, pages 113--124,
  2011.

\bibitem{mspc14}
Hans Boehm and Brian Demsky.
\newblock Outlawing ghosts: Avoiding out-of-thin-air results.
\newblock In {\em Proceedings of ACM SIGPLAN Workshop on Memory Systems
  Performance and Correctness}, pages 7:1--7:6, June 2014.

\bibitem{seqlocks}
Hans-J. Boehm.
\newblock Can seqlocks get along with programming language memory models?
\newblock In {\em Proceedings of the 2012 ACM SIGPLAN Workshop on Memory
  Systems Performance and Correctness}, pages 12--20, June 2012.

\bibitem{N3786}
Hans-J. Boehm.
\newblock N3786: Prohibiting ``out of thin air'' results in {C++14}.
\newblock
  \url{http://www.open-std.org/jtc1/sc22/wg21/docs/papers/2013/n3786.htm},
  September 2013.

\bibitem{N3710}
Hans-J. Boehm, Mark Batty, Brian Demsky, Olivier Giroux, Paul McKenney, Peter
  Sewell, Francesco~Zappa Nardelli, et~al.
\newblock N3710: Specifying the absence of ``out of thin air'' results
  ({LWG2265}).
\newblock
  \url{http://www.open-std.org/jtc1/sc22/wg21/docs/papers/2013/n3710.html},
  August 2013.

\bibitem{releasesequences}
Hans-J. Boehm, Olivier Giroux, and Viktor Vafeiades.
\newblock P0982r0: Weaken release sequences.
\newblock
  \url{http://www.open-std.org/jtc1/sc22/wg21/docs/papers/2018/p0982r0.html},
  April 2018.

\bibitem{checkfence}
Sebastian Burckhardt, Rajeev Alur, and Milo M.~K. Martin.
\newblock {CheckFence}: Checking consistency of concurrent data types on
  relaxed memory models.
\newblock In {\em Proceedings of the 2007 Conference on Programming Language
  Design and Implementation}, pages 12--21, 2007.

\bibitem{lineup}
Sebastian Burckhardt, Chris Dern, Madanlal Musuvathi, and Roy Tan.
\newblock Line-up: A complete and automatic linearizability checker.
\newblock In {\em Proceedings of the 2010 ACM SIGPLAN Conference on Programming
  Language Design and Implementation}, pages 330--340, 2010.

\bibitem{prescientmemory}
Man Cao, Jake Roemer, Aritra Sengupta, and Michael~D. Bond.
\newblock Prescient memory: Exposing weak memory model behavior by looking into
  the future.
\newblock In {\em Proceedings of the 2016 ACM SIGPLAN International Symposium
  on Memory Management}, pages 99--110, 2016.

\bibitem{satcheck}
Brian Demsky and Patrick Lam.
\newblock {SATCheck}: {SAT}-directed stateless model checking for {SC} and
  {TSO}.
\newblock In {\em Proceedings of the 2015 Conference on Object-Oriented
  Programming, Systems, Languages, and Applications}, pages 20--36, October
  2015.

\bibitem{mabain}
Changxue Deng.
\newblock Mabain: A fast and light-weighted key-value store library.
\newblock \url{https://github.com/chxdeng/mabain}, November 2018.

\bibitem{tlslinux}
Ulrich Drepper.
\newblock {ELF} handling for thread-local storage.
\newblock \url{https://akkadia.org/drepper/tls.pdf}, August 2013.

\bibitem{concurrit}
Tayfun Elmas, Jacob Burnim, George Necula, and Koushik Sen.
\newblock Concurrit: A domain specific language for reproducing concurrency
  bugs.
\newblock In {\em Proceedings of the 34th ACM SIGPLAN Conference on Programming
  Language Design and Implementation}, PLDI '13, pages 153--164, 2013.

\bibitem{goldilocks}
Tayfun Elmas, Shaz Qadeer, and Serdar Tasiran.
\newblock Goldilocks: A race and transaction-aware {Java} runtime.
\newblock In {\em Proceedings of the 2007 ACM SIGPLAN Conference on Programming
  Language Design and Implementation}, pages 245--255, 2007.

\bibitem{racerx}
Dawson Engler and Ken Ashcraft.
\newblock {RacerX}: Effective, static detection of race conditions and
  deadlocks.
\newblock In {\em Proceedings of the Nineteenth ACM Symposium on Operating
  Systems Principles}, pages 237--252, 2003.

\bibitem{fasttrack}
Cormac Flanagan and Stephen~N. Freund.
\newblock {FastTrack}: Efficient and precise dynamic race detection.
\newblock In {\em Proceedings of the 2009 ACM SIGPLAN Conference on Programming
  Language Design and Implementation}, pages 121--133, 2009.

\bibitem{adversarialmemory}
Cormac Flanagan and Stephen~N. Freund.
\newblock Adversarial memory for detecting destructive races.
\newblock In {\em Proceedings of the 2010 ACM SIGPLAN Conference on Programming
  Language Design and Implementation}, pages 244--254, 2010.

\bibitem{dpor}
Cormac Flanagan and Patrice Godefroid.
\newblock Dynamic partial-order reduction for model checking software.
\newblock In {\em Proceedings of the 2005 Symposium on Principles of
  Programming Languages}, pages 110--121, 2005.

\bibitem{statelessmcr}
Jeff Huang.
\newblock Stateless model checking concurrent programs with maximal causality
  reduction.
\newblock In {\em Proceedings of the 2015 Conference on Programming Language
  Design and Implementation}, pages 165--174, 2015.

\bibitem{maximalracedetection}
Jeff Huang, Patrick Meredith, and Grigore Rosu.
\newblock Maximal sound predictive race detection with control flow
  abstraction.
\newblock In {\em Proceedings of the 35th annual ACM SIGPLAN conference on
  Programming Language Design and Implementation (PLDI'14)}, pages 337--348.
  ACM, June 2014.

\bibitem{mcrpsotso}
Shiyou Huang and Jeff Huang.
\newblock Maximal causality reduction for {TSO} and {PSO}.
\newblock In {\em Proceedings of the 2016 ACM SIGPLAN International Conference
  on Object-Oriented Programming, Systems, Languages, and Applications}, pages
  447--461, 2016.

\bibitem{sparcmodel}
Bengt Jonsson.
\newblock State-space exploration for concurrent algorithms under weak memory
  orderings.
\newblock {\em SIGARCH Computer Architecture News}, 36(5):65--71, June 2009.

\bibitem{c11spec}
ISO JTC.
\newblock {ISO/IEC 9899:2011}, {Information} technology -- programming
  languages -- {C}, 2011.

\bibitem{libcds-url}
Max Khiszinsky.
\newblock \url{https://github.com/khizmax/libcds}, Dec 2017.

\bibitem{rcmc}
Michalis Kokologiannakis, Ori Lahav, Konstantinos Sagonas, and Viktor
  Vafeiadis.
\newblock Effective stateless model checking for {C/C++} concurrency.
\newblock {\em Proceedings of the ACM on Programming Languages},
  2(POPL):17:1--17:32, December 2017.

\bibitem{genmc}
Michalis Kokologiannakis, Azalea Raad, and Viktor Vafeiadis.
\newblock Model checking for weakly consistent libraries.
\newblock In {\em Proceedings of the 40th ACM SIGPLAN Conference on Programming
  Language Design and Implementation}, PLDI 2019, pages 96--110, 2019.

\bibitem{vechevfmcad}
Michael Kuperstein, Martin Vechev, and Eran Yahav.
\newblock Automatic inference of memory fences.
\newblock In {\em Proceedings of the Conference on Formal Methods in
  Computer-Aided Design}, pages 111--120, 2010.

\bibitem{l-clocks}
Leslie Lamport.
\newblock Time, clocks, and the ordering of events in a distributed system.
\newblock {\em Communications of the ACM}, 21(7):558--565, July 1978.

\bibitem{tsan11}
Christopher Lidbury and Alastair~F. Donaldson.
\newblock Dynamic race detection for {C++11}.
\newblock In {\em Proceedings of the 44th ACM SIGPLAN Symposium on Principles
  of Programming Languages}, POPL 2017, pages 443--457, New York, NY, USA,
  2017. ACM.

\bibitem{tsan11rec}
Christopher Lidbury and Alastair~F. Donaldson.
\newblock Sparse record and replay with controlled scheduling.
\newblock In {\em Proceedings of the 40th ACM SIGPLAN Conference on Programming
  Language Design and Implementation}, PLDI 2019, pages 576--593, 2019.

\bibitem{conflictexceptions}
Brandon Lucia, Luis Ceze, Karin Strauss, Shaz Qadeer, and Hans Boehm.
\newblock Conflict exceptions: Simplifying concurrent language semantics with
  precise hardware exceptions for data-races.
\newblock In {\em Proceedings of the 37th Annual International Symposium on
  Computer Architecture}, pages 210--221, 2010.

\bibitem{symbiosis}
Nuno Machado, Brandon Lucia, and Lu\'{\i}s Rodrigues.
\newblock Concurrency debugging with differential schedule projections.
\newblock {\em SIGPLAN Not.}, 50(6):586--595, June 2015.

\bibitem{cortex}
Nuno Machado, Brandon Lucia, and Lu\'{\i}s Rodrigues.
\newblock Production-guided concurrency debugging.
\newblock {\em SIGPLAN Not.}, 51(8), February 2016.

\bibitem{deloren}
Pablo Montesinos, Luis Ceze, and Josep Torrellas.
\newblock Delorean: Recording and deterministically replaying shared-memory
  multiprocessor execution efficiently.
\newblock {\em SIGARCH Comput. Archit. News}, 36(3):289--300, June 2008.

\bibitem{chess}
Madanlal Musuvathi, Shaz Qadeer, Piramanayagam~Arumuga Nainar, Thomas Ball,
  Gerard Basler, and Iulian Neamtiu.
\newblock Finding and reproducing {Heisenbugs} in concurrent programs.
\newblock In {\em Proceedings of the Eighth USENIX Symposium on Operating
  Systems Design and Implementation}, pages 267--280, 2008.

\bibitem{oopsla2013}
Brian Norris and Brian Demsky.
\newblock {CDSChecker}: Checking concurrent data structures written with
  {C/C++} atomics.
\newblock In {\em Proceedings of the 2013 Conference on Object-Oriented
  Programming, Systems, Languages, and Applications}, pages 131--150, October
  2013.

\bibitem{toplascdschecker}
Brian Norris and Brian Demsky.
\newblock A practical approach for model checking {C/C++11} code.
\newblock {\em ACM Transactions on Programming Languages and Systems},
  38(3):10:1--10:51, May 2016.

\bibitem{oota}
Peizhao Ou and Brian Demsky.
\newblock Towards understanding the costs of avoiding out-of-thin-air results.
\newblock {\em Proceedings of the ACM on Programming Languages Volume 2 Issue
  OOPSLA}, 2(OOPSLA):136:1--136:29, October 2018.

\bibitem{dillmodel}
Seungjoon Park and David~L. Dill.
\newblock An executable specification and verifier for relaxed memory order.
\newblock {\em IEEE Transactions on Computers}, 48(2):227--235, February 1999.

\bibitem{jsbench}
Gregor Richards, Andreas Gal, Brendan Eich, and Jan Vitek.
\newblock Automated construction of javascript benchmarks.
\newblock In {\em Proceedings of the 2011 ACM International Conference on
  Object Oriented Programming Systems Languages and Applications}, OOPSLA
  ’11, page 677–694, New York, NY, USA, 2011. Association for Computing
  Machinery.

\bibitem{concutestjunit}
Mathias~Guenter Ricken.
\newblock {\em A Framework for Testing Concurrent Programs}.
\newblock PhD thesis, Houston, TX, USA, 2011.
\newblock AAI3463989.

\bibitem{Savage1997}
Stefan Savage, Michael Burrows, Greg Nelson, Patrick Sobalvarro, and Thomas
  Anderson.
\newblock Eraser: A dynamic data race detector for multithreaded programs.
\newblock {\em ACM Transactions on Computer Systems}, 15:391--411, November
  1997.

\bibitem{calfuzzer}
Koushik Sen.
\newblock Effective random testing of concurrent programs.
\newblock In {\em Proceedings of the Twenty-second IEEE/ACM International
  Conference on Automated Software Engineering}, ASE '07, pages 323--332, 2007.

\bibitem{silocode}
Stephen Tu, Wenting Zheng, and Eddie Kohler.
\newblock Silo: Multicore in-memory storage engine.
\newblock \url{https://github.com/stephentu/silo}, March 2015.

\bibitem{silopaper}
Stephen Tu, Wenting Zheng, Eddie Kohler, Barbara Liskov, and Samuel Madden.
\newblock Speedy transactions in multicore in-memory databases.
\newblock In {\em Proceedings of the Twenty-Fourth ACM Symposium on Operating
  Systems Principles}, SOSP '13, pages 18--32, 2013.

\bibitem{googleuser}
Paul Turner.
\newblock User-level threads...with threads.
\newblock
  \url{https://blog.linuxplumbersconf.org/2013/ocw/system/presentations/1653/original/LPC\%20-\%20User\%20Threading.pdf\#4},
  2013.

\bibitem{contest}
Shmuel Ur.
\newblock Testing and debugging concurrent software – challenges and
  solutions.
\newblock
  \url{https://www.research.ibm.com/haifa/conferences/hvc2010/present/Testing_and_Debugging_Concurrent_Software.pdf},
  2010.

\bibitem{vafeiadis2013relaxed}
Viktor Vafeiadis and Chinmay Narayan.
\newblock Relaxed separation logic: A program logic for c11 concurrency.
\newblock In {\em Proceedings of the 2013 ACM SIGPLAN International Conference
  on Object-Oriented Programming, Systems, Languages, and Applications}, pages
  867--884, 2013.

\bibitem{relacy}
Dmitriy Vyukov.
\newblock Relacy race detector.
\newblock \url{http://relacy.sourceforge.net/}, October 2011.

\bibitem{inspect2}
Chao Wang, Yu~Yang, Aarti Gupta, and Ganesh Gopalakrishnan.
\newblock Dynamic model checking with property driven pruning to detect race
  conditions.
\newblock {\em ATVA LNCS}, (126--140), 2008.

\bibitem{inspect4}
Yu~Yang, Xiaofang Chen, Ganesh Gopalakrishnan, and Robert~M. Kirby.
\newblock Distributed dynamic partial order reduction based verification of
  threaded software.
\newblock In {\em Proceedings of the 14th International {SPIN} Conference on
  Model Checking Software}, pages 58--75, 2007.

\bibitem{inspect3}
Yu~Yang, Xiaofang Chen, Ganesh Gopalakrishnan, and Robert~M. Kirby.
\newblock Efficient stateful dynamic partial order reduction.
\newblock In {\em Proceedings of the Fifteenth International SPIN Workshop},
  pages 288--305, August 2008.

\bibitem{inspect1}
Yu~Yang, Xiaofang Chen, Ganesh Gopalakrishnan, and Chao Wang.
\newblock Automatic discovery of transition symmetry in multithreaded programs
  using dynamic analysis.
\newblock In {\em Proceedings of the 16th International {SPIN} Workshop on
  Model Checking Software}, pages 279--295, 2009.

\bibitem{pldipsotso}
Naling Zhang, Markus Kusano, and Chao Wang.
\newblock Dynamic partial order reduction for relaxed memory models.
\newblock In {\em Proceedings of the 36th ACM SIGPLAN Conference on Programming
  Language Design and Implementation}, pages 250--259, 2015.

\bibitem{Zhou2007}
Pin Zhou, Radu Teodorescu, and Yuanyuan Zhou.
\newblock {HARD}: Hardware-assisted lockset-based race detection.
\newblock In {\em Proceedings of the 2007 IEEE 13th International Symposium on
  High Performance Computer Architecture}, pages 121--132, 2007.

\bibitem{iris}
Xinjing Zhou.
\newblock Iris: A low latency asynchronous {C++} logging library.
\newblock \url{https://github.com/zxjcarrot/iris}, October 2015.

\end{thebibliography}

\paper{\clearpage
%
%
%
%
%



\appendix
\section{Artifact Appendix}

\subsection{Abstract}

The artifact contains a \code{c11tester-vagrant} directory and a \code{tsan11-tsan11rec-docker} directory.
The \code{c11tester-vagrant} directory is a vagrant repository that
compiles source codes for C11Tester, LLVM, the companion compiler pass,
and benchmarks for C11Tester.
The \code{tsan11-tsan11rec-docker} directory contains benchmarks and
a docker image with prebuilt LLVMs for tsan11 and tsan11rec.
We had attempted to install tsan11 and
tsan11rec in the same VM as C11Tester.  However, tsan11rec became significantly slower
and some benchmarks were even unrunnable under tsan11rec.  So we had to build tsan11,
tsan11rec, and benchmarks under the same environment as provided by their
artifact documentations.

\subsection{Artifact Check-List (Meta-Information)}

{\small
\begin{itemize}
  \item {\bf Algorithm: } Testing/race detection algorithm for C/C++ memory model.
  \item {\bf Program: } C11Tester.
  \item {\bf Compilation: } Clang 8.0.0 for C11Tester, Clang 3.9.0 for tsan11, and Clang 4.0.0 for tsan11rec.
  \item {\bf Transformations: } An LLVM pass.
  \item {\bf Binary: } Modified LLVMs for tsan11 and tsan11rec are included in the docker image.
  \item {\bf Run-time environment: } Ubuntu 18.04 for C11Tester and Ubuntu 14.04 for tsan11 and tsan11rec.
  \item {\bf Hardware: } An Intel x86 machine with 6 cores.
  \item {\bf Execution: } Automated via shell scripts.
  \item {\bf Metrics: } Execution time, data race detection rate, assertion detection rate.
  \item {\bf Output: } Numerical results printed in console.
  \item {\bf Experiments: } GDAX, Iris, Silo, Mabain, the Javascript Engine of Firefox, a broken seqlock, a broken reader-writer lock, and some data structure benchmarks.  We measure both the performance (execution time or throughput) and the ability to detect data races and assertions. 
  \item {\bf How much disk space required (approximately)?: } 10G for the VM that contains C11Tester and 15G for the docker container that contains tsan11 and tsan11rec.
  \item {\bf How much time is needed to prepare workflow (approximately)?: } About 40 minutes for compilation.
  \item {\bf How much time is needed to complete experiments (approximately)?: } 2 hours for C11Tester, 3 hours for tsan11, and 6.5 hours for tsan11rec.
  \item {\bf Publicly available?: } Yes.
  \item {\bf Code licenses (if publicly available)?: } GNU GPL v2.
  \item {\bf Data licenses (if publicly available)?: } Varies depending on benchmark.
  \item {\bf Workflow framework used?: } Vagrant \& scripts are provided to automate the measurements.
  \item {\bf Archived (provide DOI)?: } https://doi.org/10.1145/3410278
\end{itemize}

\subsection{Description}

\subsubsection{How to Access}

The artifact is available at: \url{https://doi.org/10.1145/3410278}

\subsubsection{Hardware Dependencies}
An Intel x86 CPUs with at least 6 cores and at least 40G RAM is required to reproduce results.  The VM for C11Tester requires 40G RAM because one particular benchmark
(GDAX) consumes 36G RAM under C11Tester. 
Experiments additionally require CPUs to have Intel VT-d support. 

\subsubsection{Software Dependencies}
C++ Compiler, CMake, Clang, LLVM Compiler Infrastructure, Docker, Vagrant, and VirtualBox.

\subsection{Installation}
First download the artifact and extract it.
The extracted file contains two folders: \code{c11tester-vagrant}
and \code{tsan11-tsan11rec-docker}.\\
\code{\$ cd c11tester-artifact} \\

To build C11Tester and benchmarks using Vagrant:\\
\code{\$ cd c11tester-vagrant} \\
\code{\$ vagrant up} \\


The \code{tsan11-tsan11rec-docker} folder contains a docker image named \code{tsan11-tsan11rec-image.tar.gz}
with prebuilt LLVMs for tsan11 and tsan11rec. 
For instructions on creating docker containers from the docker image,
please see the README.md file in the \code{tsan11-tsan11rec-docker} repository.

To find the IP address of the container (assuming the container is named tsan11-tsan11rec-container):\\
\code{\$ docker inspect tsan11-tsan11rec-container}\\

Then use \code{scp} to copy the \code{scripts} and \code{src} directories
in the \code{tsan11-tsan11rec-docker} folder
to the container (replace 172.17.0.2 by the container's IP address):\\
\code{\$ scp -i insecure\_key -r scripts root@172.17.0.2:/data}\\
\code{\$ scp -i insecure\_key -r src root@172.17.0.2:/data}\\

Logging into the container as root (replace 172.17.0.2 by the container's IP address): \\
\code{\$ ssh -i insecure\_key root@172.17.0.2}\\

After logging into the docker container, to build benchmarks
for tsan11 and tsan11rec: \\
\code{\# ./data/scripts/setup.sh}

\subsection{Experiment Workflow}
Scripts are provided to run experiments.
To run experiments for C11Tester, logging into the Vagrant VM: \\
\code{\$ cd \textasciitilde/c11tester-benchmarks} \\
\code{\$ ./do\_test\_all.sh} \\

To run experiments for tsan11, logging into the docker container: \\
\code{\# cd /data/tsan11-benchmarks} \\
\code{\# ./do\_test\_all.sh} \\

To run experiments for tsan11rec, inside the same docker container: \\
\code{\# cd /data/tsan11rec-benchmarks} \\
\code{\# ./do\_test\_all.sh} \\

\subsection{Evaluation and Expected Results}

Once the workflow is completed, the data race detection rates and assertion rates for
data structure benchmarks are printed in the console. 

For application benchmarks, the result for each benchmark is written to
log files (such \code{gdax.log} and \code{silo.log}, etc).  These log files
are stored in \code{all-core/} and \code{single-core/} directories under
the benchmark directories \code{c11tester-benchmarks}, \code{tsan11-benchmarks},
and \code{tsan11rec-benchmarks}. 

The \code{do\_test\_all.sh} script also executes the python script \code{calculator.py}
that prints out result summaries for all of five application benchmarks executed under
both the all-core and single-core configurations. 
Each benchmark directory has this python script.  If you wish to regenerate
result summaries from log files using the python script, you can first
go to one benchmark directory (we will use \code{c11tester-benchmarks} as an example here): \\
\code{\$ cd \textasciitilde/c11tester-benchmarks} \\
and type:\\
\code{\$ python calculator.py all-core} \\
or \\
\code{\$ python calculator.py single-core} \\
to print out result summaries for all of five application benchmarks
executed under the all-core of single-core configuration. 

\subsection{Experiment Customization}
In the benchmark directory, the two scripts
\begin{itemize}
\item \code{tsan11-missingbug/test.sh}
\item \code{cdschecker\_modified\_benchmarks/test.sh}
\end{itemize}
can be customized to run different times by changing the shell
variable \code{TOTAL\_RUN}.

The \code{run.sh} and \code{app\_assertion\_test.sh} scripts
in the benchmark directory accept an optional argument that specifies how many times
the test programs are run.  The default is 10 times. 
Besides that, you can also decide which test program are run by
modifying the \code{TESTS} variable in these two scripts.

\subsection{Notes}
Tsan11 may occasionally get stuck when testing Silo in the single-core configuration.
If this happens, we suggest to rerun Silo individually by customizing the 
\code{tsan11-benchmarks/run.sh} script.



}
\numberwithin{equation}{section}
\setcounter{theorem}{0}

\clearpage

\appendix
\mycomment{
\section{Proofs of Theorem~\ref{thm:main} \label{sec:appendix-proof}}

To prove the correctness of Theorem~\ref{thm:main}, we first prove three Lemmas.
Lemma~\ref{thm:lemma1} and Lemma~\ref{thm:lemma2} characterize some important
properties of \textit{mo-graph} clock vectors.  Lemma~\ref{thm:lemma-backward} proves
one direction in Theorem~\ref{thm:main}.  \textit{Mo-graph} clock vectors are
simply referred to as clock vectors in the following context. 

\begin{lemma} \label{thm:lemma1}
Let $C_0 \stackrel{\mo}{\rightarrow} C_1 \stackrel{\mo}{\rightarrow} ... \stackrel{\mo}{\rightarrow} C_{n-1} \stackrel{\mo}{\rightarrow} C_n$
be a path in a modification order graph $G$, such that
$CV_{C_0} \leq ... \leq CV_{C_n}$.
Then if any new edge $E$ is added to $G$ using procedures in Figure~\ref{fig:mopseudo},
it holds that 
\begin{align}
  CV_{C_0}' \leq ... \leq CV_{C_n}' \label{eq:lemma1-main}
\end{align}
for the updated clock vectors.
We define $CV_{C_i}' := CV_{C_i}$ if the values of $CV_{C_i}$ are not actually updated.
\end{lemma}

\begin{proof}
To simplify notation, we define $CV_i := CV_{C_i}$ for all $i \in \{0...,n\}$.
Let's first consider the case where no \code{rmw} edge is added, \ie the \textsc{AddRMWEdge}
procedure is not called.

By the definition of the union operator, each slot in clock vectors is monotonically increasing
when the \textsc{Merge} procedure is called. By the structure of procedure \textsc{AddEdge}'s
algorithm, a node $X$ is added to $Q$ if and only if this node's clock vector
is updated by the \textsc{Merge} procedure.

Let's assume that adding the new edge $E$ updates any of $CV_0, ..., CV_n$.
Otherwise, it is trivial.  Let $i$ be the smallest integer in $\{0, ..., n\}$
such that $CV_i$ is updated.  Then $CV_k' = CV_k$ for
all $k \in I := \{0, ..., i-1\}$, and we have
\begin{align}
  CV_0' \leq ... \leq CV_i'. \label{eq:lemma1-1}
\end{align}
If $i = 0$, then we take $I = \varnothing$. There are two cases.

\textbf{Case 1}: Suppose $CV_i' \leq CV_j$ for some $j \in \{i+1, ..., n\}$,
let $j_0$ be the smallest such integer.
Then $ CV_k' = CV_k$ for all $k \in \{j_0, ..., n\}$, as nodes $\{C_{j_0},...,C_n\}$ 
will not be added to $Q$ in the \textsc{AddEdge} procedure, and
\begin{align}
  CV_{j_0}' \leq ... \leq CV_n' \label{eq:lemma1-2}
\end{align}
holds trivially.  By line~\ref{addedge:merge} to line~\ref{addedge:merge-end}
in the \textsc{AddEdge} procedure, we have
\begin{align}
  CV_k' = CV_k \cup CV_{k-1}', \label{eq:lemma1-3}
\end{align}
for all $k \in S := \{i+1, ..., j_0-1\}$.  If $j_0$ happens to be $i+1$,
then take $S = \varnothing$.
And we have for all $k \in S$, $CV_{k-1}' \leq CV_k'$. 
Then combining with inequality (\ref{eq:lemma1-1}), we have
\[ CV_0' \leq ... \leq CV_i \leq ... \leq CV_{j_0-1}'. \]
Together with inequality (\ref{eq:lemma1-2}), we only need to
show that $CV_{j_0-1}' \leq CV_{j_0}'$ to complete the proof.

If $j_0 = i+1$, then we are done, because by assumption
$CV_i' \leq CV_{j_0} = CV_{j_0}'$.
If $j_0 > i+1$, then $CV_i' \leq CV_{j_0}$ and $CV_{i+1} \leq CV_{j_0}$ implies
that \[CV_{i+1}' = CV_{i+1} \cup CV_i' \leq CV_{j_0} = CV_{j_0}'. \]
Based on equation (\ref{eq:lemma1-3}), we can deduce in a similar way that
$CV_{i+2}' \leq ... \leq CV_{j_0-1}' \leq CV_{j_0}'$.

\textbf{Case 2}:
Suppose $CV_i \nleq CV_j$ for all $j \in \{i+1, ..., n\}$.
Then by line~\ref{addedge:merge} to line~\ref{addedge:merge-end}
in the \textsc{AddEdge} procedure, all nodes $\{C_i, ..., C_n\}$ are added to
$Q$ in the \textsc{AddEdge} procedure, and
\[ CV_k' = CV_k \cup CV_{k-1}', \]
for all $k \in S := \{i+1, ..., n\}$.  This recursive formula guarantees that
for all $k \in S$, \[ CV_{k-1}' \leq CV_k'. \]
Therefore, combining with inequality (\ref{eq:lemma1-1}), we have
\[ CV_0' \leq ... \leq CV_n'. \]

Now suppose the newly added edge $E$ is a \code{rmw} edge. 
If $E: X \xrightarrow{\textit{rmw}} C_i$ where $i \in \{0,...,n\}$
and $X$ is some node not in path $P$,
then the path $P$ remains unchanged and \Call{AddEdge}{$X$,$C_i$} is called.
Then the above proof shows that inequality (\ref{eq:lemma1-main}) holds. 
If $E:C_i \xrightarrow{\textit{rmw}} X$, then $C_i \stackrel{\mo}{\rightarrow} C_{i+1}$
is migrated to $X \stackrel{\mo}{\rightarrow} C_{i+1}$ by
line~\ref{addrmwedge:migrate-start} to line~\ref{addrmwedge:migrate-end}
in the \textsc{AddRMWEdge} procedure, and $C_i \stackrel{\mo}{\rightarrow} X$ is added.

If $X$ is not in path $P$, then path $P$ becomes
\[C_0 \stackrel{\mo}{\rightarrow} ... \stackrel{\mo}{\rightarrow} C_i \stackrel{\mo}{\rightarrow} X \stackrel{\mo}{\rightarrow} C_{i+1} \stackrel{\mo}{\rightarrow} ... \stackrel{\mo}{\rightarrow} C_n. \]
Since \Call{AddEdge}{$C_i$,$X$} is called, the same proof in the case without \code{rmw} edges
applies.  If $X$ is in path $P$, then $X$ can only be $C_{i+1}$
and the path $P$ remains unchanged.
Otherwise, a cycle is created and this execution is invalid. 
In any case, the same proof applies. 
\end{proof}

Let $\vec{x} = (x_1,x_2,...,x_n)$.
We define the projection function $U_i$ that extracts the $i^{\textit{th}}$
position of $\vec{x}$ as
$U_i(\vec{x}) = x_i,$ where we assume $i \leq n$.

\begin{lemma} \label{thm:lemma2}
Let $A$ be a store with sequence number $s_A$ performed
by thread $i$ in an acyclic modification order graph $G$.
Then $U_i(CV_A) = U_i(\perp_{CV_A}) = s_A$ throughout each execution that terminates. 
\end{lemma}

\begin{proof}
We will prove by contradiction.
Let $S = \{A_1, A_2, ...\}$ be the sequence of stores performed by thread $i$
with sequence numbers $\{s_1, s_2, ...\}$ respectively.
Suppose that there is a point of time in a terminating execution such that the first store $A_n$ 
in the sequence with $U_i(CV_{A_n}) > s_n$ appears.
Sequence numbers are strictly increasing and
by the \textsc{Merge} procedure, $U_i(CV_{A_n}) \in \{s_{n+1},s_{n+2},...,\}$.
Let $U_i(CV_{A_n}) = s_N$ for some $N > n$.

For $U_i(CV_{A_n})$ to increase to $s_N$ from $s_n$,
$CV_{A_n}$ must be merged with the clock vector of some node $X$ (\ie some store $X$)
in $G$ such that $U_i(CV_X) = s_N$.  Such node $X$ is modification ordered before $A_n$.

If $X$ is performed by thread $i$, then $X$ has to be the store
$A_N$, because $U_i(CV_{A_j})$ is unique for all stores $A_j$
in the sequence $S$ other than $A_n$.
Then $\perp_{CV_X} \geq \perp_{CV_{A_n}}$.  By the definition of initial values
of clock vectors and sequence numbers, $X$ happens after
and is modification ordered after $A_n$.
However, $X$ is also modification ordered before $A_n$, and we have a cycle in the graph $G$.
This is a contradiction.

If $X$ is not performed by thread $i$, then $U_i(\perp_{CV_X}) = 0$.
For $U_i(CV_X)$ to be $s_N$,
$X$ must be modification ordered after by some store
$Y$ in $G$ such that $U_i(CV_Y) = s_N$.
If $Y$ is done by thread $i$, then the same argument in the last paragraph
leads to a contradiction; otherwise, by repeating the same argument as in
this paragraph finitely many times (there are only a finite number
of stores in such a terminating execution), we would eventually deduce that
$X$ is modification ordered after some store by thread $i$. Hence, we would have
a cycle in $G$, a contradiction. 

Therefore, the claim is proved. 
\end{proof}

\begin{lemma} \label{thm:lemma-backward}
Let $A$ and $B$ be two nodes that write to the same location in
an acyclic modification order graph $G$.  If $B$ is reachable from $A$ in $G$,
then $CV_A \leq CV_B$.
\end{lemma}

\begin{proof}
Suppose that $B$ is reachable from $A$ in $G$.
Let $A \stackrel{\mo}{\rightarrow} C_1 \stackrel{\mo}{\rightarrow} ... \stackrel{\mo}{\rightarrow} C_{n-1} \stackrel{\mo}{\rightarrow} B$ be the shortest path $P$ from $A$ to $B$ in graph $G$.
To simplify notation, $X \stackrel{\mo}{\rightarrow} Y$ is abbreviated as
$X \rightarrow Y$ in the following.
As the \textsc{AddRMWEdge} procedure calls the \textsc{AddEdge} procedure
to create an \textit{mo} edge, we can assume that all the \textit{mo} edges
in $P$ are created by directly calling \textsc{AddEdge}. 

\textbf{Base Case 1}: Suppose the path $P$ has length 1, \ie $A$ immediately precedes $B$.
Then when the edge $A \rightarrow B$ was formed by calling \Call{AddEdge}{$A$,$B$},
$CV_B$ was merged with $CV_A$ in line~\ref{addedge:merge}
of the \textsc{AddEdge} procedure.  In other words,
$CV_B = CV_B \cup CV_A \geq CV_A.$

\textbf{Base Case 2}: Suppose the path $P$ has length 2, \ie $A \rightarrow C_1 \rightarrow B$. There are two cases:

(a) If $A \rightarrow C_1$ was formed first, then $CV_A \leq CV_{C_1}$.
When $C_1 \rightarrow B$ was formed, $CV_B$ was merged with $CV_{C_1}$ and $CV_{C_1} \leq CV_B$.
According to Lemma~\ref{thm:lemma1}, adding the edge $C_1 \rightarrow B$ or any edge
not in path $P$ (if any such edges were formed before $C_1 \rightarrow B$ was formed)
to $G$ would not break the inequality $CV_A \leq CV_{C_1}$.
It follows that $CV_A \leq CV_{C_1} \leq CV_B$.

(b) If $C_1 \rightarrow B$ was formed first, then 
$CV_{C_1} \leq CV_B$.  Based on Lemma~\ref{thm:lemma1}, this inequality remains true
when $A \rightarrow C_1$ was formed.  Therefore $CV_A \leq CV_{C_1} \leq CV_B$.

\textbf{Inductive Step}: 
Suppose that $B$ being reachable from $A$ implies that $CV_A \leq CV_B$ for
all paths with length $k$ or less, where $k$ is some number greater than 2.
We want to prove that the same holds for paths with length $k + 1$.
Let $P$ be a path from $A$ to $B$  with length $k+1$,
\[ P: A = C_0 \rightarrow C_1 \rightarrow ... \rightarrow C_{k} \rightarrow C_{k+1} = B. \]
We denote $A$ as $C_0$ and $B$ as $C_{k+1}$ in the following.

Let $E: C_i \rightarrow C_{i+1}$ be the last edge formed in path $P$, where $i \in \{0,...,k\}$.
Then before edge $E$ was formed, the inductive hypothesis
implies that $CV_{C_0} \leq ... \leq CV_{C_i}$
and $CV_{C_{i+1}} \leq ... \leq CV_{C_{k+1}}$,
because both $C_0 \rightarrow ... \rightarrow C_i$ and
$C_{i+1} \rightarrow ... \rightarrow C_{k+1}$ have length $k$ or less.
Lemma~\ref{thm:lemma1} guarantees that 
\begin{align*}
  CV_{C_0} &\leq ... \leq CV_{C_i}, \\
  CV_{C_{i+1}} &\leq ... \leq CV_{C_{k+1}}
\end{align*}
remain true if any edge not in path $P$ was added to $G$
as well as the moment when $E$ was formed.
Therefore when the edge $E$ was formed,
we have $CV_{C_i} \leq CV_{C_{i+1}}$, and
\[ CV_A = CV_{C_0} \leq ... \leq CV_{C_{k+1}} = CV_B. \]
\end{proof}

\begin{theorem}
Let $A$ and $B$ be two nodes that write to the same location in
an acyclic modification order graph $G$ for a terminating execution.
Then $CV_A \leq CV_B$ if and only if $B$ is reachable from $A$ in $G$.
\end{theorem}

\begin{proof}
The backward direction has been proved by Lemma~\ref{thm:lemma-backward}, so
we only need to prove the forward direction. 
Suppose that $CV_A \leq CV_B$. Let's first consider the situation
where the graph $G$ does not contain \code{rmw} edges.

\textbf{Case 1}:
$A$ and $B$ are two stores performed by the same thread with thread id $i$.
Then it is either $A$ happens before $B$ or $B$ happens before $A$. 
If $A$ happens before $B$, then $A$ precedes $B$ in the modification order
because $A$ and $B$ are performed by the same thread.
Hence $B$ is reachable from $A$ in $G$.  We want to show that the other
case is impossible.

If $B$ happens before $A$ and hence precedes $A$ in the modification order,
then $A$ is reachable from $B$. 
By Lemma~\ref{thm:lemma-backward}, $A$ being reachable from $B$ implies that
$CV_B \leq CV_A$.  Since $CV_A \leq CV_B$ by assumption, we deduce that
$CV_A = CV_B$.  This is impossible according to Lemma~\ref{thm:lemma2},
because each store has a unique sequence number and
$U_i(CV_A) = s_A \neq s_B = U_i(CV_B)$, implying that $CV_A \neq CV_B$.

\textbf{Case 2}: $A$ and $B$ are two stores done by different threads.
Suppose that $A$ is performed by thread $i$.
Let $CV_A = (...,s_A,...)$ and $CV_B = (...,t_b,...)$ where both $s_A$
and $t_b$ are in the $i^{\textit{th}}$ position.
By assumption, we have $0 < s_A \leq t_b$.

Since $B$ is not performed by thread $i$, we have $U_i(\perp_{CV_B}) = 0$.
We can apply the same argument similar to the second, third and fourth paragraphs
in the proof of Lemma~\ref{thm:lemma2} and deduce that
$B$ is modification ordered after $A$ or some store sequenced after $A$.
Since modification order is consistent with \textit{sequenced-before} relation,
if follows that $B$ is reachable from $A$ in graph $G$.

Now, consider the case where \code{rmw} edges are present.
Since adding a \code{rmw} edge from a node $S$ to a node $R$
transfers to $R$ all outgoing \textit{mo} edges coming from $S$ and adds a normal
\textit{mo} edge from $S$ to $R$.  Any updates in $CV_S$ are still propagated
to all nodes that are reachable from $S$. By the construction of \code{rmw} edges,
the above argument still applies.
\end{proof}
}

\section{Proof of Equivalence Between Operational and Axiomatic Models\label{sec:equivalence}}
We first present the formalization of our axiomatic model. 
We then show how to lift a trace produced by our operational model to an
axiomatic-style execution and prove that lifting the set of traces
produced by our operational model gives rise to executions that exactly
match the executions allowed by our restricted axiomatic model.
We refer to the axiomatic model that is based on the C++11 memory model
but incorporates the first and the third changes described in
Section~\ref{sec:restricted-model} below as the \textit{modified C++11 memory model}.
We refer to the axiomatic model presented in Section~\ref{sec:restricted-model}
as the \textit{restricted axiomatic model} or \textit{our axiomatic model}.
Our axiomatic model is stronger than the modified C++11 memory model. 

We also introduce some notations in this Section. 
Let $P$ be a program written in our language described in
Figure~\ref{fig:grammar} of the paper.  Let $\Consistent(P)$ denote
the set of executions allowed by the modified C++11 memory model,
$\rConsistent(P)$ denote the set of executions allowed by our axiomatic memory model,
and $\traces(P)$ denote the set of traces produced by our operational model.
We use $\trace{}$ to denote an individual trace, which is a finite sequence of state transitions,
\ie \[ \trace{} = \myState{0} \xrightarrow{\myTrans{1}} \myState{1} \xrightarrow{\myTrans{2}} ... \xrightarrow{\myTrans{m}} \myState{m} \]
The set of axiomatic-style executions obtained by lifting a trace $\trace{}$
is denoted as $\lift(\trace{})$.  Lifting a trace gives rise to
a set of executions because the extension of the $\Mograph$ is not unique,
as explained in Section~\ref{sec:lifting-traces}.
We write $\overline{\lift}(\sigma)$
when we wish to refer to a single execution in $\lift(\trace{})$.

\subsection{Restricted Axiomatic Model \label{sec:restricted-model}}
We present the formalization of our axiomatic model by making following changes
to the formalization of Batty \etal~\cite{batty-memory-model-url,c11popl}:

{\bf 1) Use the C/C++20 release sequence definition:}
This corresponds
to changing the definition of "rs\_element" (in Section 3.6 of their formalization)
by dropping the "same\_thread a rs\_head" term.

{\bf 2) Add $\hb \cup \sco \cup \rf$ is acyclic:} This is implemented by
adding the following to their formalization in Section 3:

\begin{itemize}
 \item Add the following definition:\\
 "let \textit{acyclic\_hb\_sc\_rf actions} $\hb$ $\sco$ $\rf$ =
 irrefl \textit{actions} tc ($\hb \cup \sco \cup \rf$)"

 \item Add the following term to the conjunct in Section 3.11:\\
 "acyclic\_hb\_sc\_rf \textit{Xo.actions $\hb$ Xw.$\sco$ Xw.$\rf$} $\land$"
\end{itemize}

{\bf 3) Strengthen consume atomics to acquire:} This is implemented
with the following two changes in Section 2.1:
\begin{itemize}
 \item Make is\_consume always false.
 \item Change the MO\_CONSUME case for \textit{is\_acquire a}
 to "is\_read \textit{a} $\lor$ is\_fence \textit{a}".
\end{itemize}

Therefore, the set of executions allowed by our restricted axiomatic model
can be expressed as:
\[ \rConsistent(P) = \Consistent(P) \land \hbox{\textit{acyclic}}(\hb \cup \sco \cup \rf). \]
Since we do not consider the consume memory ordering, the $\hb$
relation is the transitive closure of $\sbo$ and $\sw$.
In the simple language described in Figure~\ref{fig:grammar}, an $\sw$
edge is either an \textit{additional synchronizes with} ($\asw$) edge,
an $\rf$ edge, or a combination of $\sbo$ and $\rf$ edges
(release-acquire synchronization involving fences).  Therefore, we have
\[ \sbo \cup \asw \subseteq \hb \subseteq \sbo \cup \asw \cup \rf, \]
and we can deduce that
\[ \sbo \cup \asw \cup \sco \cup \rf \subseteq \hb \cup \sco \cup \rf \subseteq \sbo \cup \asw \cup \sco \cup \rf, \]
which implies that $\hb \cup \sco \cup \rf = \sbo \cup \asw \cup \sco \cup \rf$.
Therefore,
\[ \rConsistent(P) = \Consistent(P) \land \hbox{\textit{acyclic}}(\sbo \cup \asw \cup \sco \cup \rf). \]

\subsection{Lifting Traces} \label{sec:lifting-traces}
We need to extend our operational states with auxiliary labels in order to track events.
We define a label as $\textit{Label} \triangleq \{1,2,3,...\} \cup \{\perp\}$.
We extend $\ThrState$ with a \textit{last sequenced before} ($\lsb$) label and
a \textit{last additional synchronizes with} ($\lasw$) label to track $\sbo$ and $\asw$ relations,
and $\SysState$ with a \textit{last sequentially consistent} ($\lsc$) label to track
$\sco$ relations.  The $\lasw$ label stores the last instruction the parent thread
performs before forking a new thread.  Each load, store, or RMW element
will have an \textit{event} label representing its unique event id. 

The $\code{Load}$, $\code{Store}$, $\code{RMW}$, and $\code{Fence}$ in Figure~\ref{fig:grammar}
correspond to atomic load, store, RMW, and fence events in an execution.  
Loads from and stores to $\code{LocNA}$ correspond to non-atomic reads and writes in an execution. 
The event labels inside \textit{LoadElem}, \textit{StoreElem}, \textit{RMWElem},
and \textit{FenceElem} will match the event ids of 
their corresponding events in the execution.

We will describe how to lift traces in the following.  
The instructions referred to below are the ones that create events in an execution.  

When a thread $T$ performs an instruction and $T.\lsb \neq \perp$,
an $\sbo$ edge is created from $T.\lsb$ to the current instruction.  
Similarly, when a seq\_cst instruction is performed and $\System.\lsc \neq \perp$,
an $\sco$ edge is created from $\System.\lsc$ to the current seq\_cst instruction.
The $\rf$ edges can be created by inspecting traces and checking the $\rf$ fields of
\textit{LoadElem}s and \textit{RMWElem}s.

The $\asw$ edges can be created in two ways: a) When a thread $T$ performs a \code{Fork}
instruction, creating a new thread $T'$, the new thread $T'$ stores $T.\lsb$ in
the field $T'.\lasw$.  Then when $T'$ performs an instruction and $T'.\lasw \neq \perp$,
a $\asw$ edge is created;  b) When thread $T'$ has finished, the parent thread $T$
performs a \code{Join} instruction with the thread id of $T'$.
If $T'.\lsb \neq \perp$, an $\asw$ edge
is created when $T$ performs the next instruction.  

The $\Mograph$ of a trace models $\mo$ relations in an execution.
However, the $\Mograph$ at a memory location may sometimes only capture
a partial order over all stores to the location.  We know that a partial order can
always be extended to a total order.
Therefore, we can perform a topological sort of the $\Mograph$ at each memory location,
and extend the obtained $\mo$ relations to total orders if necessary. 
Then, we can create $\mo$ edges
in the lifted trace based on the extended $\mo$ relations. 
Since linear extensions of a partial order are not unique,
lifting a trace may give rise up multiple axiomatic-style executions. 

Since the operational model requires atomic loads to read from stores that have been
processed by the operational model, and the $\sbo$, $\sco$, and $\asw$ edges are constructed
to be consistent with the program order in the lifting process, we have that 
($\sbo \cup \asw \cup \rf \cup \sco$) is acyclic.  We summarize this relation
as the Lemma below. 

\begin{lemma} \label{lemma:acyclic}
Let $P$ be an arbitrary program and $\sigma \in \traces(P)$. 
Then for any $\Exec{} \in \lift(\sigma)$, the union of
$\sbo$, $\asw$, $\sco$, and $\rf$ edges in $\Exec{}$ is acyclic.
\end{lemma}

\subsection{Equivalence of Axiomatic and Operational Models}
Our goal is to show that for an arbitrary program $P$, the set of executions allowable by our restricted axiomatic model is equivalent to the set of executions we get by lifting traces that our operational model can produce, i.e.,
\[ \forall P \forall \Exec{} . \Exec{} \in \rConsistent(P) \Leftrightarrow \exists \sigma \in \traces(P). \Exec{} \in \lift(\sigma). \]

\begin{definition}
Let $P$ be an arbitrary program and $\Exec{}$ be an axiomatic-style execution of $P$.
We define $\oExec{}$ as the execution that only contains 
$\sbo$, $\asw$, $\sco$, and $\rf$ edges in $\Exec{}$ together with
events in $\Exec{}$.
\end{definition}

Given an execution $\Exec{} \in \rConsistent(P)$ that consists of $n$ events,
$\oExec{}$ is a DAG, which can be topologically sorted to give an ordering,
$\event{1}, ..., \event{n}$, that is consistent with the order
that events are added to $\Exec{}$ as the program is running.

\begin{figure}[!htbp]
  \centering
  \includegraphics[scale=0.5]{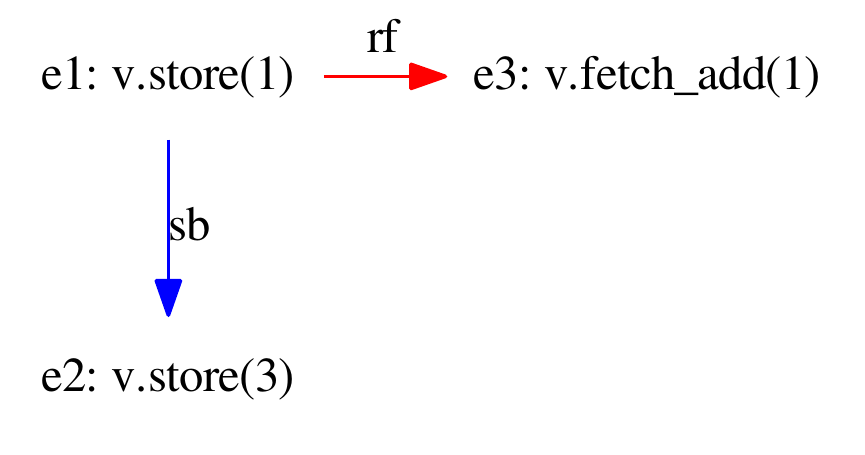}
  \caption{Example where the partial execution graph $\Exec{2}$ misses $\mo$ edges
    \label{fig:missing-mo}
  }
\end{figure}

Based on the topological sort, we define the \textit{partial execution graph}
$\Exec{i}$ of $\Exec{}$ as the execution that consists of the first $i$ events,
together with $\sbo$, $\asw$, $\sco$, $\rf$, and $\mo$ edges
such that the sources and destinations of included relations are events in $\Exec{i}$,
where $0 \leq i \leq n$.  $\Exec{0}$ is defined as the empty execution.

Since we do not consider $\mo$ edges in the topological sort, 
some partial execution graph $\Exec{i}$ may contain events
where there exists modification ordering between them in $\Exec{}$
but the $\mo$ edges are missing in $\Exec{i}$.
For example, in Figure~\ref{fig:missing-mo}, $\{e_1, e_2, e_3\}$ is a valid
topological sort of $\oExec{}$, and we also have
$e_1 \stackrel{\mo}{\rightarrow} e_3 \stackrel{\mo}{\rightarrow} e_2$,
but no $\mo$ edge is present in the partial execution graph $\Exec{2}$.
To deal with this issue, we define the \textit{modification order} at an atomic location $M$
in partial execution graph $\Exec{i}$ as a total order $\partialmo{i}{M}$
over events of $\Exec{i}$ that modifies $M$ such that $\partialmo{i}{M}$
is consistent with the modification order at location $M$ in
the complete execution graph $\Exec{}$.
For atomic modifications $X$ and $Y$
in $\partialmo{i}{M}$, if $X$ precedes $Y$,
then we write $X \stackrel{\partialmo{i}{M}}{\rightarrow} Y$. 

The equivalence proof between axiomatic and operational models contains two directions,
and we will break it down into Lemma~\ref{lemma:forward} and Lemma~\ref{lemma:backward}.
In the proofs of the two lemmas below, for a store event $\event{i}$ in $\Exec{}$,
there is a corresponding node in the $\Mograph$ of the equivalent trace in the
operational model.
Although $\event{i}$ is technically an event in an axiomatic execution,
we sometimes abuse the notation and use $\event{i}$ to refer to the corresponding
node in the $\Mograph$ of the equivalent trace.
If the node corresponding to $\event{i}$ is modification ordered before the node
corresponding to $\event{j}$ in a $\Mograph$, then we may say
$\event{i} \stackrel{\mo}{\rightarrow} \event{j}$ exists in the $\Mograph$.

In the forward direction (Lemma~\ref{lemma:forward}), the proof strategy is to apply induction on the construction of partial execution graphs.
More specifically, if $\Exec{i}$ is a partial execution graph of $\Exec{}$
such that there exists at least one trace $\trace{i}$ where
$\Exec{i} \in \lift(\trace{i})$, 
then when $\Exec{i}$ is extended to $\Exec{i+1}$, we can construct
a trace $\trace{i+1}$ that is an extension of $\trace{i}$ such that
$\Exec{i+1} \in \lift(\trace{i+1})$, 

\begin{lemma} \label{lemma:forward}
Let $P$ be an arbitrary program, and $\Exec{} \in \rConsistent(P)$ be an execution.
Then there exists a trace $\trace{} \in \traces(P)$ such that $\Exec{} \in \lift(\trace{})$.
\end{lemma}

\begin{proof}
Let $P$ be an arbitrary program, and $\Exec{} \in \rConsistent(P)$ be an execution.
Let $\event{1}, ..., \event{n}$ be a topological sort of $\oExec{}$.
We will prove by induction on the construction of partial execution graphs of $\Exec{}$
as described in our proof strategy.

Base case: When $i = 0$, $\Exec{0}$ is the empty execution.
We can take the initial trace $\trace{0}$ that is
the initial state of the program $P$ without any transitions,
and $\Exec{0} \in \lift(\trace{0})$.  At this point the $\Mograph$ is empty as well. 

Inductive step: Suppose that for some $i < n$,
we have constructed the partial execution graph $\Exec{i}$ 
and that there exists some trace $\trace{i}$ such that $\Exec{i} \in \lift(\trace{i})$.
We will assume throughout the proof that the instructions used for the state
transitions match the events being added.  Also note that each visible
instruction adds a context switch after it, so there is always the
ability to change to the required thread.  We will
extend $\Exec{i}$ by adding the next event $\event{i+1}$, and show that we
can construct a $\trace{i+1}$ that is the extension of $\trace{i}$
such that $\Exec{i+1} \in \lift(\trace{i+1})$.
In the following, we will first analyze five different cases of incoming edges
to $\event{j}$, and then show that we can extend the $\Mograph$ of $\trace{i+1}$
to the modification orders in $\Exec{i+1}$. 

\paragraph{A.1}
If $\event{i+1}$ has no incoming edges, then $\event{i+1}$ must be the event $\event{1}$,
corresponding to the first visible instruction in the initial thread.  
Because if $\event{i+1} \neq \event{1}$, there must be an $\sbo$ edge 
or an $\asw$ edge coming to $\event{i+1}$.  We can construct $\trace{i+1} = \trace{1}$ 
by letting the initial thread execute until the instruction corresponding to $\event{1}$ is executed. 

\paragraph{A.2}
If $\event{i+1}$ has an incoming $\sbo$ edge from some event $\event{j}$, 
then from the last state of $\trace{i}$, we must switch to the thread 
that executes $\event{j}$ and continue until the instruction corresponding to $\event{i+1}$ is executed, to obtain $\trace{i+1}$. 

\paragraph{A.3}
If $\event{i+1}$ has an incoming $\asw$ edge, 
then we have two scenarios: $\event{i+1}$ is the first event in a newly created thread; 
or a thread is finishing and joining onto its parent thread,
and $\event{i+1}$ is an event in the parent thread.
In any case,  let $\event{j}$ be the source of the $\asw$ edge,
and $t$ be the thread performing $\event{j}$. 

In the first case, there must be a \code{Fork} instruction after $\event{j}$ and before the next visible instruction in thread $t$.  
The event $\event{j}$ must also be the last event completed in thread $t$, 
because otherwise, the source of the $\asw$ edge would be some event other than $\event{j}$.  
To produce the trace $\trace{i+1}$, we can switch to thread $t$ 
and run the program until the \code{Fork} is done, switch to the newly created thread, 
and run until $\event{i+1}$ is done.

In the second case, the thread $t$ is finishing,
and there is no more visible instruction in thread $t$.  
To produce the trace $\trace{i+1}$, we can switch to thread $t$ and run until $t$ finishes, 
then switch to the parent thread, and run until $\event{i+1}$ is done.  
We must encounter a \code{Join} instruction before $\event{i+1}$ is done, 
because otherwise the destination of the $\asw$ edge would be some event other than $\event{i+1}$. 

\paragraph{A.4}
If $\event{i+1}$ has an incoming $\sco$ edge from some event $\event{j}$, 
then it is similar to the case of $\sbo$.  To obtain $\trace{i+1}$, 
we will switch to the thread that performs $\event{i+1}$ and continue until $\event{i+1}$ is done.  
There shall be no seq\_cst events sequenced before $\event{i+1}$ that has not yet be performed, 
because otherwise, there cannot be an $\sco$ edge from $\event{j}$ to $\event{i+1}$ in $\Exec{i+1}$. 

\paragraph{A.5}
If $\event{i+1}$ has an incoming $\rf$ edge from some event $\event{j}$, 
then we need to show that it is valid for $\event{i+1}$
to read from $\event{j}$ in $\trace{i+1}$.  
Let $\event{i+1}$ be an RMW or atomic read at $M$. 
We claim that $\event{j}$ belongs to the set constructed by
\textsc{BuildMayReadFrom} procedure when the operational model processes the
instruction that creates $\event{i+1}$.
The \code{for} loop in the procedure considers the thread that performs $\event{j}$.
If event $\event{j}$ does not happens before $\event{i+1}$, 
then $\event{j} \in \textit{base}$ at line~\ref{line:may-read-from-base}.  
If $\event{j}$ happens before $\event{i+1}$, then there cannot be any event $\event{k}$ 
that modifies $M$ and that $\event{j} \stackrel{\sbo}{\rightarrow} \event{k} \stackrel{\hb}{\rightarrow} \event{i+1}$,
because in that case, Write-Read Coherence (CoWR) would forbid
$\event{i+1}$ from reading from $\event{j}$ in $\Exec{i+1}$.  
Therefore, $\event{j} \in \textit{base}$ at line~\ref{line:may-read-from-base}.
Now we will show that $\event{j}$ is not removed at line~\ref{line:may-read-from-scrm} 
when $\event{i+1}$ has seq\_cst memory ordering.  
If the last seq\_cst modification of $M$ that precedes $\event{i+1}$ in the total order of $\sco$, 
\ie $S$ at line~\ref{line:may-read-from-lastsc}, does not exist, then we are done.  
Assume such $S$ exists.  According to Section 29.3 statement 3 of the C++11 standard, 
$\event{i+1}$ either reads from $S$ or some non-seq\_cst modification of $M$ 
that does not happen before $S$.
The fact that $\event{i+1}$ reads from $\event{j}$ in $\Exec{i+1}$ 
implies that $\event{j}$ is either $S$ or a non-seq\_cst modification that does not
happen before $S$.  Hence, $\event{j}$ is not removed from $\textit{base}$ 
at line~\ref{line:may-read-from-scrm}.  If $\event{i+1}$ is an RMW, then $\event{j}$ 
has not been read by any other RMW, because no two RMWs 
can read from the same modification in $\Exec{i+1}$.  Hence, $\event{j}$ is in the set 
returned from the \textsc{BuildMayReadFrom} procedure.

We still need to show that having $\event{i+1}$ read from $\event{j}$ does not
create a cycle in the $\Mograph$.
The discussion in paragraph $\textit{B.2}$ below about 
how an atomic load updates the $\Mograph$ shows that having $\event{i+1}$ read from $\event{j}$, 
the updated $\Mograph$ is still consistent with modification orders in $\Exec{i+1}$.
Thus, $\event{i+1}$ reading from $\event{j}$ does not create a cycle in $\Mograph$.
We defer the proof to paragraph \textit{B.2}. 
\\

Now we will show that the $\Mograph$ of $\trace{i+1}$ can be extended to
the modification orders in $\Exec{i+1}$.
The $\Mograph$ is updated when the operational model processes an atomic store,
load, or RMW.  So we will assume that $\event{i+1}$
corresponds to an atomic store, load, or RMW.  
Otherwise, the modification orders in $\Exec{i+1}$ are the same as those in $\Exec{i}$, 
and the $\mo$ edges in $\Mograph$ are not updated, and hence $\Mograph$ of $\trace{i+1}$ 
can be extended to the modification orders in $\Exec{i+1}$ by inductive hypothesis.

\paragraph{B.1}
Suppose $\event{i+1}$ is an atomic store that modifies atomic location $M$.  
We consider two cases: $\event{i+1}$ is the last element in $\partialmo{i+1}{M}$; or
$\event{i+1}$ is not the last element $\partialmo{i+1}{M}$.  

In the first case, $\event{i+1}$ is the last element in the modification order $\partialmo{i+1}{M}$ 
of $\Exec{i+1}$.  Let $\event{j}$ be the second last element in $\partialmo{i+1}{M}$.  
Since $\event{j}$ precedes $\event{i+1}$ in modification order, the modification ordering
is either forced by coherence rules, consistent with $\sco$ relations, 
consistent with Section 29.3 statement 7 of C++11 standard, or not forced by any relations in $\Exec{i+1}$.  
If the modification ordering between $\event{j}$ and $\event{i+1}$
is forced by coherence rules under $\hb$ and $\rf$ relations in $\Exec{i+1}$, 
then the coherence rules being inferred can only be Coherence of Write-Write (CoWW) 
or Coherence of Read-Write (CoRW), because the other two coherence rules require $\rf$ relations 
that are not present in $\Exec{i+1}$.  For CoWW, line~\ref{line:wpriorset-hb} 
in the \textsc{WritePriorSet} procedure considers the atomic store $X$
corresponding to $\event{j}$ and adds it to \textit{priorset}.
We claim that the store $X$ will not be filtered out
by the \textit{last} function call in line~\ref{line:wpriorset-last},
because otherwise there would be an atomic store event $\event{k}$
(different from $\event{i+1}$)
sequenced after $\event{j}$, contradicting the assumption that
$\event{j}$ is the second last element in $\partialmo{i+1}{M}$.
The case for CoRW is similar. 
If the modification ordering between $\event{j}$ and $\event{i+1}$
is consistent with $\sco$ relations, 
then $\event{j}$ and $\event{i+1}$ are both seq\_cst atomic stores 
and line~\ref{line:wpriorset-sc} in the \textsc{WritePriorSet} procedure considers such case.  
If the modification ordering between $\event{j}$ and $\event{i+1}$
is consistent with Section 29.3 statement 7 of C++11 standard, 
then this case is dealt with at line~\ref{line:wpriorset-fence}
in the \textsc{WritePriorSet} procedure.
If the modification ordering between $\event{j}$ and $\event{i+1}$
is not forced by any relations in $\Exec{i+1}$, 
then the $\Mograph$ does not contain a $\mo$ edge from $\event{j}$ to $\event{i+1}$,
and we are free to extend $\Mograph$ to include $\partialmo{i+1}{M}$.
In fact, the algorithm \textsc{WritePriorSet} may add $\mo$ edges from events 
modification ordered before $\event{j}$ to $\event{i+1}$.  
This is not a problem, because adding such edges does not introduce modification ordering 
not present in the modification order of $\Exec{i+1}$. 

In the second case, let $\event{j}$ be the event immediately preceding $\event{i+1}$ 
in $\partialmo{i+1}{M}$, and $\event{k}$ be the event immediately succeeding $\event{i+1}$ 
in $\partialmo{i+1}{M}$.  Without loss of generality, assume $\event{k}$
is the last event in $\partialmo{i+1}{M}$.
The modification ordering between $\event{j}$ and $\event{i+1}$ 
could be any one of the cases discussed in the first case in this paragraph.
So we do not repeat the same argument.  However, since all other events in $\Exec{i+1}$ 
including $\event{k}$ come before $\event{i+1}$ in the topological order, 
there is no chain of $\rf$, $\sbo$, $\asw$, and $\sco$ edges that come from $\event{i+1}$ 
to any other event in $\Exec{i+1}$.  Thus, the only possibility is that no relations in $\Exec{i+1}$ 
force the existence of the modification ordering between $\event{i+1}$ and $\event{k}$.  
Hence, the $\mo$ edge between $\event{i+1}$ and $\event{k}$ does not exist
in the $\Mograph$ of $\trace{i+1}$, and
we are free to extend $\Mograph$ to include $\partialmo{i+1}{M}$.
If $\event{k}$ is not the last event in $\partialmo{i+1}{M}$, then the same argument
applies to any event modification ordered after $\event{k}$ in $\partialmo{i+1}{M}$.

\paragraph{B.2}
Suppose that $\event{i+1}$ is an atomic load that reads from event $\event{j}$ at atomic location $M$.  
Adding event $\event{i+1}$ does not change the modification order at $M$
from $\Exec{i}$ to $\Exec{i+1}$, but performing the instruction corresponding
to $\event{i+1}$ may change the $\Mograph$ from $\trace{i}$ to $\trace{i+1}$.  
We will show that the $\Mograph$ in $\trace{i+1}$ can still be extended 
to the modification orders in $\Exec{i+1}$.
Line~\ref{line:rpriorset-s1},~\ref{line:rpriorset-s2},and~\ref{line:rpriorset-s3}
in \textsc{ReadPriorSet} procedure consider statements 5, 4, and 6 in Section 29.3 of the standard.  
For each thread, if such $S_1$, $S_2$, and $S_3$ exist, the events 
corresponding to them are all modification ordered before $\event{j}$ in $\Exec{i+1}$.  
Therefore, having $\mo$ edges from $S_1$, $S_2$, and $S_3$ to $\event{j}$ 
in $\Mograph$ does not conflict with the modification orders in $\Exec{i+1}$.  
Line~\ref{line:rpriorset-s4} in the \textsc{ReadPriorSet} procedure considers 
Write-Read Coherence (CoWR) and Read-Read Coherence (CoRR).  
If an $\mo$ edge from some event $\event{k}$ to $\event{j}$ in $\Mograph$
is induced by CoWR and CoRR, then $\event{k}$ must be modification
ordered before $\event{j}$ in $\Exec{i+1}$ by the standard. 
Hence, when extending $\trace{i}$ to $\trace{i+1}$, the newly created $\mo$ edges
in $\Mograph$ do not conflict with modification orders in $\Exec{i+1}$.
Because the $\Mograph$ in $\trace{i}$ can be extended to the modification orders
of $\Exec{i}$ by inductive hypothesis, so can the $\Mograph$ in $\trace{i}$
be extended to the modification orders of $\Exec{i+1}$.

\paragraph{B.3}
Suppose that $\event{i+1}$ is an atomic RMW that reads from $\event{j}$.  
An RMW modifies the $\Mograph$ in three phases: performing an atomic load, 
migrating $\mo$ edges using the AddRMWEdge procedure, and performing an atomic store.  
We also consider two cases.

In the first case, $\event{i+1}$ is the last element 
in $\partialmo{i+1}{M}$ of $\Exec{i+1}$.
The first and third phases have been discussed in paragraphs \textit{B.1} and \textit{B.2}.
In the second phase, since $\event{i+1}$ is the last element in $\partialmo{i+1}{M}$, 
no edges in $\Mograph$ are migrated.
Then an $\mo$ edge is added from $\event{j}$ to $\event{i+1}$ in $\Mograph$, 
which does not conflict with the modification order $\partialmo{i+1}{M}$, because
an atomic RMW is immediately modification ordered after the modification it reads from.  

In the second case, $\event{i+1}$ is not the last element in $\partialmo{i+1}{M}$.  
Since $\event{i+1}$ reads from $\event{j}$, $\event{j}$ must immediately 
precede $\event{i+1}$ in $\partialmo{i+1}{M}$.  The first phase of the RMW is equivalent to 
an atomic load.  In the second phase, any outgoing $\mo$ edges from $\event{j}$
will be migrated to outgoing $\mo$ edges from $\event{i+1}$, 
and an $\mo$ edge is added from $\event{j}$ to $\event{i+1}$ in $\Mograph$.  
The third phase is the same as the second case of an atomic store, 
except that for an event $\event{k}$ modification ordered after $\event{i+1}$ in $\partialmo{i+1}{M}$, 
there may exist $\mo$ edges from $\event{i+1}$ to $\event{k}$ in $\Mograph$ 
due to edge migrations in the second phase. 

In both cases, the $\Mograph$ of $\trace{i+1}$ does not contain $\mo$ edges
that conflict with modification orders in $\Exec{i+1}$
Hence, the $\Mograph$ of $\trace{i+1}$ can be extended to include modification
orders in $\Exec{i+1}$.

Considering all above cases in paragraphs \textit{B.1}, \textit{B.2},
and \textit{B.3}, the $\Mograph$ of $\trace{i+1}$ can be extended to
the modification orders in $\Exec{i+1}$, and the proof completes.
\end{proof}

\begin{lemma} \label{lemma:backward}
Let $P$ be an arbitrary program, and $\trace{} \in \traces(P)$ be a trace.
Then for all $\Exec{} \in \lift(\trace{})$, we have $\Exec{} \in \rConsistent(P)$.
\end{lemma}

\begin{proof}
In the backward direction, we want to show that given a program $P$ and
a trace $\trace{}$ produced by the operational model, then any execution
$\Exec{} \in \lift(\trace{})$  obtained by lifting the trace $\trace{}$
is an element of $\rConsistent(P)$.
We will prove by induction on the construction of the partial trace $\trace{i}$.
Specially, if we have $\Exec{k} \in \lift(\trace{i})$,
where $\Exec{k}$ is a partial execution graph of $\Exec{}$
based on a topological sort of $\oExec{}$,
then when $\trace{i}$ is extended to $\trace{i+1}$, we have $\Exec{k} \in \lift(\trace{i+1})$
or $\Exec{k+1} \in \lift(\trace{i+1})$, where $\Exec{k+1}$ is also a partial execution
graph of $\Exec{}$ subject to the same topological sort.

Let $P$ be an arbitrary program and $\trace{} \in \traces(P)$ be a trace produced by our operational model.
Let $\Exec{} \in \lift(\trace{})$ be an execution obtained by lifting the trace $\trace{}$. 
In Section~\ref{sec:lifting-traces}, we have argued that when only considering
$\sbo$, $\asw$, $\rf$ and $\sco$ edges, any execution obtained by lifting a trace is acyclic.  In the process of lifting traces, some transitions create events while
some do not.  Label the events in $\Exec{}$ in the order that they are created by 
transitions in $\trace{}$ is a natural topological sort of $\oExec{}$.
\emph{All the partial execution graphs described below are based on
this natural topological sort.  We also have the $\Mograph$ of all partial traces
$\trace{i}$ be extended in a way that is consistent with $\Exec{}$ in the lifting process.}
We will use $\olift(\trace{i})$ to refer to the specific execution in $\lift(\trace{i})$
whose $\Mograph$ extension is consistent with that of $\Exec{}$.

Base case: when $i = 0$, $\trace{0}$ is the empty trace, which is the initial state
of the operation model for $P$.  Thus, $\olift(\trace{0})$ is
the empty execution graph $\Exec{0}$, which is a valid partial execution graph of $E$.

Inductive step: suppose we have constructed a partial trace $\trace{i}$ of $\trace{}$
and that $\olift(\trace{i}) = \Exec{k}$,
where $\Exec{k}$ is a partial execution graph of $\Exec{}$. 
We will show that when $\trace{i}$ is extended to $\trace{i+1}$ by executing the next
transition $\myTrans{i+1}$, we have either $\olift(\trace{i+1}) = \Exec{k}$ or 
$\olift(\trace{i+1}) = \Exec{k+1}$, where $\Exec{k+1}$ is a partial execution graph 
of $\Exec{}$. 

We will consider different cases for the transition $\myTrans{i+1}$ below. 

\paragraph{Invisible Instruction}
If the transition $\myTrans{i+1}$ is an invisible instruction,
such as an \code{if} statements, an assignment to non-atomic locations,
and the empty statement $\epsilon$, then it leaves $\Exec{k}$ unchanged,
and $\olift(\trace{i+1}) = \Exec{k}$.  For \code{if} statements, the partial
trace at this point determines a branch to take, and proving that taking
the branch will produce a valid partial execution graph when lifted
comes down to proving the rest of cases analyzed.

\paragraph{Visible Instruction}
If the next transition $\myTrans{i+1}$ creates a new event $\event{k+1}$
in the lifting process, we have $\olift(\trace{i+1}) = \Exec{k+1}$.
Moreover, new $\sbo$, $\asw$, $\sco$ edges, and updates in the modification
orders may also be added to $\Exec{k}$ to form $\Exec{k+1}$.
We already show that $\Exec{k+1}$ has acyclic $\sbo \cup \asw \cup \rf \cup \sco$ edges.
So we only need to show that $\Exec{k+1}$ is a partial execution graph
of $\Exec{}$.

We will discuss the newly added $\sbo$ and $\asw$ edges first.
Suppose the transition $\myTrans{i+1}$ is a general visible instruction
that corresponds to the event $\event{k+1}$.
We will consider new $\sbo$ and $\asw$ edges that may be added to $\Exec{k}$
when lifting $\trace{i+1}$.
Suppose that $\myTrans{i+1}$ is performed by thread $t$
and is not the first visible instruction in thread $t$.
Then an $\sbo$ edge will be drawn from the last visible event performed by $t$
to $\event{k+1}$.  Suppose that $\myTrans{i+1}$ is the first visible instruction
performed by thread $t$.  If $t$ is the main thread, then no new edges to
$\event{k+1}$ will be added during lifting.  If $t$ is not the main thread, 
then the parent thread $t'$ that created $t$ must have performed a \code{Fork} instruction,
and an $\asw$ edge from the last visible instruction sequenced
before the \code{Fork} instruction to $\event{k+1}$ will be added to $\Exec{k}$, to obtain
$\Exec{k+1}$. 
It is clear that the way we construct $\sbo$ and $\asw$ edges in lifting traces
is consistent with the C++ axiomatic model.

\paragraph{Visible Instruction (Atomic Store)}
If the next transition $\myTrans{i+1}$ is an atomic store statement at location $M$,
it will create an $\textit{StoreElem}$ that corresponds to the event $\event{k+1}$.
We will focus on the changes to modification orders and $\sco$ relations.
If $\event{k+1}$ is a seq\_cst store, lifting $\trace{i+1}$ will cause an $\sco$
edge to be added from the last seq\_cst event to $\event{k+1}$. 
Line~\ref{line:wpriorset-sc} in the \textsc{WritePriorSet} procedure ensures that
$\mo$ edges in the $\Mograph$ conform with $\sco$ relations by adding an $\mo$ edge
from the last seq\_cst store at $M$ to $\event{k+1}$ in the $\Mograph$
if the last seq\_cst store at $M$ exists.
So modification orders confirm with $\sco$ relations in $\Exec{k+1}$. 
Since the operation model may only add incoming $\mo$ edges to $\event{k+1}$
in the $\Mograph$ when processing $\myTrans{i+1}$, then modification orders
in $\Exec{k+1}$ do not have cycles.
By Lemma~\ref{lemma:acyclic}, $\sbo \cup \asw \cup \rf \cup \sco$ in $\Exec{}$,
it follows that $\sco$ relations conform with $\hb$ relations
in $\Exec{k+1} = \olift(\trace{i+1})$.
Line~\ref{line:wpriorset-fence} and line~\ref{line:wpriorset-hb} in the 
\textsc{WritePriorSet} procedure ensure that $\mo$ edges induced by seq\_cst
fences, CoRW and CoWW are added to the $\Mograph$.
Therefore, the new modification orders in $\Exec{k+1}$ induced by
the changes in the $\Mograph$ in the lifting process conform with CoRW, CoWW,
Section 29.3 statement 7 of the C++11 standard, and $\sco$ relations. 
Conformity with CoRW and CoWW implies that $\mo$ conforms with $\hb$.

Now, if $\event{k+1}$ is the last element in $\partialmo{k+1}{M}$, then the above
discussion shows that $\Exec{k+1}$ is a valid partial execution graph of $\Exec{}$
based on the natural topological sort.
It is also possible that $\event{k+1}$ is not the last element in $\partialmo{k+1}{M}$.
Let $\event{j}$ be any event modification ordered after $\event{k+1}$ in
$\partialmo{k+1}{M}$. Note that $\event{j}$ is topologically ordered before $\event{k+1}$.
Since the \textsc{WritePriorSet} procedure only adds incoming $\mo$ edges 
to $\event{k+1}$, no $\mo$ or chain of $\mo$ edges from $\event{k+1}$ to $\event{j}$
exists in $\Mograph$.  Since the operational model forbids cycles in $\Mograph$,
the final $\Mograph$ of $\trace{}$ is free of cycles, 
and the modification orders in the final $\Mograph$
at each location are extended to a total order in $\Exec{}$ during lifting.
Because we assume that $\Mograph$ of all partial traces $\trace{i}$ be extended in a way
that is consistent with $\Exec{}$ in the lifting process, we can conclude that
the modification ordering between $\event{k+1}$ and $\event{j}$ do not cause
any cycles in modification orders of $\Exec{k+1}$ and is only added to make $\partialmo{k+1}{M}$
a total order.  Therefore, $\Exec{k+1}$ is a valid partial execution graph of $\Exec{}$.

\paragraph{Visible Instruction (Atomic Load)}
If the transition $\myTrans{i+1}$ is an atomic load statement at location $M$,
it creates an $\textit{LoadElem}$ that corresponds to the event $\event{k+1}$.
To obtain $\Exec{k+1} = \olift(\trace{i+1})$, a new $\rf$ edge is added to $\Exec{k}$, and the
modification orders at $M$ may be updated.  Suppose that $\event{k+1}$
reads from $\event{j}$, where $\event{j}$ is topologically ordered before $\event{k+1}$.
We make the following claim:

\textbf{(Claim 1)} Any valid store that the operational model allows the
\textit{LoadElem} corresponding to $\event{k+1}$
to read from is also valid for $\event{k+1}$ to read from under our axiomatic model.

We will prove this claim by contradiction or contrapositive using case analysis.
We will assume our axiomatic model forbids $\event{k+1}$ from reading from $\event{j}$. 

\textit{Case 1:} If having $\event{k+1}$ read from $\event{j}$ violates CoWR,
then there exists an event $\event{l}$ in $\Exec{k+1}$
such that $\event{j} \stackrel{\partialmo{k+1}{M}}{\rightarrow} \event{l}$ and
$\event{l} \stackrel{\hb}{\rightarrow} \event{k+1}$. 
We will first assume that $\event{j} \stackrel{\mo}{\rightarrow} \event{l}$
exists in the $\Mograph$ of $\trace{i}$.
The \textsc{ReadPriorSet} procedure iterates over each thread, and when considering
the thread that performs $\event{l}$, line~\ref{line:rpriorset-s4} finds either 
the store $\event{l}$, any store sequenced after $\event{l}$,
or any load sequenced after $\event{l}$.
Then the store $A$ in line~\ref{line:rpriorset-last} of \textsc{ReadPriorSet}
is the store $\event{l}$,
a store sequenced after $\event{l}$, or a store read by a load sequenced after $\event{l}$.
In any case, based on CoWW, CoWR and the inductive hypothesis
that $\Exec{k}$ is a valid partial execution graph of $\Exec{}$, we can deduce that
the $\mo$ edge (or the equivalent chain of $\mo$ edges) 
$\event{j} \stackrel{\mo}{\rightarrow} A$ exists in the $\Mograph$ of $\trace{i}$.
Since $A$ is reachable from $\event{j}$, line~\ref{line:rpriorset-reject}
in the \textsc{ReadPriorSet} procedure forbids the \textit{LoadElem}
corresponding to $\event{k+1}$ from reading from the store corresponding to
$\event{j}$ in the operational model.  Then we prove the Claim 1 by contrapositive. 

However, it is also possible that $\event{j}$ and $\event{l}$ are unordered
in the $\Mograph$ of $\trace{i}$.  Then the modification ordering between $\event{j}$
and $\event{l}$ in $\Exec{k+1}$ is due to the extension of the final
$\Mograph$ of $\trace{}$ in $\Exec{}$, as the $\Mograph$ of
all partial traces $\trace{i}$ are extended in a way that is consistent with $\Exec{}$. 
Having $\event{k+1}$ read from $\event{j}$ will add $\mo$ edges so that
$\event{l} \stackrel{\mo}{\rightarrow} \event{j}$ exists in the $\Mograph$ of $\trace{i+1}$
and the final $\Mograph$ of $\trace{}$. 
This is a contradiction because the extension of the final $\Mograph$ is only
required between two unordered stores in $\Mograph$.  Therefore, we prove
the Claim 1 by contradiction. 

\textit{Case 2:} If having $\event{k+1}$ read from $\event{j}$ violates CoRR, the proof is similar
to the case of CoWR. 

\textit{Case 3:} If having $\event{k+1}$ read from $\event{j}$ violates Section 29.3
statement 4 in the C++11 standard,
then in $\Exec{k+1}$, there exists a fence $X$ sequenced before $\event{k+1}$.
Let $X'$ be the last seq\_cst store preceding $X$ in $\sco$ in the partial execution
graph $\Exec{k+1}$.
Then $\event{j}$ is either some store modification ordered before $X'$
or is a seq\_cst store that precedes $X'$ in $\sco$.
Consider that $\event{j} \stackrel{\sco}{\rightarrow} X'$ in $\Exec{k+1}$.
Then the inductive hypothesis implies that $\event{j} \stackrel{\mo}{\rightarrow} X'$
in the $\Mograph$ of $\trace{i}$, as line~\ref{line:wpriorset-sc} in the \textsc{WritePriorSet}
procedure adds $\mo$ edges between seq\_cst stores at the same location in $\Mograph$.
When the \textsc{ReadPriorSet} procedure iterates over the thread that performs $X'$,
line~\ref{line:rpriorset-s2} returns $S_2$ as $X'$,
and line~\ref{line:rpriorset-last} will return $A$ as either $X'$,
a store sequenced after $X'$, or a store read from by a load sequenced after $X'$. 
In any case, CoWW, CoWR, and the inductive hypothesis guarantees that
$X' \stackrel{\mo}{\rightarrow} A$ in the $\Mograph$ of $\trace{i}$.
Therefore, we have $\event{j} \stackrel{\mo}{\rightarrow} A$
in the $\Mograph$ of $\trace{i}$, and line~\ref{line:rpriorset-reject}
in the \textsc{ReadPriorSet} procedure forbids $\event{k+1}$ from reading
from $\event{j}$ in the operational model.

Consider that $\event{j}$ is modification ordered before $X'$ in $\Exec{k+1}$. 
If $\event{j} \stackrel{\mo}{\rightarrow} X'$ exists in the $\Mograph$ of $\trace{i}$,
then it the same as the last paragraph.  
If $\event{j}$ and $X'$ are two unordered stores in the $\Mograph$
of $\trace{i}$, then we can deduce the same contradiction as in the analysis of
CoWR violation in \textit{Case 1}.

\textit{Case 4:} If having $\event{k+1}$ read from $\event{j}$ violates statement 5 or statement 6
of Section 29.3 in the C++11 standard, then the proof is similar to the analysis
of the violation of statement 4 in \textit{Case 3} by considering line~\ref{line:rpriorset-s1}
or line~\ref{line:rpriorset-s3} in the \textsc{ReadPriorSet} procedure. 
So we do not present it here. 

\textit{Case 5:} If having $\event{k+1}$ read from $\event{j}$ violates Section 29.3 statement 3
in the C++11 standard, then $\event{k+1}$ has seq\_cst ordering, the last seq\_cst $X$
store at $M$ that precedes $\event{k+1}$ in $\sco$ exists, and $\event{j}$
is either a seq\_cst store that precedes $X$ is $\sco$ or a store that happens before
$X$.  Then $\event{j}$ will be removed from the $\textit{may-read-from}$ set
in line~\ref{line:may-read-from-scrm} of the \textsc{BuildMayReadFrom} procedure, 
and $\event{j}$ is not a valid store for $\event{k+1}$ to read from in the operational
model. 

We have completed the proof of Claim 1 by analyzing the above five cases. 
Claim 1 shows that $\Exec{k+1}$ has valid $\rf$ edges.
We will next show the acyclicity for modification orders and conformity for $\sco$ edges
in $\Exec{k+1}$.  Suppose $\event{k+1}$ reads from some event $\event{j}$.
Establishing this $\rf$ relation adds incoming $\mo$ edges to $\event{j}$
to the $\Mograph$ of $\trace{i}$.  We claim that the updated $\Mograph$ is free of cycles.
If adding these edges causes a cycle in the $\Mograph$ of $\trace{i+1}$,
then the cycle contains only one of the newly added edges,
and line~\ref{line:rpriorset-reject} in the \textsc{ReadPriorSet} procedure
should have forbidden $\event{k+1}$ from reading from $\event{j}$.  Therefore, 
the $\Mograph$ of $\trace{i+1}$ is free of cycles. 
Lines~\ref{line:rpriorset-s1} to~\ref{line:rpriorset-s4} in the \textsc{ReadPriorSet}
procedure considers $\mo$ edges
that are enforced by statements 5, 4, and 6 of Section 29.3 in the C++11 standard,
CoRR, and CoWR.  Line~\ref{line:rpriorset-last} filters out redundant $\mo$ edges,
as once the $\mo$ edges that are not filtered out are added, then the $\mo$ edges
that are filtered out will follow from the transitivity of $\mo$ edges.
Since the $\Mograph$ of $\trace{i+1}$ is acyclic, modification orders in $\Exec{k+1}$
is also acyclic in the lifting process.  
If $\event{k+1}$ has seq\_cst ordering, then an $\sco$ edge will be drawn from the last
seq\_cst event to $\event{k+1}$ in $\Exec{k+1}$, this $\sco$ edge conforms with 
modification orders as $\event{k+1}$ is a load and not an element in $\partialmo{k+1}{M}$. 
However, we also need to show that modification orders in $\Exec{k+1}$
conform with other $\sco$ edges. 
Because we assume conformity of modification orders with $\sco$ edges in $\Exec{k}$, if
any cycle exists in union of $\sco$ and $\partialmo{k+1}{M}$ in $\Exec{k+1}$,
the cycle must involve one of the newly added modification ordering.
Suppose a cycle $C$ exists and it contains events $X$
and $\event{j}$ where $X \stackrel{\partialmo{k+1}{M}}{\rightarrow} \event{j}$ 
is a newly added modification ordering in $\Exec{k+1}$.
Since all events in $\partialmo{k+1}{M}$ are atomic stores, the two endpoints of
any maximal chain of $\sco$ edges (could also be an $\sco$ edge) in $C$ must be
seq\_cst atomic stores.  Let the two endpoints be events $Z_1$ and $Z_2$.
Then, the relation $Z_1 \stackrel{\mo}{\rightarrow} Z_2$ or
$Z_2 \stackrel{\mo}{\rightarrow} Z_1$ whichever conforms with the $\sco$ edges
must exist in the $\Mograph$ of $\trace{i}$.
No $\stackrel{\partialmo{k+1}{M}}{\rightarrow}$ edges in $C$ can be due to the extensions
of unordered stores in $\Mograph$ of $\trace{i+1}$,
because we assume that the $\Mograph$ of $\trace{}$ is extended without cycles.
Therefore, we must have $\event{j} \stackrel{\mo}{\rightarrow} X$
in the $\Mograph$ of $\trace{i}$.
Then, the relation $X \stackrel{\mo}{\rightarrow} \event{j}$ cannot exist in
the $\Mograph$ of $\trace{i+1}$, as it makes the $\Mograph$ of $\trace{i+1}$ acyclic.
However, $X \stackrel{\partialmo{k+1}{M}}{\rightarrow} \event{j}$ cannot be due to
the extension of unordered stores in the $\Mograph$ of $\trace{i+1}$.
Therefore, we have a contradiction, and $\sco$ conforms with modification orders
in $\Exec{k+1}$.  Therefore, we have $\Exec{k+1}$ is a valid partial execution
graph of $\Exec{}$.

\paragraph{Visible Instruction (Atomic RMW)}
If the transition $\myTrans{i+1}$ is an atomic RMW statement at location $M$,
it creates an $\textit{RMWElem}$ that corresponds to the event $\event{k+1}$.
An atomic RMW is both a load and a store except that the standard requires that
RMW operations shall always read the last value (in the modification order) written
before the write associated with the RMW operation.
In the operational model, the \textsc{BuildMayReadFrom} procedure forbids two
RMWs to read from the same store.
Denote the store that $\event{k+1}$ reads from as $\event{j}$. 
Then the \textsc{AddRMWEdge} procedure adds all outgoing $\mo$ edges from $\event{j}$
to the set of outgoing edges of $\event{k+1}$ and adds an $\mo$ edge from $\event{j}$
to $\event{k+1}$ to form the $\Mograph$ of $\trace{i+1}$.
Therefore, when lifting $\trace{i+1}$, $\event{j}$
is immediately modification ordered before $\event{k+1}$ in $\Exec{k+1}$. 
Hence, $\Exec{k+1}$ is a valid partial execution graph of $\Exec{}$. 

If the transition $\myTrans{i+1}$ is a fence instruction, it creates a
\textit{FenceElem} which corresponds to the event $\event{k+1}$.  
Lifting $\trace{i+1}$ adds $\event{k+1}$ to $\Exec{k}$ but does not
create $\rf$ edges or change modification ordering.  If $\event{k+1}$ has
seq\_cst ordering, an $\sco$ edge from the last seq\_cst event (if exists)
to $\event{k+1}$ will be created, and $\sco$ conforms with $\hb$ and modification
orders in $\Exec{k+1}$.  Thus, $\Exec{k+1}$ is a valid partial execution graph of $\Exec{}$.

We have completed the proof of induction by case analysis. 
\end{proof}

\section{Additional Data}
\begin{table*}[h]
\centering

\caption{Performance results (in ms) for individual JavaScript benchmarks in the JSBench suite for tsan11, tsan11rec, and \tool under two configurations.  Smaller times are better.  The "Memory Accesses" columns report the number of atomic operations (including synchronization operations such as mutex and condition variable operations) and normal accesses to shared memory locations executed by individual JavaScript benchmarks under \tool. \label{table:result-additional-1}}
{\small
\begin{tabular}{|l|rr|rr|rr|r|r|}
  \hline
  & \multicolumn{6}{c|}{\textbf{Performance Results (in ms)}} & \multicolumn{2}{c|}{\textbf{\# Memory Accesses}} \\
  & \multicolumn{6}{c|}{} & \multicolumn{2}{c|}{\textbf{By \tool}} \\
  \hline
  & \multicolumn{2}{c|}{\textbf{tsan11}} & \multicolumn{2}{c|}{\textbf{tsan11rec}} & \multicolumn{2}{c|}{\textbf{\tool}} &&  \\
  JavaScript benchmark & 1 core & all cores & 1 core & all cores & 1 core & all cores & \# non-atomics & \# atomics\\
  \hline
  amazon/chrome       & 368     &       348     &       1054    &       383     &	767	&	766	&	92M	&	99M	\\
  amazon/chrome-win   & 368     &       349     &       1053    &       385     &	767	&	767	&	90M	&	99M	\\
  amazon/firefox      & 362     &       341     &       1085    &       368     &	718	&	718	&	71M	&	73M	\\
  amazon/firefox-win  & 351     &       331     &       1038    &       357     &	700	&	701	&	67M	&	71M	\\
  amazon/safari       & 400     &       372     &       1170    &       414     &	778	&	779	&	77M	&	86M	\\
  facebook/chrome     & 2323    &       2017    &       7693    &       2669    &	4421	&	4434	&	471M	&	449M	\\
  facebook/chrome-win & 3715    &       3135    &       11463   &       3772    &	7084	&	7092	&	607M	&	850M	\\
  facebook/firefox    & 1458    &       1303    &       5219    &       1835    &	3002	&	3015	&	278M	&	318M	\\
  facebook/firefox-win& 818     &       721     &       2858    &       944     &	1619	&	1628	&	149M	&	189M	\\
  facebook/safari     & 3639    &       3034    &       14009   &       4508    &	7360	&	7375	&	577M	&	1,036M	\\
  google/chrome       & 1616    &       1488    &       5619    &       1901    &	3358	&	3365	&	333M	&	399M	\\
  google/chrome-win   & 1579    &       1447    &       5489    &       1866    &	3260	&	3249	&	400M	&	359M	\\
  google/firefox      & 962     &       899     &       2544    &       1069    &	1984	&	1986	&	171M	&	173M	\\
  google/firefox-win  & 1123    &       1033    &       3383    &       1265    &	2279	&	2283	&	214M	&	208M	\\
  google/safari       & 1445    &       1320    &       4829    &       1670    &	2943	&	2941	&	273M	&	319M	\\
  twitter/chrome      & 618     &       588     &       1078    &       645     &	1305	&	1306	&	154M	&	141M	\\
  twitter/chrome-win  & 620     &       587     &       1073    &       644     &	1310	&	1310	&	158M	&	151M	\\
  twitter/firefox     & 175     &       165     &       275     &       178     &	376	&	376	&	50M	&	47M	\\
  twitter/firefox-win & 174     &       164     &       277     &       177     &	373	&	374	&	50M	&	47M	\\
  twitter/safari      & 466     &       443     &       880     &       490     &	956	&	963	&	103M	&	106M	\\
  yahoo/chrome        & 1638    &       1358    &       6191    &       2000    &	4294	&	4308	&	339M	&	663M	\\
  yahoo/chrome-win    & 1345    &       1115    &       4953    &       1633    &	3172	&	3176	&	234M	&	496M	\\
  yahoo/firefox       & 1638    &       1363    &       6250    &       2016    &	4292	&	4283	&	338M	&	663M	\\
  yahoo/firefox-win   & 883     &       747     &       3460    &       1100    &	1895	&	1888	&	113M	&	320M	\\
  yahoo/safari        & 1635    &       1370    &       6263    &       2019    &	4292	&	4291	&	338M	&	661M	\\
  \hline
\end{tabular}
}
\end{table*}

Table~\ref{table:result-additional-1} reports some detailed statistics about
the 25 JavaScript benchmarks in the JSBench suite.

\end{document}